\documentclass[reqno, a4paper, 11pt]{amsart}
\usepackage{amsmath, amssymb, amsthm, enumitem, mathrsfs}
\usepackage{placeins}

\usepackage{xcolor}

\usepackage{hyperref}
\usepackage{txfonts}

\usepackage{newtxtext}

\usepackage{tikz}
\usetikzlibrary{external}
\usetikzlibrary{3d}
\usepackage{tikz-3dplot}
\usetikzlibrary{arrows.meta,backgrounds,patterns,decorations.pathreplacing,decorations.pathmorphing}
\usepackage{caption}
\usepackage{accents}

\newtheorem{theorem}{Theorem}

\newtheorem{corollary}[theorem]{Corollary}
\newtheorem{definition}[theorem]{Definition}
\newtheorem{lemma}[theorem]{Lemma}
\newtheorem{proposition}[theorem]{Proposition}

{Important Convention}

\theoremstyle{remark}
\newtheorem{remark}[theorem]{Remark}
\newtheorem{example}[theorem]{Example}

 \newcommand{\eps}{\varepsilon}
 \newcommand{\h}{\mathcal{H}}
 \newcommand{\M}{\mathcal{M}}
  \newcommand{\F}{\mathcal{F}}
 
 \newcommand{\sgn}{\operatorname{sgn}}
 \renewcommand{\phi}{\varphi}

\newcommand{\E}{\mathbb{E}}
\renewcommand{\P}{\mathbb{P}}
\newcommand{\N}{\mathbb{N}}
\newcommand{\Q}{\mathbb{Q}}
\newcommand{\R}{\mathbb{R}}

\usepackage{bbm}

\renewcommand{\eps}{\varepsilon}
\renewcommand{\epsilon}{\varepsilon}

\DeclareMathOperator{\id}{Id}

\DeclareMathOperator{\esssup}{ess\, sup}

\numberwithin{equation}{section}
\numberwithin{theorem}{section}

\renewcommand{\subset}{\subseteq}
\renewcommand{\supset}{\supseteq}
\usepackage{graphicx}
\usepackage{subcaption}

\makeatletter
\newcommand{\mylabel}[2]{#2\def\@currentlabel{#2}\label{#1}}
\makeatother

\newcommand{\AW}{\mathcal{AW}}
\newcommand{\NA}{\textrm{NA}}

\usepackage{relsize}
\def\fcmp{\mathbin{\raise 0.6ex\hbox{\oalign{\hfil$\scriptscriptstyle \mathrm{o}$\hfil\cr\hfil$\scriptscriptstyle\mathrm{9}$\hfil}}}}

\begin{document}

\author{Acciaio~Beatrice}
\author{Backhoff-Veraguas~Julio}
\author{Pammer~Gudmund}
\thanks{ETH Zurich, Switzerland,
\href{mailto:beatrice.acciaio@math.ethz.ch}{beatrice.acciaio@math.ethz.ch},
University of Vienna, Austria,
\href{mailto:julio.backhoff@univie.ac.at}{julio.backhoff@univie.ac.at},
ETH Zurich, Switzerland,
\href{mailto:gudmund.pammer@math.ethz.ch}{gudmund.pammer@math.ethz.ch}}
\thanks{The authors thank Mathias Beiglb\"ock for insightful discussions.}
\title{Quantitative Fundamental Theorem of Asset Pricing}

\maketitle

\begin{abstract}
In this paper we provide a quantitative analysis to the concept of arbitrage, that allows to deal with model uncertainty without imposing the no-arbitrage condition. In markets that admit ``small arbitrage", we can still make sense of the problems of pricing and hedging. The pricing measures here will be such that asset price processes are close to being martingales, and the hedging strategies will need to cover some additional cost. We show a quantitative version of the Fundamental Theorem of Asset Pricing and of the Super-Replication Theorem. Finally, we study robustness of the amount of arbitrage and existence of respective pricing measures, showing stability of these concepts with respect to a strong adapted Wasserstein distance. \\

\end{abstract}

\section{Introduction}

The paradigm of no-arbitrage is arguably the pillar upon which mathematical theory of finance is based. It is mostly considered as a minimal equilibrium or efficiency condition to require on a financial market, because if it were to fail, then agents would exploit it and consequently arbitrage opportunities would disappear from the market. Huge effort is put in order to understand whether a given setting admits arbitrage or not, and to provide characterizations thereof. 
The renowned Fundamental Theorem of Asset Pricing (FTAP) deserves a special mention, as it crucially connects arbitrage to pricing; see the early works \cite{Kr81,HaKr79,HaPl81,DaMoVi90}, as well as \cite{DeSc06,KaMh09} for an overview and
\cite{GuRaSc10,AcBePeSc13,biagini2017robust,bayraktar2016fundamental,herdegen2017no} for further developments.

In the present work, in a finite discrete-time setting, we go beyond this paradigm by accommodating markets that allow for some ``small arbitrage". The main motivation behind our study is the possibility to consider model uncertainty in a strong sense, i.e.\ without imposing no-arbitrage conditions on markets neighbouring a reference market model. Indeed, usually model uncertainty is dealt with by considering some neighbourhood of the selected model, while still imposing some idealized conditions (such as no arbitrage, no transaction costs, etc.). Instead, we would like to allow uncertainty to go beyond such restrictions.
This will allow us to understand robustness of known results w.r.t.\ small perturbations of the model, e.g.\ also accommodating small price shocks in the market.{ We can for example consider the situation where we fixed a no-arbitrage model which we use for pricing and hedging, and we want to understand how these operations and our decision making perform w.r.t.\ model uncertainty. Now, without imposing no-arbitrage conditions on the model misspecification, in neighbouring models we can still make sense of the concepts of pricing and hedging, and prove continuity in an appropriate sense. This will in particular require a relaxation of the martingale property of pricing measures. The fact that one can study hedging and pricing while working outside of the classical no-arbitrage and martingale setting provides a relaxed framework that should in turn be useful in data-driven models, where it may be difficult to check for no-arbitrage conditions, as well as in numerical and machine learning approaches, where imposing a martingale constraint would be costly, see e.g.\ \cite{xu2020cot}.}

In order to talk of ``small arbitrage", we need to quantify a concept that so far was a dichotomy: either there are arbitrage opportunities in the market, or there aren't.
For us, the ``amount of arbitrage" in a market will correspond to the maximal amount that one can get by investing in it without risk. Clearly this makes sense only under some normalization or constraint on the strategies, as otherwise one would be able to push the obtainable amount to infinity as soon as this is non-zero. This leads us to the concept of (strict) $\eps$-arbitrage, defined as a self-financing strategy $H$ such that 
\begin{align}\label{eq:intro_arb}
\P\left[  (H \bullet S)_T - \epsilon \| H \| \ge 0 \right] = 1\quad \text{ and }\quad  \P\left[ (H \bullet S)_T - \epsilon \| H \| > 0 \right] > 0.
\end{align}
The additional term as compared to the classical notion of arbitrage, namely $\epsilon \| H \|$, has a natural interpretation in terms of costs for holding or managing the portfolio associated to the strategy $H$, {but should not be confused for a transaction cost; cf. discussion on related literature at the end of the introduction}. 

In parallel, the martingale property of pricing measures will be also relaxed. 
The notion of pricing measure in a market that may admit small arbitrage, is that of $\eps$-martingale measure, which is a measure $\Q\sim \P$ such that, at every time step $t$, it makes the market $\eps$-away from being a martingale: {i.e.\ we have the almost sure inequality}
\begin{align}\label{eq:intro_mart}
\left|
\E_\Q[S_t | \F_{t - 1}] - S_{t - 1}
\right|
\le \epsilon .
\end{align}
This concept is the analogous, but in a model-dependent setting, to that of $\eps$-approximating martingale measure in \cite{GuOb19}. 
Such approximating martingale measures were used by Guo and Ob{\l}{\'o}j in order to develop a convergence result, and computational methods, in a model-free framework. Also in a model-free setting, Dolinsky and Soner derive in \cite{DoSo13} a similar but different concept of $\eps$-martingale in connection to proportional transaction costs.

In our $\eps$-relaxed framework, we first establish two core results comparable to those in classical mathematical finance: the fundamental theorem of asset pricing and the pricing-hedging duality. The former can be roughly stated as
\[
\text{there is no $\eps$-arbitrage}\quad \Longleftrightarrow\quad \text{there exists an $\eps$-martingale measure}.
\]
As for the pricing-hedging duality result, as in the classical setting we are relating the supremum over the fair prices of an option to the minimum amount necessary to set up a strategy that super-replicates the option at maturity, with the exception that in our $\eps$-relaxed framework we also need to cover some additional cost.
A natural question is whether these results can be recovered directly by means of the analogous results from the classical arbitrage framework.
In this spirit, we show in Remark~\ref{rem:eps_mart_arb} below that the existence of an $\epsilon$-martingale measure has the absence of $\epsilon$-arbitrage as a consequence.
An essential observation in this context is that $\epsilon$-martingale measures can be characterized by the Doob decomposition of the underlying asset.
It turns out that proving the reverse implication and also the pricing-hedging duality is mathematically more challenging as it requires a deeper understanding of the concept of $\epsilon$-arbitrage, which we explore in Section~\ref{sec:canonic_subsec}.

A fundamental insight in the no-arbitrage theory in continuous time is that in order to obtain the fundamental theorem of asset pricing one has to understand the sequential closure of the set of stochastic integrals of admissible trading strategies.
Strikingly, even though we consider here a discrete-time setting, a similar phenomenon occurs and one has to distinguish between the notion of (strict) $\epsilon$-arbitrage and its sequential closure, cf.\ Definition~\ref{def:arb} and Example~\ref{ex:nostrictarb_but_arb}.
We further show that trading strategies can be canonically decomposed, see Lemma~\ref{lem:H_bar}, and use this to derive an explicit representation of this sequential closure which is also used to derive the pricing-hedging duality. Two admittedly technical findings arise which we also consider worth emphasizing: On the one hand, all our results and arguments work the same if we choose in \eqref{eq:intro_arb} a strictly convex $p$-norm (resp.\ in \eqref{eq:intro_mart} its dual norm), meaning that the Hilbertian structure present when $p=2$ is of no advantage. And on the other hand, when we choose the $1$-norm in \eqref{eq:intro_arb}, or  when $\epsilon$ is larger than a model-dependent threshold, then our arguments substantially simplify. 
For instance, strict arbitrage is already enough and no sequential closures are needed anymore.

After settling the aforementioned classical questions, we move on to study stability of the above concepts. We introduce an adapted $L^\infty$ distance that ensures continuity/stability of three key quantities: the amount of arbitrage in a market, the existence of pricing measures, and fair prices. We will see how, in order to guarantee stability, a  distance of $L^\infty$-type is needed. In fact, any distance that would allow processes to be far with small probability, would fail the purpose. We will also see how continuity may fail without the adaptedness (bicausal) constraint. This comes as no surprise, as adapted distances proved to be the correct class of distances in order to obtain robustness of stochastic optimization problems in finance, see \cite{BaBaBeEd19a}. Finally, we establish an approximation of the adapted $L^\infty$ distance by adapted log-exponential divergences. This is relevant since the adapted $L^\infty$ distance, albeit natural in this context, seems to be unwieldy for applications, while the log-exponential approximating versions thereof appear to be more tractable. To substantiate this point, we show the consistency of adapted empirical estimators, introduced in \cite{BaBaBeWi20}, w.r.t.\ the adapted log-exponential divergences. 
\\

{ 
{\bf On related literature}
From the early works \cite{Kr81,HaKr79,HaPl81,DaMoVi90}, the concept of arbitrage (in its variations) and its relation with pricing measures have played a crucial role in any development of the mathematical finance field. 
Both the quantitative definition of arbitrage \eqref{eq:intro_arb}, and the related concept of $\epsilon$-martingale measure \eqref{eq:intro_mart}, have some similarity to related concepts coming from the theory of transaction costs. However, to the best of our knowledge, the concepts proposed here cannot be accommodated into the existing literature. At a conceptual level, a transaction cost should only apply when portfolios are rebalanced, i.e.\ they should be a function of the increments of $H$ in our notation. By contrast \eqref{eq:intro_arb} pertains the actual size of the trading strategy $H$ rather. Even in the context of Kabanov's very general model \cite{Ka09,Sch04} of transaction costs, we could not see how to subsume  \eqref{eq:intro_arb} into that framework. On the dual side, the analogue of \eqref{eq:intro_mart} typically involves the term $\mathbb E_{\mathbb Q}[S_T|\mathcal F_{t-1}]$ rather than  $\mathbb E_{\mathbb Q}[S_t|\mathcal F_{t-1}]$, see e.g.\ \cite{ChKuTa17}, or the r.h.s.\ in \eqref{eq:intro_mart} would have to be replaced by $\epsilon S_{t-1}$ e.g.\ if $S$ denotes the process of bid prices and $(1+\epsilon)S$ is the process of ask prices. We refer the reader to  \cite{GuMu13,Sch17} for further bibliographical discussion and to \cite{PePe18} for a recent reference concerning the modelling of iliquitity in a very general discrete-time framework.

In a model-independent context, the case of transactions costs has been considered in \cite{DoSo13}, where it is shown that pricing measures correspond to a modification of martingale measures similar to the $\eps$-martingales introduced in the present paper. However, the concept of pricing measures in \cite{DoSo13}, similarly to the literature mentioned above, cannot be reduced to a local condition to be satisfied at every time step as we have in  \eqref{eq:intro_mart}.
Also in a robust setting, and
closely related to the topic of the present paper, is the work by Guo and Ob{\l}{\'o}j \cite{GuOb19}.
There the authors introduce a relaxed martingale optimal transport in order to study the numerical resolution as well as stability of martingale optimal transport.
This setting can be seen as a model-independent version (given marginal observations) of our model-dependent setting.
In particular, \cite[Theorem 4.1]{GuOb19} is a robust analogue of the pricing-hedging duality, cf.\ Theorem \ref{thm:duality}, established down below.
A central part of their analysis is \cite[Proposition 4.2]{GuOb19} which is an existence result of $\epsilon$-approximating martingales under perturbations, akin to Proposition \ref{prop:cont} in the current paper.
To measure the magnitude of the perturbation, a Wasserstein distance specific to their setting is used instead of adapted Wasserstein distances.

Finally we stress the existence of a line of literature concerned with the exploitation of arbitrages, or the idea of hedging even in models admitting arbitrage; see e.g.\ \cite{Ru13,KaFe09}.
}
\\

\noindent{\bf Organization of the paper.} We end this section by introducing the notation that will be used throughout the paper. Then, in Section~\ref{sect:arb}, we define the concepts of $\eps$-arbitrage and $\eps$-martingale measure, and state the fundamental theorem of asset pricing and the pricing-hedging duality in this framework.
Section~\ref{sect:cont} is devoted to study stability of the above concepts. Here we introduce a strong adapted distance, and show that continuity w.r.t.\ it holds.
In Section~\ref{sect:reg} we study regularization, introducing exponential distances that approximate the strong adapted distance and allow for empirical estimation.
Postponed proofs and auxiliary results are contained in Section~\ref{sect:proofs}.
\\

\noindent{\bf Notation and Conventions:} 

\begin{itemize}
\item We fix throughout $p \in [1,\infty)$ and $q \in (1,\infty]$  such that $\frac1p + \frac1q = 1$.
    \item For $x$ and $y$ vectors in $\R^d$, we denote by $x\cdot y$ their scalar product.
    \item By $|\cdot|_p$ we denote the $p$-norm on $\R^d$, that is, $|x|_p^p = \sum_{i = 1}^d |x_i|^p$
    \item For a $d$-dimensional stochastic processes $X=(X_t)_{t=t_0,\ldots,T}$, we write $X_{t,i}$ for the $i$-th component of the random vector $X_t\in\R^d$. Moreover, we set $\|X\|_p := \sum_{t = t_0}^{T} |X_t|_p$, where we will have either $t_0=0$ or $t_0=1$ depending on the context.
    \item For $X=(X_t)_{t=0,\ldots,T}$ and $Y=(Y_t)_{t=0,\ldots,T}$ $d$-dimensional stochastic processes, we denote by $(Y\bullet X)$ the component-wise stochastic integral of $Y$ w.r.t.\ $X$, so that $(Y\bullet X)_0=0$ and, for $u=1,\ldots,T$, $(Y\bullet X)_u=\sum_{t=1}^uY_t\cdot \Delta X_t$, where $\Delta X_t:=X_t-X_{t-1}$.
    \item Equalities and inequalities between random variables are intended almost surely (a.s.). 
    \item $L^0(\Omega,\mathcal G, \P;A)$ denotes the space of random variables defined on the stochastic basis $(\Omega,\mathcal G, \P)$ with values in the measurable space $A$.
    \item $L^r(\Omega,\mathcal G, \P;\R^d)$ denotes the space of $r$-integrable random $d$-vectors. We will usually write $L^r_+(\Omega,\mathcal G, \P):= L^r(\Omega,\mathcal G, \P;\R_+) $. We say that a stochastic process $X$ is in $L^r(\Omega,\mathcal G, \P)$ if each $X_t\in L^r(\Omega,\mathcal G, \P;\R^d)$.
    \item A sequence $(f_n)_{n \in \N}\subset L^0(\Omega,\mathcal G, \P; \R)$ is bounded (resp.\ bounded from below) if there is $c\in L^0(\Omega,\mathcal G, \P; \R_+)$ such that $|f_n|\leq c$ (resp.\ $f_n\geq -c $) for all $n$ a.s. 
Moreover, we say that a sequence $(f_n)_{n \in \N}\subset L^0(\Omega,\mathcal G, \P; \R^d)$ is bounded if $(|f_n|)_{n \in \N}\subset L^0(\Omega,\mathcal G, \P; \R)$ is bounded.
\end{itemize}

\section{On \texorpdfstring{$\epsilon$}{}-arbitrage and \texorpdfstring{$\epsilon$}{}-martingale measures} \label{sect:arb}

We consider a financial market model, defined on a filtered probability space $(\Omega, \F, \mathbb{F}:=(\F_t)_{t  = 0}^T, \P)$, { consisting of $d$ risky assets and one riskless asset used as num\'eraire. The evolution of the discounted value of the $d$ risky assets is then an $\R^d$-valued, $\mathbb{F}$-adapted process $S$.}
As usual trading strategies on the risky assets are assumed to be predictable processes. The set of admissible trading strategies is then denoted by
$\mathcal H := \{ H=(H_t)_{t=1}^{T} \colon H_t\in\h_t \}$,
where $\h_t=L^0(\Omega,\F_{t-1}, \P;\R^d)$. Investing in the market $S$ with an initial capital $x\in\R$ and following a strategy $H\in \mathcal H$ {in a self-financing manner} then results in the portfolio value $x+(H\bullet S)_T$.

This section contains the main results of this paper, and to make it more readable we postpone most of the proofs to {Section~\ref{sect:proofs}}. Unless differently specified, definitions and results below hold for any $p\in[1,\infty)$. In some cases we will need to differentiate, in particular distinguishing between the two cases $p=1$ and $p>1$.

\begin{definition}[$\epsilon$-arbitrage]\label{def:arb} { Let $\epsilon\geq 0$.}
A trading strategy $H$ in $\mathcal H$ is a \emph{strict $\epsilon$-arbitrage} if
{
\begin{equation}
\label{eq:def_strict_epsilon_arbitrage}
\P\left[  (H \bullet S)_T - \epsilon \| H \|_p=:Y \ge 0 \right] = 1\quad \text{ and }\quad  \P\left[ Y > 0 \right] > 0.
\end{equation}
A sequence $(H^k)_{k \in \N}$ of trading strategies in $\mathcal H$ is called an \emph{$\epsilon$-arbitrage} if
\begin{equation}\label{eq:def_epsilon_arbitrage}
\P\left[\liminf_{k \to \infty} \left\{ (H^k \bullet S)_T - \epsilon \|H^k\|_p \right\} =: Y_\infty \ge 0\right]=1\quad 
\text{ and }\quad
\P \left[ Y_\infty > 0 \right] > 0.
\end{equation}
}
The market is said to have no $\epsilon$-arbitrage, or said to satisfy $\NA_\epsilon$, if there exists no such $\epsilon$-arbitrage.
\end{definition}

{ In this definition one can  replace $Y$ in \eqref{eq:def_strict_epsilon_arbitrage}  by $V_T-V_0-\epsilon \| H \|_p$, whereby $V$ is the uniquely defined discounted self-financing wealth process associated to $H\in\mathcal H$ and $V_0\in\mathbb R$. A similar adaptation of \eqref{eq:def_epsilon_arbitrage} is possible. As in the classical theory, this does not change the class of (strict) $\epsilon$-arbitrages, so we will employ Definition \ref{def:arb} throughout.} 

{ Clearly, for $\epsilon=0$, a strategy satisfying \eqref{eq:def_strict_epsilon_arbitrage} is an arbitrage strategy in the classical sense (see \cite{DaMoVi90}). Moreover, absence of such strategies already excludes existence of sequences of strategies satisfying
\eqref{eq:def_epsilon_arbitrage}; see e.g.\ \cite{FoSc16,DeSc06}. This means that our notion of $\NA_0$ coincides with the classical no-arbitrage ($\NA$) condition.
In the case $\epsilon=0$ the analysis that follows is then already covered by the classical theory, and therefore we assume throughout that $$\epsilon >0.$$
}

As mentioned in the introduction, in order to have a notion of  ``amount of arbitrage'' admitted in a market, some normalization or constraint on the admissible trading strategies is needed{, otherwise by linearity one would get an ``infinite arbitrage" as soon as arbitrage is possible}. This is the role played by the term $\epsilon \| H \|_p$ in \eqref{eq:def_strict_epsilon_arbitrage}, that can be interpreted as the cost faced for holding or managing the portfolio associated to the strategy $H$. 
In fact, existence of strict $\epsilon$-arbitrage means existence of strategies $K\in\mathcal H$ with $\E_\P[\|K\|_p]=1$ whose implementation leads to a gain of at least $\epsilon$ and with positive probability strictly larger than that.

It is clear from Definition~\ref{def:arb} that the absence of strict $\epsilon$-arbitrage implies that, for any $\epsilon' > \epsilon$ there is no strict $\epsilon'$-arbitrage either.
Therefore, the set of all values $\epsilon > 0$ for which the market model admits no strict $\epsilon$-arbitrage is an interval and its infimum will play a special role.

\begin{definition}[critical value]\label{def:critical_value}
We define the critical value as
\begin{equation}\label{def:crit}
\eps(\P):=\inf\{ \epsilon > 0 \colon \text{there is no strict $\epsilon$-arbitrage under model $\P$}\}. 
\end{equation}
\end{definition}

In fact, in  the definition of critical value one could replace  `strict $\epsilon$-arbitrage' by `$\epsilon$-arbitrage' without changing its value, as the next proposition clarifies.

{
\begin{proposition}\label{thm:epscritval}
There is no $\epsilon$-arbitrage for any $\epsilon$  strictly bigger than the critical value $\epsilon(\P)$.
\end{proposition}
The proof is postponed to Section~\ref{sec:NA_multistep}.  
}

The critical value $\epsilon(\P)$ plays a special role, as for $\epsilon$ strictly bigger than $\epsilon(\P)$ our results simplify significantly, resembling more the classical setting; see Section~\ref{sect:crit}.

{
\begin{remark}\label{rem.epsp1}
Note that the statement of Proposition~\ref{thm:epscritval} cannot be strengthened by including the critical value, that is, there may exist $\eps(\P)$-arbitrage,  as the next example shows.
\hfill$\Diamond$
\end{remark}
 }

\begin{example}[A market with $\eps$-arbitrage but no strict $\eps$-arbitrage]
\label{ex:nostrictarb_but_arb}
To motivate the necessity of considering the sequential closure of strict $\epsilon$-arbitrage, that is $\epsilon$-arbitrage as in Definition \ref{def:arb},
for simplicity we let $p = 2$, and consider a single-period market consisting of two assets and two possible (i.e., with positive probability) events, i.e., $\Omega := \{\omega_1,\omega_2\}$.
The asset $S$ is given by
\[
S_0(\omega) := (0,0)\quad \mbox{ and }\quad
S_1(\omega) := 
\begin{cases}
(\epsilon, 0), & \omega = \omega_1, \\
(\epsilon, 1), & \omega = \omega_2.
\end{cases}
\]
We refer to Figure \ref{fig:na} below for a depiction of this setting.
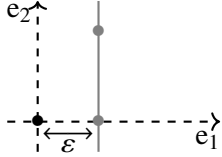
\begin{figure}[h]
\centering
\begin{tikzpicture}[scale=0.8]
\draw[thick,dashed,->] (-0.5,0)--(3,0);
\draw[thick,dashed,->] (0,-0.5)--(0,2);
\draw[thick,gray] (1,-0.5) -- (1,2);
\node at (0,0) {\color{black}$\bullet$};
\node at (1,0) {\color{gray}$\bullet$};
\node at (1,1.5) {\color{gray}$\bullet$};
\node at (-0.3,1.8) {e${}_2$};
\node at (2.8,-0.3) {e${}_1$};
\node at (0.5,-0.2) {$\longleftrightarrow$};
\node at (0.5,-0.4) {$\eps$};
\end{tikzpicture}
\caption{A market with $\eps$-arbitrage but no strict $\eps$-arbitrage.
The black bullet indicates the initial state of the market, $S_0=(0,0)$, while the gray bullets represent the two possible states of the market at time $1$. Here we have that $\Delta S_1\in \{\text{e}_1\}^\perp+\eps\text{e}_1$.}
\label{fig:na}
\end{figure}

Consider the filtration generated by $S$.
Then $\mathcal H=\mathcal H_1$ is composed exclusively of constant trading strategies, that is $\mathcal H \equiv \R^2$.
Let $H = (a,b) \in \R^2$, then we have
\[
    (H \bullet S)_1 =
    \begin{cases}
        \epsilon a, & \omega = \omega_1,
        \\
        \epsilon a + b, & \omega = \omega_2;
    \end{cases}
    \quad
    \mbox{ and }
    \quad
    \|H\|_2 = \sqrt{a^2 + b^2}.
\]
If $b = 0$ then $(H \bullet S)_1 - \epsilon \|H\|_2 \leq 0$ whereas if $b \neq 0$ then $(H \bullet S)_1 - \epsilon \|H\|_2 < 0$ on $\{ \omega_1 \}$, hence, there is no strict $\epsilon$-arbitrage. { One can also easily see that $\eps=\eps(\P)$.}
On the other hand, consider the admissible sequence of trading strategies $H^k \equiv (k,1)$, $k \in \N$.
By elementary calculus we have that $\lim_{k \to \infty} k - \sqrt{k^2 + 1} = 0$, thus,
\begin{align*}
    \lim_{k \to \infty} (H^k \bullet S)_1(\omega) - \epsilon \| H^k (\omega) \|_2 &=
    (0,1) \cdot \Delta S_1(\omega) + \epsilon \lim_{k \to \infty} (k - \sqrt{k^2 + 1})
    \\
    &=
    \begin{cases}
        0, & \omega = \omega_1, \\
        1, & \omega = \omega_2.
    \end{cases}
\end{align*}
We conclude that $(H^k)_{k \in \N}$ is an $\epsilon$-arbitrage.\hfill$\Diamond$
\end{example}

{
\begin{example}
We complement Example \ref{ex:nostrictarb_but_arb} by making the case that, independent of the power $p \ge 1$, the critical value $\epsilon(\P)$ may yield strict $\epsilon(\P)$-arbitrage.
In the same setting as Example \ref{ex:nostrictarb_but_arb}, we consider a single asset given by
\[
S_0(\omega) := 0 \quad \text{and} \quad S_1(\omega) := 
\begin{cases}
\epsilon, & \omega = \omega_1, \\
1 + \epsilon, & \omega = \omega_2.
\end{cases}
\]
For this particular scenario, we have $\epsilon(\mathbb P) = \epsilon$ while there is strict $\epsilon$-arbitrage.
\hfill$\Diamond$
\end{example}
}

Similarly to the classical case, it will prove crucial to introduce the set
$$K := \{ (H\bullet S)_T - \epsilon \|H\|_p \colon H \in \mathcal H\},$$
and denote by $\overline{K}$ its closure w.r.t.\ convergence in probability. {Indeed, the set $K$ need not be closed, see 
Example~\ref{ex:KbarK} for a simple setting showing this phenomenon and Theorem~\ref{thm:K=barK} for sufficient conditions ensuring $K = \overline{K}$.}
With this at hand, we also set
\[
C := \{ X \in L^0(\Omega,\F_T, \P; \R) \colon \exists Y \in \overline{K} \text{ with } X \le Y \; {\P\text{-a.s.}}\}.
\]
It turns out that $\NA_\epsilon$ provides much information concerning the latter set. For instance it guarantees the crucial closure property:

\begin{theorem}
\label{cor:C_closed}
Under $\NA_\epsilon$ the set $C$ is a convex cone that is closed w.r.t.\ convergence in probability and satisfies $C \cap L^0(\Omega,\F_T, \P; \R_+) = \{ 0 \}$.
\end{theorem}
{
The proof is postponed to Section~\ref{sec:NA_multistep}.  
}

\subsection{Fundamental theorem of asset pricing }\label{sect:ftap}
We relate the concept of $\epsilon$-arbitrage introduced above to measures which are close to being martingale measures in the following sense:    
\begin{definition}[$\epsilon$-martingale and $\epsilon$-martingale measure]
\label{def:eps_mart}
An adapted process $X$ on a filtered probability space $(\Omega, (\F_t)_{t = 0}^T, \F, \Q)$ is called an $\epsilon$-martingale under $\Q$ if $X_T \in L^1(\Q)$ and, for $t = 1,\ldots,T$,
\begin{equation}
\label{eq:def_eps_mart}
\left|
\E_\Q[X_t | \F_{t - 1}] - X_{t - 1}
\right|_q
\le
\epsilon {\quad\Q\text{-a.s.}}
\end{equation}
We call a measure $\Q\sim\P$ an $\epsilon$-martingale measure if the discounted asset price $S$ is an $\epsilon$-martingale under $\Q$.
\end{definition}
The above concept of $\epsilon$-martingale measure is the analogous, but in a model-dependent setting, to that of $\eps$-approximating martingale measure introduced in \cite{GuOb19} in a model-free framework.

\begin{remark}
Note that, for fixed $t \in \{1,\ldots,T\}$, the condition given in \eqref{eq:def_eps_mart} can be tested by considering bounded strategies in $\mathcal H_t$.
Indeed, first remark that \eqref{eq:def_eps_mart}, $X_T \in L^1(\Q)$, and the tower property, imply that $X_t \in L^1(\Q)$ for all $t$. Second, as $|\cdot|_p$ and $|\cdot|_q$ are dual norms, given $\Delta X_t \in L^1(\Q)$, \eqref{eq:def_eps_mart} is equivalent to demanding for all $H_t \in \mathcal H_t$ with $|H_t|_p \leq 1$ that
\[
\E_\Q[H_t \cdot \Delta X_t - \epsilon |H_t|_p ] \le 0.
\]
In particular, given $X_T \in L^1(\Q)$, \eqref{eq:def_eps_mart} holds for all $t \in \{1,\ldots, T\}$ if and only if, for all $H \in \mathcal H$ with $\|H\|_p\le 1$,
\[
\E_\Q[(H \bullet X)_T - \epsilon \|H\|_p] \le 0.
\]
\hfill$\Diamond$
\end{remark}

In Example~\ref{ex:nsaem} below we show how the absence of strict $\epsilon$-arbitrage in the market does not guarantee the existence of an $\epsilon$-martingale measure. On the other hand, we will see in Remark~\ref{rem:eps_mart_arb} that the reverse is true even in a stronger form, that is, the existence of such a measure excludes the possibility of $\epsilon$-arbitrage opportunities.

\begin{example}\label{ex:nsaem}
We build on Example \ref{ex:nostrictarb_but_arb} and set $\P$ to be the uniform distribution on $\Omega$.
The set of equivalent measures $\{ \Q \colon \Q \sim \P \}$ coincides with $\{ \Q_\alpha := \alpha \delta_{\omega_1} + (1 - \alpha) \delta_{\omega_2} \colon \alpha \in (0,1) \}$.
For $\alpha \in (0,1)$ and $p \in (1,\infty)$ we compute \eqref{eq:def_eps_mart}, which is, as the filtration is trivial at time 0, given by
\[
\left|\E_{\Q_\alpha} [S_1 | \F_0 ] - S_0\right|_q = 
\left|\E_{\Q_\alpha}[S_1] \right|_q = 
|(\epsilon,0) \alpha + (\epsilon,1) (1 - \alpha)|_q > \epsilon.
\]
This situation is represented in Figure \ref{fig:no_eps_mart} below.
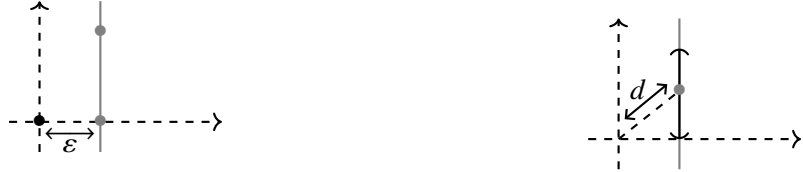
\begin{figure}[h]
\centering
\begin{subfigure}{0.45\textwidth}
\centering
\begin{tikzpicture}[scale=0.8]
\draw[thick,dashed,->] (-0.5,0)--(3,0);
\draw[thick,dashed,->] (0,-0.5)--(0,2);
\draw[thick,gray] (1,-0.5) -- (1,2);
\node at (0,0) {$\bullet$};
\node at (1,0) {\color{gray}$\bullet$};
\node at (1,1.5) {\color{gray}$\bullet$};
\node at (0.5,-0.2) {$\longleftrightarrow$};
\node at (0.5,-0.4) {$\eps$};
\end{tikzpicture}
\caption{Evolution of two assets: black bullet = position at time $0$, gray bullets = possible positions at time $1$.}
\label{fig:no_eps_mart.left}
\end{subfigure}
\qquad
\begin{subfigure}{0.45\textwidth}
\centering
\begin{tikzpicture}[scale=0.8]
\draw[thick,dashed,->] (-0.5,0)--(3,0);
\draw[thick,dashed,->] (0,-0.5)--(0,2);
\draw[thick,gray] (1,-0.5) -- (1,2);
\draw[(-),thick,black] (1,0) -- (1,1.5);
\draw[dashed,thick] (0,0)--(1,0.8);
\node at (1,0.8) {\color{gray}$\bullet$};
\draw[<->,>={Straight Barb[scale=0.8]},thick,black] (0.1,0.33) -- (0.8,0.91) node[midway,xshift=-3pt,yshift=4.5pt] {$d$};
\end{tikzpicture}
\caption{The grey bullet is the mean of $S_1$ under some measure $\Q\sim\P$, $d$ is the distance of this average from $S_0$.}
\label{fig:no_eps_mart.right}
\end{subfigure}
\caption{Figure \ref{fig:no_eps_mart.right} depicts the interval of expectations $\E_\Q[S_1]$, where $\Q$ runs over the set of measures equivalent to $\P$.
We see that any such mean (illustrated by a gray bullet) satisfies that $|\E_\Q[S_1| \F_0 ] - S_0|_q = d > \epsilon$ and, therefore, the asset $S$ violates \eqref{eq:def_eps_mart} under $\Q$.}
\label{fig:no_eps_mart}
\end{figure}
In particular, there exists no $\epsilon$-martingale measure in this market. \hfill$\Diamond$
\end{example}

In order to  show that the existence of an $\eps$-martingale measure excludes the possibility of $\eps$-arbitrage opportunities (in Remark~\ref{rem:eps_mart_arb}), we connect the concept of Doob decomposition to the notion of $\epsilon$-martingale.
In this regard, recall that in the classical no-arbitrage framework, the Doob decomposition of the assets price process $S$ under a measure $\Q$  already characterizes whether $\Q$ is a martingale measure.
That is, $\Q$ is a martingale measure if and only if the predictable part of the Doob decomposition of $S$ vanishes.
The next remark explores an analogous relation in the current framework of $\epsilon$-arbitrage.
\begin{remark}[Doob decomposition]
\label{rem:Doob_decomposition}
Let $X$ be an adapted process on the filtered probability space $(\Omega, (\F_t)_{t = 0}^T, \F, \Q)$ and denote by $(A,M)$ its Doob decomposition, with $A_0 = 0$, $A$ predictable, and $M$ a $\Q$-martingale.
Then we have
\[
\E_\Q [ X_{t + 1} | \F_t ] - X_t = \E_\Q [ \Delta X_{t+1} | \F_t ] = \E_\Q [\Delta A_{t+1} | \F_t] = A_{t + 1 } - A_t,
\]
which shows that $X$ is an $\epsilon$-martingale under $\Q$ if and only if $|\Delta A_t|_q \le \epsilon$ for every $t \in \{1,\ldots, T\}$.
To put this differently,
$X$ is an $\epsilon$-martingale under $\Q$ if and only if there exists a predictable process $A$ with $|\Delta A_t|_q \le \epsilon$  for every $t \in \{1,\ldots, T\}$ such that $X - A$ is a martingale under $\Q$.
\hfill$\Diamond$
\end{remark}

\begin{remark}[$\epsilon$-martingales and $\epsilon$-arbitrage]
\label{rem:eps_mart_arb}
Let $\Q\sim\P$ be such that $S$ is an $\epsilon$-martingale under $\Q$ with Doob decomposition $(A,M)$.
Since $M$ is a martingale under $\Q$ we obtain by classical no-arbitrage theory that, for all $H \in \mathcal H$,
\begin{equation}
\label{eq:rem.eps_mart_strict_arb.1}
(H \bullet M)_T \ge 0 \implies (H\bullet M)_T = 0.
\end{equation}
On the other hand, as $S = A + M$ and $|\Delta A_t|_q \le \epsilon$ for $t = 1,\ldots, T$ we also have, 
for all $H \in \mathcal H$,
\begin{equation}
\label{eq:rem.eps_mart_strict_arb.2}
(H\bullet M)_T = (H \bullet S)_T - (H \bullet A)_T
\ge
(H \bullet S)_T - \epsilon \|H\|_p.
\end{equation}
Let $(H^k)_{k \in \N}$ be a sequence of strategies in $\mathcal H$ such that $Y := \liminf_{k \to \infty} (H^k \bullet S)_T - \epsilon \|H^k\|_p$ is non-negative.
By \eqref{eq:rem.eps_mart_strict_arb.2} we have that $((H^k \bullet M)_T)_{k \in \N}$ is almost surely bounded from below, thus, we may invoke \cite[Proposition 6.9.1]{DeSc06} to get $H \in \mathcal H$ with
\[
(H \bullet M)_T \ge \liminf_{k \to \infty} (H^k \bullet M)_T \ge Y \ge 0.
\]
Hence, \eqref{eq:rem.eps_mart_strict_arb.1} yields that $(H \bullet M)_T = 0$ and therefore the absence of $\epsilon$-arbitrage.\hfill $\Diamond$
\end{remark}

Building on Theorem~\ref{cor:C_closed} we can obtain a version of the Fundamental Theorem of Asset Pricing (FTAP) in the present setting. In the sequel we denote 
\[\mathcal M_\epsilon(\P):= \left\{\Q\text{ is an $\epsilon$-martingale measure and } \frac{d\Q}{d\P}\in L^\infty(\Omega,\F_T,\P) \right\}.\]
 
\begin{theorem}[FTAP]\label{thm:ftap}
The market satisfies $\NA_\epsilon$ if and only if  $\mathcal M_\epsilon(\P)\neq \emptyset$.
\end{theorem}
{
The proof is postponed to Section~\ref{sect:proof:ftap}.  
}

A simple consequence of this result is the alternative representation of $\epsilon(\P)$:
\begin{corollary} \label{corol:crit}
The critical value, given in Definition~\ref{def:critical_value}, satisfies
\begin{equation}\label{eq:critm}
\epsilon(\P)=\inf\{\epsilon\geq 0 : \mathcal M_\epsilon(\P)\neq\emptyset \}.
\end{equation}
\end{corollary}
\begin{proof}
This follows by Theorem~\ref{thm:ftap} and the fact that, by Proposition~\ref{thm:epscritval},  one can replace  `strict $\epsilon$-arbitrage' with `$\epsilon$-arbitrage' in Definition~\ref{def:critical_value}.
\end{proof}

\begin{remark}\label{rem.epsp2}
We note that the infimum in \eqref{eq:critm} need not be attained, that is, we can have  $\mathcal M_{\epsilon(\P)}(\P)=\emptyset$, 
{
see Example~\ref{ex:nsaem}
and Remark~\ref{rem.epsp1}}.
\hfill$\Diamond$
\end{remark}

{
\begin{remark} 
Note that the critical value satisfies
\begin{equation}\label{epsp3}
\eps(\P)=\inf_{\Q}\max_{t=0,1,\ldots,T-1}\esssup_\P\left|\E_\Q[S_{t+1}-S_t|\F_t]\right|,
\end{equation}
where the infimum is taken over measures $\Q$ equivalent to $\P$ having  bounded density w.r.t.\ $\P$.
Indeed, let us denote by $\hat\eps$ the rhs of \eqref{epsp3}, and 
consider any $\eps$ s.t. $\mathcal{M}_\eps(\P)\neq\emptyset$, say with $\Q\in\mathcal{M}_\eps(\P)$. 
Then Theorem~\ref{thm:ftap} implies that $\sup_{t=0,1,\ldots,T-1}\left|\E_\Q[S_{t+1}-S_t|\F_t]\right|\leq\eps$ a.s., so that $\hat\eps\leq\eps$. This in turn gives $\hat\eps\leq\eps(\P)$. 
On the other hand, for any $\eps<\eps(\P)$, consider $\delta>0$ s.t.\ $\eps<\eps(\P)-\delta$.
Then $\mathcal{M}_{\eps+\delta}(\P)=\emptyset$, which, by definition of $\mathcal{M}_{\eps+\delta}(\P)$,
implies that for every $\Q\sim\P$  with bounded density w.r.t. $\P$ there is a $t\in\{0,1,\ldots,T-1\}$ s.t. $\left|\E_\Q[S_{t+1}-S_t|\F_t]\right|>\eps+\delta$ with strictly positive probability, so that $\max_{t=0,1,\ldots,T-1}\esssup_\P\left|\E_\Q[S_{t+1}-S_t|\F_t]\right| >\epsilon+\delta$. 
This yields $\hat\eps\geq \eps+\delta>\eps$, and thus $\hat\eps=\eps(\P)$.

\hfill$\Diamond$
\end{remark}
}

\subsection{Canonical decomposition of trading strategies}\label{sec:canonic_subsec}
In order to obtain the previous theorems we need to understand much better what are the consequences of $\NA_\epsilon$.
For this reason we introduce
\begin{align*}
E_{\epsilon, t} :=& \left\{ H_t \in \mathcal H_t \colon H_t \cdot \Delta S_t = \epsilon |H_t|_p  \right\}, t=1,\ldots,T,\quad
&E_\epsilon := \left\{ H \in \mathcal H \colon H_t \in E_{\epsilon,t} \text{ for } t = 1,\ldots, T \right\},\\
E_{0, t} :=& \left\{ H_t \in \mathcal H_t \colon H_t \cdot \Delta S_t = 0  \right\}, t=1,\ldots,T,\quad
&E_0 :=\left\{ H \in \mathcal H \colon H_t \in E_{0,t}
\text{ for } t = 1,\ldots, T \right\}.
\end{align*}
As $p$- and $q$-norms satisfy a dual relation, for $x \in \R^d$, we denote by
\begin{equation}
\label{eq:def.dual_vector}
x^\ast := \arg \max \left\{ y \cdot x \colon |y|_q = 1 \right\}
\end{equation}
a measurable selection of dual vectors w.r.t.\ the $p$-norm, which can be explicitly defined by
$x^\ast_i = \sgn(x_i) |x_i|^{p - 1} / |x|^{p-1}_p$ for $i \in \{1,\ldots,d\}$\footnote{The sign function is given by $\sgn(x)=-1$ for $x<0$, $\sgn(0)=0$, and $\sgn(x)=1$ for $x>0$.}. {Note that $x\cdot x^\ast=|x|_p$.}
It will prove equally important to consider the sets
\begin{align*}
F_{\epsilon,t} :=& \left\{ X_t \in \mathcal H_t \colon X_t \cdot H_t = |X_t|_q |H_t|_p,\,\forall H_t \in E_{\epsilon,t} \right\}, \quad
&F_\epsilon := \left\{ X \in \mathcal H \colon X_t \in F_{\epsilon,t} \mbox{ for }t = 1,\ldots,T\right\}
\\
F_{\epsilon,t}^\perp :=& \left\{ H_t \in \mathcal H_t \colon X_t \cdot H_t = 0,\,\forall X_t \in F_{\epsilon,t} \right\}, \quad
&F_\epsilon^\perp := \left\{ H \in \mathcal H \colon H_t \in F_{\epsilon,t}^\perp \mbox{ for }t = 1,\ldots,T \right\},
\end{align*}
where the definitions above of $F_{\epsilon,t}^\perp$ and $F_{\epsilon}^\perp$ are only for $p > 1$.
As the case $p = 1$ behaves differently, we give the definition of $F_{\epsilon,t}^\perp$ when $p = 1$ after Lemma \ref{lem:H_bar} below.

We illustrate a first connection between the absence of $\epsilon$-arbitrage and these sets in the following simple lemma.

\begin{lemma}\label{lem:E_conv_cone}
Under the absence of strict $\epsilon$-arbitrage, the sets $ E_{\epsilon,t}$, $t = 1,\ldots, T$, are convex cones.     
\end{lemma}

\begin{proof} Fix any $t \in \{1,\ldots, T\}$.
Clearly $H_t\in  E_{\epsilon,t}$ implies $c H_t\in  E_{\epsilon,t}$ if $c$ is a positive constant. On the other hand, if $H^1_t,H^2_t\in  E_{\epsilon,t}$ then
\begin{equation}
\label{eq:conv_cone}
(H^1_t+H^2_t)\cdot \Delta S_t = \epsilon\{|H^1_t|_p+ |H^2_t|_p\}\geq \epsilon |H^1_t+H^2_t|_p.
\end{equation}
Then, by the absence of strict $\epsilon$-arbitrage, this must be an equality, and thus $H^1_t+H^2_t \in E_{\epsilon,t}$.
\end{proof}

However this connection runs much deeper. For instance, in Lemma~\ref{lem:H_bar} below we establish that, provided there is no strict $\epsilon$-arbitrage and $p > 1$, the set $E_{\epsilon,t}$ is essentially at most one-dimensional.

\begin{lemma}\label{lem:H_bar}
Assume the absence of strict $\epsilon$-arbitrage and let $t = 1,\ldots, T$.
Then there is $\bar H_t \in E_{\epsilon,t}$ with $|\bar H_t|_p \in \{0,1\}$ such that $\bar H_t^\ast \in F_{\epsilon,t}$ 
and
\begin{align*}
    \bar H_t &= \arg\max \{ \P[ H_t \neq 0] \colon H_t \in E_{\epsilon,t} \} &&\mbox{if }p>1, \\
    \bar H_{t{,i}}
    &= \arg\max \{ \P[H_{t,i} \neq 0] \colon H_t \in E_{\epsilon,t} \} && \mbox{if }p = 1, \,  i =1,\ldots,d.
\end{align*}
Moreover, when $p > 1$ we have
\begin{align*}
    E_{\epsilon,t} =\{ a_t \bar H_t \colon a_t \in L^0(\Omega,\F_{t-1},\P;\R_+) \},\quad
    F_{\epsilon,t}^\perp = \left\{ H_t \in \mathcal H_t \colon H_t \cdot \bar H_t^\ast = 0, \{ H_t = 0 \} \supseteq \{ \bar H_t = 0 \} \right\}.
\end{align*}
\end{lemma}
{
The proof is postponed to Section~\ref{sect:proofs.candec}. 
}

Now, we can introduce $F_{\epsilon,t}^\perp$ for the case $p = 1$:
\begin{align*}
F_{\epsilon,t}^\perp := \left\{ H_t \in \mathcal H_t \colon \bar H_t^\ast \cdot H_t = 0, \, \{H_{t,i} \neq 0 \} \subseteq \{\bar H_{t,i} \neq 0 \}\,  \forall i \in \{1,\ldots,d\} \right\}.
\end{align*}
Since $\bar H_t$, that is given in Lemma \ref{lem:H_bar} above, and in particular the set $\{ \bar H_t \neq 0 \}$ will play a special role, cf.\ Lemma \ref{lem:1_step_H_bar_complement}, we introduce the notation
\begin{align*}
\bar{\mathcal H}_t :=& 
\begin{cases}
    \left\{ H_t \in \mathcal H_t \colon \{H_t \neq 0 \} \subseteq \{ \bar H_t \neq 0\} \right\} & p > 1, \\[0.1cm]
    \left\{ H_t \in \mathcal H_t \colon  \{H_{t,i} \neq 0 \} \subseteq \{ \bar H_{t,i} \neq 0 \}\, \forall i \in \{1,\ldots,d\} \right\} & p = 1,
\end{cases}
\\
\bar{\mathcal H} :=& \left\{ H \in \mathcal H \colon H_t \in \bar{\mathcal H}_t\, \mbox{ for }t=1,\ldots, T\right\}.
\end{align*}

In Lemma~\ref{lem:decomp} below we show how trading strategies admit a canonical  decomposition in terms of the sets defined above. This decomposition will be crucial in the proofs of our main results.

\begin{lemma}
\label{lem:decomp}
Assume the absence of strict $\epsilon$-arbitrage.
Then any $H_t \in \bar{\mathcal H}_t$ can be uniquely decomposed as
\begin{equation}
    \label{eq:canonical_decomp}
    H_t = a_t \bar H_t + G_t + \tilde G_t,
\end{equation}
where $a_t \in L^0(\Omega,\F_{t-1},\P;\R)$, $G_t \in F_{\epsilon,t}^\perp \cap E_{0,t}^\perp$, and $\tilde G_t \in F_{\epsilon,t}^\perp \cap E_{0,t}$.
In particular, on $\{ \bar H_t \neq 0 \}$ we have
\begin{equation}
    \label{eq:p-norm_lower_bound}
    a_t = \bar H_t^\ast \cdot H_t \le |H_t|_p\quad {\P\text{-a.s.}}
\end{equation}
In case $p = 2$, the decomposition in \eqref{eq:canonical_decomp} is orthogonal.
\end{lemma}
{
The proof is postponed to Section~\ref{sect:proofs.candec}.  
}

The way of representing $H_t \in \h_t$ as
in \eqref{eq:canonical_decomp} will be referred to as the \emph{canonical decomposition}. 

\begin{remark}
Figure~\ref{fig:dec} is an illustration of the canonical decomposition of elements of $\mathcal H_t$ that is provided in \eqref{eq:canonical_decomp}.
There, $F_{\epsilon,t}^\perp$ can be interpreted as a proper hyperplane which can be further orthogonally decomposed into the spaces $F_{\epsilon,t}^\perp \cap E_{0,t}$ and $F_{\epsilon,t}^\perp \cap E_{0,t}^\perp$.
\begin{figure}[h]
\centering
\begin{tikzpicture}[scale=0.8]
\draw[thick,dashed,->] (-0.3,0,0) -- (2.5,0,0); 
\draw[thick,dashed,->] (0,-0.3,0) -- (0,2,0); 
\draw[thick,dashed,->] (0,0,-0.5) -- (0,0,3); 
\node at (2.65,-0.4,0) {$\bar H_t$};
\draw[thick,->,blue] (1,0,0) -- (1,1,0);
\draw[thick,->,red] (1,0,0) -- (1,0,2);
\node at (0.5,0.1,0) {$\longleftrightarrow$};
\node at (0.5,0.3,0) {$\eps$};
\node at (2.2,1,0) {\color{blue}$F_{\epsilon,t}^\perp \cap E_{0,t}^\perp$};
\node[xshift=1.5em] at (1.2,-0.5,1) {\color{red}$F_{\epsilon,t}^\perp \cap E_{0,t}$};
\draw (1,2,3) -- (1,2,-3) -- (1,-2,-3) -- (1,-2,3) -- cycle;
\node at (1.2,-2.5,0) {$F_{\epsilon,t}^\perp$};
\end{tikzpicture}
\caption{Illustration of the canonical decomposition in \eqref{eq:canonical_decomp}.}
\label{fig:dec}
\end{figure}
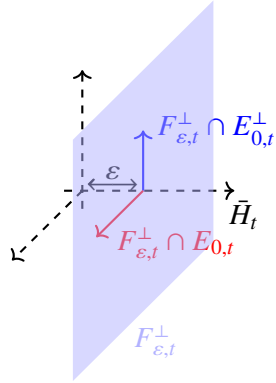
\hfill$\Diamond$
\end{remark}

\begin{lemma}
\label{lem:eps_mart_means}
Let $\Q$ be an $\epsilon$-martingale measure and $t\in \{1,\ldots,T\}$.
Then $\E_\Q[\Delta S_t | \F_{t - 1}] \in F_{\epsilon,t}$ and
\[
H_t \cdot \E_\Q[\Delta S_t | \F_{t - 1}] = 0 \quad {\Q\text{-a.s.\ }} \mbox{for all }H_t \in F_{\epsilon,t}^\perp.
\]
\end{lemma}

\begin{proof}
Since $\Q$ is an $\epsilon$-martingale measure, there is by Remark \ref{rem:eps_mart_arb} absence of strict $\epsilon$-arbitrage.
For any $H_t \in E_{\epsilon,t}$ we have by Hölder's inequality and the fact that $\Q$ is an $\epsilon$-martingale
\[
|\E_\Q[\Delta S_t | \F_{t - 1}]|_q |H_t|_p \ge H_t \cdot \E_\Q[\Delta S_t | \F_{t - 1}] = \epsilon |H_t|_p \ge |\E_\Q[\Delta S_t | \F_{t - 1}]|_q |H_t|_p,
\]
which means that $\E_\Q[\Delta S_t | \F_{t - 1}] \in F_{\epsilon,t}$.
\end{proof}

\subsection{Characterization of no \texorpdfstring{$\epsilon$}{}-arbitrage}
Now we introduce a notion of no $\epsilon$-arbitrage that initially may appear weaker than it really is. 
\begin{definition}\label{def:NA'}
We say a market satisfies $\NA_\epsilon'$
if the following two conditions hold:
\begin{enumerate}[label = (\arabic*)]
\item \label{it:def_NA'_strict} there is no strict $\epsilon$-arbitrage;
\item \label{it:def_NA'_NA} whenever $G \in F_\epsilon^\perp$ is s.t.  $(G\bullet S)_T \ge 0$ {$\P$-a.s.}, then  $(G \bullet S)_T = 0$ {$\P$-a.s}.
\end{enumerate}
\end{definition}
A market satisfies $\NA_\epsilon'$ if, beside not admitting strict $\eps$-arbitrage, it does not admit strict arbitrage (in the usual sense) for a subfamily of strategies (those in $F_\epsilon^\perp$). 
Indeed, the $\eps$-costs related to those strategies can be compensated by implementing a mixed strategy that is obtained by adding a strategy in $E_\eps$. See also argument after Theorem~\ref{thm:duality}.
\begin{example}
We revisit Example \ref{ex:nostrictarb_but_arb} which is again depicted down below in Figure \ref{fig:epsarb_na'.left}.
Here we find that $\bar H_1 \equiv \textrm{e}_1$, so that $E_\eps\equiv\{\alpha \textrm{e}_1\colon \alpha \in \R_+ \}$, and $F_\epsilon^\perp \equiv \{ \alpha \textrm{e}_2 \colon \alpha \in \R \}$.
For $G \equiv (0,1)$, then $(G \bullet S)_1(\omega) = \mathbbm 1_{\{\omega_2\}}(\omega)$, which violates property \ref{it:def_NA'_NA} of Definition \ref{def:NA'}.
\begin{figure}[h]
\centering
\begin{subfigure}{0.25\textwidth}
\centering
\begin{tikzpicture}[scale=0.8]
\draw[thick,dashed,->] (-0.5,0)--(3,0);
\draw[thick,dashed,->] (0,-0.5)--(0,2);
\draw[thick,gray] (1,-0.5) -- (1,2);
\node at (0,0) {$\bullet$};
\node at (1,0) {\color{gray}$\bullet$};
\node at (1,1.5) {\color{gray}$\bullet$};
\node at (-0.3,1.8) {e${}_2$};
\node at (2.8,-0.3) {e${}_1$};
\end{tikzpicture}
\caption{}
\label{fig:epsarb_na'.left}
\end{subfigure}
\qquad
\begin{subfigure}{0.25\textwidth}
\centering
\begin{tikzpicture}[scale=0.8]
\draw[thick,dashed,->] (-0.5,0)--(3,0);
\draw[thick,dashed,->] (0,-0.5)--(0,2);
\draw[thick,gray] (1.1,-0.5) -- (0.6,2);
\draw[thick,dashed] (0,0) -- (1,0.2);
\node at (0,0) {$\bullet$};
\node at (1,0) {\color{gray}$\bullet$};
\node at (0.7,1.5) {\color{gray}$\bullet$};
\draw[<->,>={Straight Barb[scale=0.8]},thick,black] (0,0.2) -- (0.9,0.38) node[midway,xshift=0pt,yshift=6pt] {$d$};
\node at (-0.3,1.8) {e${}_2$};
\node at (2.8,-0.3) {e${}_1$};
\end{tikzpicture}
\caption{}
\label{fig:epsarb_na'.middle}
\end{subfigure}
\qquad
\begin{subfigure}{0.25\textwidth}
\centering
\begin{tikzpicture}[scale=0.8]
\draw[thick,dashed,->] (-0.5,0)--(3,0);
\draw[thick,dashed,->] (0,-0.5)--(0,2);
\draw[thick,gray] (1,-0.5) -- (1,2);
\node at (0,0) {$\bullet$};
\node at (1,-0.3) {\color{gray}$\bullet$};
\node at (1,1.5) {\color{gray}$\bullet$};
\node at (-0.3,1.8) {e${}_2$};
\node at (2.8,-0.3) {e${}_1$};
\end{tikzpicture}
\caption{}
\label{fig:epsarb_na'.right}
\end{subfigure}
\caption{A market that satisfies property \ref{it:def_NA'_strict} but not \ref{it:def_NA'_NA} of Definition~\ref{def:NA'} (Figure \ref{fig:epsarb_na'.left}), and two ways to recover property \ref{it:def_NA'_NA} (Figures \ref{fig:epsarb_na'.middle} and \ref{fig:epsarb_na'.right}).}
\label{fig:epsarb_na'}
\end{figure}
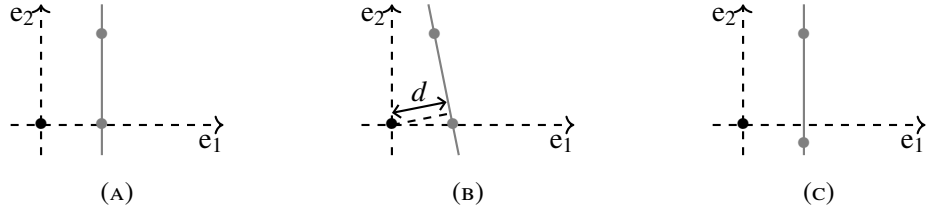
To recover such property, one essentially has two options, which are illustrated in Figure \ref{fig:epsarb_na'.middle} and Figure \ref{fig:epsarb_na'.right}, respectively:
In Figure \ref{fig:epsarb_na'.middle} we move the value of $S_1(\omega_2)$ in the direction of $-\textrm{e}_{1}$, therefore obtaining that $\bar H_1 = 0$, and property \ref{it:def_NA'_NA} of Definition \ref{def:NA'} is trivially satisfied.
In Figure \ref{fig:epsarb_na'.right} we move the value of $S_1(\omega_1)$ in the direction of $-\textrm{e}_2$. 
Again, we find that $F_\epsilon^\perp \equiv \{ \alpha \textrm{e}_2 \colon \alpha \in \R \}$.
Since the second asset, $S^2$, admits no arbitrage in the classical sense, we have for any $G \in F_\epsilon^\perp$ that
\[
(G \bullet S)_1 = (G^2 \bullet S^2)_1 \ge 0 \implies
(G^2 \bullet S^2)_1 = (G \bullet S)_1 = 0,
\]
which is property \ref{it:def_NA'_NA} of Definition \ref{def:NA'}.\hfill$\Diamond$
\end{example} 

By a similar limiting argument as that employed in Example \ref{ex:nostrictarb_but_arb}, one can easily prove that condition $\NA_\epsilon'$ is weaker than $\NA_\epsilon$. To wit, we employ here and throughout the elementary observation:

\begin{remark}\label{lem:square_root_asympt}
For every $x, y \in \mathbb R^d$, the map 
\[
\R \ni k \mapsto k|x|_p - |kx + y|_p
\]
is continuous and by the de L'H\^opital rule
\[
\lim_{k \to \infty} k|x|_p - |kx + y|_p =
\begin{cases}
    -x^\ast \cdot y, & p > 1,
    \\
    -x^\ast \cdot y - \sum_{i = 1, x_i = 0}^d |y_i|, & p = 1,
\end{cases}
\]
where we recall that the definition of $x^\ast$ is given in \eqref{eq:def.dual_vector}.
\hfill$\Diamond$
\end{remark}

\begin{proposition}
\label{cor:NAeps_equivalence}
A market satisfies $\NA_\epsilon$ if and only if it satisfies $\NA_\epsilon'$.
\end{proposition}   

\begin{proof} { 
Assume $\NA_\epsilon$.}
Let $G \in F_\epsilon^\perp$ and define, for $k \in \N$ and $t \in \{1,\ldots, T\}$, $H^k_t := k \bar H_t + G_t$.
By Remark \ref{lem:square_root_asympt} we have
\[k|\bar H_t|_p-|H^k_t|_p \to - \bar H_t^\ast \cdot G_t = 0 ,\]
as $k\to\infty$. 
Hence we can compute
\[
\lim_{k \to \infty} (H^k \bullet S)_T - \epsilon \| H^k \|_p = (G \bullet S)_T + \epsilon \lim_{k \to \infty} (k{ \|\bar H\|_p} - \|H^k\|_p) = (G \bullet S)_T,
\]
from which we deduce that $\NA_\epsilon$ is stronger than $\NA_\epsilon'$.

{ The proof of the reverse implication is more involved, and postponed to Section~\ref{sec:NA_multistep}.}
\end{proof}

What is remarkable here is that while $\NA_\epsilon$ is of asymptotic nature (as it involves limiting procedures), condition $\NA_\epsilon'$ is non-asymptotic. Moreover, $\NA_\epsilon'$ only involves strict arbitrages either in  the classical sense or in the sense of strict $\epsilon$-arbitrages. At a practical level, one of the advantages of condition $\NA_\epsilon'$ is that it will allows us to explore closure properties as in the next result.

\begin{theorem}
\label{thm:closure}
Assume $\NA_\epsilon'$.
Then the closure of $K := \{ (H\bullet S)_T - \epsilon \|H\|_p \colon H \in \mathcal H\}$ w.r.t.\ convergence in probability is given by
\begin{align}
\overline{K} &= \left\{ (H \bullet S)_T - \epsilon \| H \|_p + (G \bullet S)_T \colon H \in \mathcal H, G \in F_\epsilon^\perp, H_t^\ast \cdot G_t = 0\ \forall t=1,\ldots,T \right\}\label{eq:Kbar_equiv1}
\\
&=
\Big\{ (H\bullet S)_T - \epsilon \|H\|_p + \sum_{t = 1}^{T} \mathbbm 1_{\{ H_t = 0 \}} G_t \cdot \Delta S_t \colon H \in \mathcal H, G_t \in F_{\epsilon,t}^\perp\ \forall t =1,\ldots,T  \Big\}.\label{eq:Kbar_equiv2}
\end{align}
Moreover, for a sequence $(Y^k)_{k \in \N}$ in $\overline{K}$ that is a.s.\ bounded from below there is $Y \in \overline{K}$ with $Y \ge \liminf_{k \to \infty} Y^k$.
In particular, all sequences $(Y^k)_{k \in \N}$ in $\overline{K}$ satisfy
\begin{equation*}
\liminf_{k \to \infty} Y^k \ge 0{\quad\P\text{-a.s.}} \implies \liminf_{k \to \infty} Y^k = 0{\quad\P\text{-a.s}}.
\end{equation*}
\end{theorem}
{
The proof is postponed to Section~\ref{sec:NA_multistep}. 
}

With the following example we show that in general $K$ is a proper subset of $\overline{K}$; see Theorem~\ref{thm:K=barK} for sufficient conditions ensuring $K = \overline{K}$.
\begin{example}\label{ex:KbarK}
Let $p > 1$ and consider a single-period market as in Figure \ref{fig:epsarb_na'.right} with $\Omega = \{ \omega_1, \omega_2 \}$, $\P$ the uniform distribution on $\Omega$ and
    \[
        S_0 := (0,0) \quad\mbox{and}\quad
        S_1(\omega) :=
        \begin{cases}
            (\epsilon,1) & \omega = \omega_1, \\
            (\epsilon,-1) & \omega = \omega_2.
        \end{cases}
    \]
Note that we have $\eps=\eps(\P)$ and that
   $K = \{ h \cdot S_1 - \epsilon |h|_p \colon h \in \R^2 \}$ and $\overline{K} = K \cup \{ h \cdot (S_1 - (\epsilon,0)) \colon h \in \R^2 \}$.
    Let $h' = (0,x) \in \R^2$ with $x \neq 0$.
    If $h' \cdot (S_1 - (\epsilon,0))$ were contained in $K$ then there would be $h \in \R^2$ with
    \begin{align*}
        \epsilon h_1 + h_2 S_{1}^{2} - \epsilon |h|_p = h \cdot S_1 - \epsilon |h|_p = x S_1^{2} .
    \end{align*}
    This can only be the case when $h_1 = |h|_p$, which implies $h_2=0$, hence $x S_{1}^{2} = 0$, which contradicts $x \neq 0$. 
    This means that $h' \cdot (S_1 - (\epsilon,0))\in\overline{K}\setminus K$, implying that $K\subsetneq\overline{K}$.
    \hfill $\Diamond$ 
\end{example}

We close this part by observing that one can localize both strict $\epsilon$-arbitrages and $\NA_\epsilon'$. This is important as it motivates the proof strategy behind our main results, wherein we first settle the one-step setting and then move on to the multi-step one, see Section~\ref{sect:proofs}.

\begin{remark}[Bounded strategies and localization]
\label{rem:localization}
It suffices to check the absence of strict $\epsilon$-arbitrage and also $\NA_\epsilon'$ only for bounded strategies. 
Indeed if $H \in \mathcal H$ is such that $(H \bullet S)_T -\epsilon\|H\|_p \ge 0$, then there is $t$ and $\tilde H_t\in\mathcal H_t$ with $|\tilde H_t|_p \leq 1$
and such that $Y:=\tilde H_t\cdot\Delta S_t-\epsilon|\tilde H_t|_p\geq 0$ and $\P(Y>0)>0$. 
Similarly we obtain that $\NA_\epsilon'$ is equivalent to the following statements, for $t \in \{1,\ldots,T\}$:
\begin{enumerate}[label = (\roman*)]
\item all $H_t \in \mathcal H$ with $|H_t|_p \le 1$ satisfy
\[
H_t \cdot \Delta S_t - \epsilon |H_t|_p \ge 0 \implies H_t \cdot \Delta S_t - \epsilon |H_t|_p = 0;
\]
\item all $G_t \in F_{\epsilon,t}^\perp$ with $|G_t|_p \le 1$ satisfy
\[
G_t \cdot \Delta S_t \ge 0 \implies G_t \cdot \Delta S_t = 0.
\]
\end{enumerate}
We provide a proof of these facts in Section~\ref{sec:Nae_one_step}.\hfill$\Diamond$
\end{remark}

\subsection{Duality and price ranges}
Armed with a FTAP, the natural follow-up question concerns (super-)hedging and pricing of derivatives. We explore this in this section.

\begin{theorem}[Pricing-hedging duality]
\label{thm:duality}
Let $\Psi \in L^1(\Omega,\F_T,\P)$ be the payoff of a claim. 
Under $\NA_\epsilon$ we have
\begin{multline}\label{eq:dual}
\sup \left\{ \E_\Q[\Psi] \colon \Q \in \mathcal M_\epsilon(\P)\right\} \\
=
\inf \left\{ x \colon \exists H \in \mathcal H, G \in F_\epsilon^\perp \text{ s.t.\ } x + (H \bullet S)_T + (G \bullet S)_T \ge \Psi + \epsilon \|H\|_p {\;\P\text{-a.s.}} \right\}.
\end{multline}
\end{theorem}
{
The proof is postponed to Section~\ref{sect:proof:ftap}.  
}

How the LHS of the duality relates to the price of $\Psi$ will be clarified in Theorem~\ref{thm:price_range} 
below, see Remark~\ref{rem.ph_duality}. 
{ A simpler version of this duality holds for a specific range of values of $\eps$ and $p$, see Theorem~\ref{thm:simple_duality} below.}
The RHS of the duality can be seen as the smallest initial price to set a self-financing strategy that allows to super-replicate $\Psi$ while covering the costs related to the portfolio. Here the costs are intended as in the interpretation given after Definition~\ref{def:arb}. The fact that there are no costs related to a strategy $G \in F_\epsilon^\perp$ is due to the fact that these can be compensated (rather, reduced to an amount as small as wanted).
Indeed, for any $G \in F_\epsilon^\perp$, there are strategies $(H^k)_k\subseteq \h$ s.t.
\[
(G \bullet S)_T=\lim_{k \to \infty} (H^k \bullet S)_T - \epsilon \| H^k \|_p,
\]
see the proof of {Proposition~
\ref{cor:NAeps_equivalence}}.

\begin{definition}[$\epsilon$-fair price]
Assume the market satisfies $\NA_\epsilon$ and let $\Psi \in L^0(\Omega,\F_T,\P)$.
We say that $\psi \in \R$ is an $\epsilon$-fair price of $\Psi$ if having $\Psi$ at price $\psi$ in the market does not introduce $\epsilon$-arbitrage, that is, if for all sequences $(H^k)_{k \in \N}$ of trading strategies and $(a^k)_{k \in \N}$ of weights in $\R$ we have
\begin{equation}
\label{eq:def_eps_fair_price}
\liminf_{k \to \infty} \left\{ (H^k \bullet S)_T + a^k(\Psi - \psi) - \epsilon \left( \|H^k\|_p + |a^k| \right) \right\} =: Y \ge 0 {\;\P\text{-a.s.}}
\implies
Y = 0{\;\P\text{-a.s}}.
\end{equation}
\end{definition}
Note that the strategies in \eqref{eq:def_eps_fair_price} are of semi-static nature, that is, dynamic in $S$ (as the holdings in the assets can be modified at every time $t=0,1,\ldots,T-1$) and static in the option $\Psi$ (where we have buy-and-hold strategies, as the holdings are fixed at time $0$). The different nature of the holdings in $S$ and $\Psi$ is reflected in the cost as well.

\begin{theorem}[$\epsilon$-fair price range]
\label{thm:price_range}
Assume the market satisfies $\NA_\epsilon$.
Let $\Psi \in L^1(\Omega,\F_T,\P)$ and $\Q \in \mathcal M_\epsilon(\P)$.
Then the interval $[\E_\Q(\Psi) - \epsilon, \E_\Q(\Psi) + \epsilon]$ consists of $\epsilon$-fair prices for $\Psi$.

Viceversa, if $p > 1$ {or $d=1$}, and $\psi \in \R$ is an $\epsilon$-fair price for $\Psi$, then there is $\Q \in \mathcal M_\epsilon(\P)$ with $\psi \in [\E_\Q (\Psi) - \epsilon, \E_\Q (\Psi) + \epsilon]$.
Hence, in this case the set of $\epsilon$-fair prices equals $\cup_{\Q\in\mathcal M_\epsilon(\P) } [\E_\Q (\Psi) - \epsilon, \E_\Q (\Psi) + \epsilon]$.
\end{theorem}
{
The proof is postponed to Section~\ref{sect:proof:ftap}.  
}

\begin{remark}
We remark that the proof provided here for the reverse implication fails for $p = 1$.
This can be explained by the different decompositions that are required in order to deal with $p =1$ compared to $p > 1$:
Given the absence of $\epsilon$-arbitrage, one has for any $\epsilon$-martingale measure $\Q$ that
\[
    \E_\Q[\Delta S_t | \F_{t-1}] = \bar H_t^\ast \quad \mbox{on }\{ \bar H_t \neq 0 \}.
\]
Therefore, the dual vector to $\bar H_t$ determines $\E_\Q[\Delta S_t | \F_{t-1}]$ on $\{ \bar H_t \neq 0 \}$ and the shifted process $\tilde S$ with increments
\[
    \Delta \tilde S_t := \mathbbm 1_{\{ H_t \neq 0 \} } (\Delta S_t - H_t^\ast)
\]
admits no arbitrage in the classical sense.
In general, when $p = 1$ the process $\Delta \tilde S_t$ might allow for arbitrage opportunities as can be seen for the process in Example \ref{ex:nsaem} (that has no $\epsilon$-arbitrage when $p = 1$).

When $p = 1$, a possible avenue to explore the reverse implication is by analogous reasoning as in the proof of the fundamental theorem of asset pricing, that is Theorem \ref{thm:ftap}.
To that end, one can follow the approach from Sections~\ref{sec:Nae_one_step} and \ref{sec:NA_multistep} but one has to use more technical decompositions than the one detailed in Lemma \ref{lem:decomp}.
Since this is rather cumbersome and does not provide the reader with additional insights, we opted to not explore further the case
$p = 1$.\hfill $\Diamond$
\end{remark}

\begin{remark}\label{rem.ph_duality}
From Theorem~\ref{thm:price_range}, the supremum over the $\epsilon$-fair prices of a derivative $\Psi$, denoted by $\overline{p}(\Psi)$, is given by 
\[
\overline{p}(\Psi)=\sup \left\{ \E_\Q[\Psi] \colon \Q \in \mathcal M_\epsilon(\P)\right\}+\eps.
\]
Then Theorem~\ref{thm:duality} implies that
\begin{equation}\label{better_duality}
\overline{p}(\Psi)=\inf \left\{ x \colon \exists H \in \mathcal H, G \in F_\epsilon^\perp \text{ s.t.\ } x + (H \bullet S)_T + (G \bullet S)_T \ge \Psi + \epsilon (\|H\|_p+1) \right\}.
\end{equation}
The RHS in \eqref{better_duality} has the interpretation of the smallest initial price to set a self-financing strategy that allows to super-replicate $\Psi$ while covering the costs related to the whole portfolio, consisting of the dynamic strategies $H$ and $G$ (where again $G$ does not carry any cost), and the static strategy consisting in being short one unit of the derivative $\Psi$ (originating the cost $\eps$). \hfill $\Diamond$
\end{remark}

Unlike in the classical setting, we do not have in Theorem \ref{thm:price_range} that the set of $\epsilon$-fair prices for a non-replicable claim is always open.
We provide a simple example of this phenomenon below.

\begin{example}
\label{ex:price_range}
Consider the following single-period market with two possible events, $\Omega = \{\omega_1, \omega_2\}$, and a single asset $S = S^1$, where
\[
S_0 = 0 
\quad \mbox{and} \quad 
S_1(\omega) = 
    \begin{cases}
        \frac12 \epsilon & \omega = \omega_1, \\
        \frac32 \epsilon & \omega = \omega_2.
    \end{cases}
\]
Consider the filtration generated by $S$.
Clearly, then the set of $\epsilon$-martingale measures is given by $\{ \alpha \delta_{\{\omega_1\}} + (1 - \alpha) \delta_{\{\omega_2\}} \colon \alpha \in [\frac12,1) \}$.
We want to compute the $\epsilon$-fair price range of the following option:
\[
\Psi(\omega) :=
    \begin{cases}
        0 & \omega = \omega_1, \\
        1 & \omega = \omega_2.
    \end{cases}
\]
Since \[
\left\{ \E_\Q[\Psi] \colon \Q \mbox{ is an  $\epsilon$-martingale measure}\right\} = (0,1/2],
\]
by Theorem~\ref{thm:price_range} we conclude that the $\epsilon$-fair price range is given by the half-open interval $(-\epsilon,1/2+\epsilon]$.
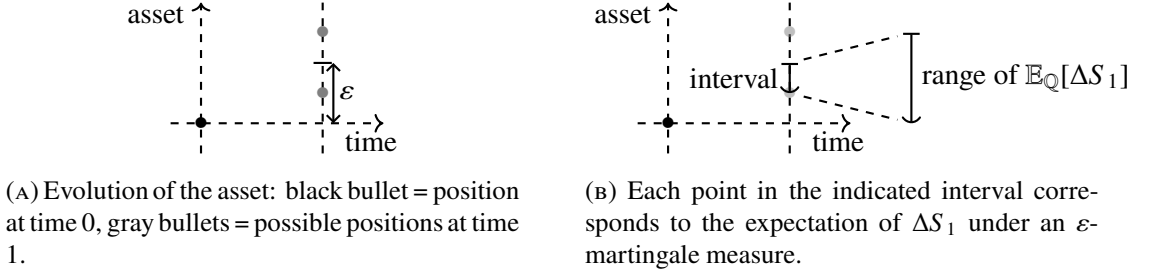
\begin{figure}[h]
\centering
\begin{subfigure}{0.45\textwidth}
\centering
\begin{tikzpicture}[scale=0.8]
\draw[thick,dashed,->] (-0.5,0)--(3,0);
\draw[thick,dashed,->] (0,-0.5)--(0,2);
\draw[thick,dashed] (2,-0.5) -- (2,2);
\node at (0,0) {$\bullet$};
\node at (2,0.5) {\color{gray}$\bullet$};
\node at (2,1.5) {\color{gray}$\bullet$};
\node[xshift=-10pt] at (-0.3,1.8) {asset};
\node at (2.8,-0.3) {time};
\draw[<->,>={Straight Barb[scale=0.8]},thick,black,xshift=5pt]
(2,0) -- (2,1) node [black,midway,xshift=5pt] {$\epsilon$};
\draw[thick] (1.85,1) -- (2.15,1);
\end{tikzpicture}
\caption{Evolution of the asset: black bullet = position at time $0$, gray bullets = possible positions at time $1$.}
\label{fig:price_range.left}
\end{subfigure}
\qquad
\begin{subfigure}{0.45\textwidth}
\centering
\begin{tikzpicture}[scale=0.8]
\draw[thick,dashed,->] (-0.5,0)--(3,0);
\draw[thick,dashed,->] (0,-0.5)--(0,2);
\draw[thick,dashed] (2,-0.5) -- (2,2);
\node at (0,0) {$\bullet$};
\node[opacity=0.5] at (2,0.5) {\color{gray}$\bullet$};
\node[opacity=0.5] at (2,1.5) {\color{gray}$\bullet$};
\node[xshift=-10pt] at (-0.3,1.8) {asset};
\node at (2.8,-0.3) {time};
\draw[(-|,thick,black] (2,0.5) -- (2,1) node[midway,left] {interval};
\draw[thick,dashed] ($($(2,0.5)!.2!(4,0)$)!5pt!(2,0.5)$)--($($(4,0)!.2!(2,0.5)$)!5pt!(4,0)$);
\draw[thick,dashed] ($($(2,1)!.2!(4,1.5)$)!5pt!(2,1)$)--($($(4,1.5)!.2!(2,1)$)!5pt!(4,1.5)$);
\draw[(-|,thick,black] (4,0) -- (4,1.5) node[midway,right,align=left] {range of $
\E_\Q[\Delta S_1] $};
\end{tikzpicture}
\caption{Each point in the indicated interval corresponds to the expectation of $\Delta S_1$ under an $\epsilon$-martingale measure.}
\label{fig:price_range_right}
\end{subfigure}
\caption{Illustration of Example \ref{ex:price_range}.}
\label{fig:price_range}
\end{figure}
\hfill$\Diamond$
\end{example}

\subsection{When strict \texorpdfstring{$\eps$}{}-arbitrage is enough}\label{sect:crit}
In this section we will see how the results presented above take a simpler form when $\eps$ is strictly greater than the critical value $\epsilon(\P)$ or when $p=1$. Indeed, we will prove that, in either case, one does not need to consider sequences of strategies as in \eqref{eq:def_epsilon_arbitrage} (for $\eps$-arbitrage), but one can directly work with the concept of strict $\eps$-arbitrage.

We start by showing that, when $\epsilon$ is strictly bigger than $\epsilon(\P)$, then the sets $E_{\epsilon,t}$ and $F_{\epsilon,t}^\perp$ are both trivial.
This will significantly facilitate the proofs and influence the hedging and continuity results.
\begin{lemma}\label{lem:critical_value}
Let $\epsilon > \epsilon(\P)$.
Then for all $H \in \mathcal H$
\begin{equation*}\label{eq:lem.critical_value.assertion}
(H\bullet S)_T - \epsilon \|H\|_p = 0 \implies  H \equiv 0.
\end{equation*}
In particular, for all $t=1,\ldots,T$ we have $\bar H_t\equiv 0$, so that $E_{\epsilon,t}=\{0\} = F_{\epsilon,t}^\perp$.
\end{lemma}

\begin{proof}
Let $0 \not\equiv H \in \mathcal H$ with $(H\bullet S)_T - \epsilon \|H\|_p = 0$.
For $\epsilon' \in (\epsilon(\P),\epsilon)$, we have
\[
(H\bullet S)_T - \epsilon' \|H\|_p \ge (H\bullet S)_T - \epsilon \|H\|_p = 0,
\]
where the inequality is strict on the set $\{ H \neq 0 \}$ which has, by assumption, positive probability.
Hence, $H$ is a strict $\epsilon'$-arbitrage which is a contradiction to our choice of $\epsilon'$.
\end{proof}

The next results establishes that when $p = 1$ the property $\NA'_\epsilon$ and the absence of strict $\epsilon$-arbitrage coincide.

\begin{lemma}\label{lem:p=1.strict_G}
Assume the absence of strict $\epsilon$-arbitrage.
Let $p = 1$ and $G \in F_\epsilon^\perp$.
Then there is $H \in \mathcal H$ such that
\begin{equation*}
    \label{eq:p=1.GtoH}
    (G \bullet S)_T = (H \bullet S)_T - \epsilon \|H \|_1.
\end{equation*}
\end{lemma}

\begin{proof}
Recall that by \eqref{eq:def.dual_vector}, when $p = 1$, the dual vector $\bar H_t^\ast$ is coordinate-wise given by $\bar H^\ast_{t,i} = \sgn(\bar H_{t,i})$.
Let $G \in F_\epsilon^\perp$, that is, for every $t = 1,\ldots, T$,
\[
    G_t \cdot \bar H^\ast_t = \sum_{i = 1}^d G_{t,i} \sgn(\bar H_{t,i}) = 0,
\]
and $\{ G_{t,i} \neq 0\} \subseteq \{ \bar H_{t,i} \neq 0 \}$ for $i = 1,\ldots, d$.
Define the weighting functions $w_t \in L^0(\Omega,\F_{t-1},\P; \R_+)$ by 
\[
    w_t := 1 +
    \max\left\{ \frac{|{ G}_{t,i}|}{|\bar H_{t,i}|} : i\in\{1,\ldots,d\} \text{ s.t. }  \bar H_{t,i} \neq 0\right\}
    \quad \mbox{for }t=1,\ldots,T,
\]
with the convention $\max\emptyset=0$.
Set $H \in \mathcal H$ as the strategy given by $H_t := G_t + w_t \bar H_t$, $t = 1,\ldots,T$, and observe that $\sgn(H_{t,i}) = \sgn(\bar H_{t,i})$.
Therefore, we compute
\begin{align*}
    (H \bullet S)_T - \epsilon \| H\|_1 &=
    (G \bullet S)_T + \epsilon \left( \sum_{t = 1}^T w_t |\bar H_t|_1 - |w_t \bar H_t + G_t|_1 \right) \\
    &=
    (G \bullet S)_T + \epsilon \left( \sum_{t = 1}^T w_t |\bar H_t|_1 - \sum_{i = 1}^d w_t |\bar H_{t,i}| + \sgn(\bar H_{t,i}) G_{t,i}\right) \\
    &= (G \bullet S)_T + \epsilon \sum_{t = 1}^T G_t \cdot \bar H_t^\ast = (G \bullet S)_T. \qedhere
\end{align*}
\end{proof}

\begin{theorem} \label{thm:K=barK}
If $\eps>\eps(\P)$ or $p=1$, then no strict $\eps$-arbitrage is equivalent to no $\eps$-arbitrage, i.e. to $\NA_\eps$.
In either case, $K = \overline{K}$.
\end{theorem}
\begin{proof}
We first consider the case $p=1$. By Lemma~\ref{lem:p=1.strict_G} we have equivalence between $\NA'_\epsilon$ and the absence of strict $\epsilon$-arbitrage. Thus also equivalence to $\NA_\epsilon$, by Proposition \ref{cor:NAeps_equivalence}. 
Finally, Lemma~\ref{lem:p=1.strict_G} and  Equation~\eqref{eq:Kbar_equiv2} in Theorem~\ref{thm:closure} imply $K = \overline K$.

Next we consider the case $\eps>\eps(\P)$. 
By Lemma~\ref{lem:critical_value} we know that  $F_{\epsilon,t}^\perp =\{0\}$.
This makes condition \ref{it:def_NA'_NA} in Definition~\ref{def:NA'} void, which in turn implies equivalence between absence of strict $\eps$-arbitrage and $\NA'_\eps$ as well as $K = \overline K$, by Theorem~\ref{thm:closure}.
Finally, Proposition~\ref{cor:NAeps_equivalence} yields the desired equivalence.
\end{proof}

In the special cases considered in this subsection, the duality result \eqref{eq:dual} takes a simpler form, as in the classical super-replication theorem.
\begin{theorem}\label{thm:simple_duality}
If $\eps>\eps(\P)$ or $p=1$, then 
under $\NA_\epsilon$ we have
\begin{align*}\label{eq:dual_easy}
\sup \left\{ \E_\Q[\Psi] \colon \Q \in \mathcal M_\epsilon(\P)\right\} =
\inf \left\{ x \colon \exists H \in \mathcal H \text{ s.t.\ } x + (H \bullet S)_T \ge \Psi + \epsilon \|H\|_p\ {\P\text{-a.s.}} \right\}.
\end{align*}
\end{theorem}

\begin{proof}
The result follows directly from Theorem \ref{thm:duality} and Theorem \ref{thm:K=barK} above.
\end{proof}

\section{Continuity results}\label{sect:cont}

The goal of this section is to connect the notion of $\eps$-arbitrage and $\eps$-martingale measure to (adapted) distances, and thereby obtain a notion of stability of arbitrage. 

If we consider the classical notion of arbitrage and a naive choice of distance between models (let us call it $D$ for the time being), then, given a model $\P$ with arbitrage and any $\delta>0$, it should be easy to build a model $\P'$ which is free of arbitrage and s.t.\ $D(\P,\P')<\delta$. This is certainly the case if $D$ is equal to the usual (adapted) Wasserstein distance. Thus, a naive choice of distance $D$ will not relate well to the concept of arbitrage. By generalizing (``quantifying") the notion of arbitrage to that of $\epsilon$-arbitrage, as we have done, and introducing a very special candidate for $D$, we will fix this undesirable fact and get a continuity of sorts for the ``amount of arbitrage" w.r.t.\ this new distance. {We refer the reader to \cite{Vi03} for an introduction to optimal transport and classical Wasserstein distances, and to \cite{BaBeLiZa17,BaBaBeEd19a} for their adapted counterparts.}

Let $\overline{\Omega}:= \R^{d(T+1)}$, which we see as the path space of $\R^d$-valued stochastic processes indexed by $t=0,1,\dots, T$.  We also see $\overline{\Omega}\times \overline{\Omega}$ as the path space of pairs of $\R^d$-valued stochastic processes in $T+1$ time steps. The canonical process on $\overline{\Omega}\times \overline{\Omega}$ is denoted by $(X,Y)$. Hence both $X$ and $Y$ will be used to denote the canonical process on $\overline{\Omega}$. If $Z=\{Z_r\}_{r=0}^T$ is a stochastic process, and $0\leq s<t\leq T$, we denote $Z_{s:t}:=\{Z_s,Z_{s+1},\dots, Z_t \}$ and we convene that $\Delta Z_0:= Z_0$. {We denote by $(\mathcal F^{X,\P}_t)_t$ the $\P$-completion of the natural filtration of $X$, with similar notation for $Y$.}

\begin{definition}\label{def:cpl}
For $\P,\P'$ probability measures on $\overline{\Omega}$, we denote by $\Pi(\P,\P')$ the set of measures  $\pi$ on $\overline{\Omega}\times \overline{\Omega}$ with $X\sim_\pi \P$ and $Y\sim_\pi \P'$. We say that $\pi\in \Pi(\P,\P')$ is causal from $X$ to $Y$ (resp.\ from $Y$ to $X$) if for each $t$ the law of $Y_t$ given $X_{0:T}$ under $\pi$ is $\mathcal F^{X,\P}_{t}$-measurable (resp.\ for each $t$ the law of $X_t$ given $Y_{0:T}$ under $\pi$ is $\mathcal F^{Y,\P'}_{t}$-measurable). We denote by $\Pi^{bc}(\P,\P')$ the set of those $\pi\in \Pi(\P,\P')$ which are causal both from $X$ to $Y$ and from $Y$ to $X$.
\end{definition}

{
In order to help the intuition of the reader: suppose that $\pi^T:= \text{Law}_{\P}(X,T(X))$ where $T:\overline{\Omega}\to\overline{\Omega}$ is such that $T(\P)=\P'$, i.e.\ the law of $T(X)$ under $\P$ is equal to $\P'$. Then  $\pi^T\in \Pi(\P,\P')$. If additionally $T$ is an adapted map, meaning that the t-th coordinate of $T(x)$ only depends on $x_0,\dots,x_t$, then  $\pi^T$ is causal from $X$ to $Y$. Definition \ref{def:cpl} is the logical generalization of these ideas.

We now move on tho the object that will be crucial for our undertakings.}

\begin{definition}\label{def:AWinftyq}
We define the adapted $L^\infty$-distance between $\P$ and $\P'$ as
\[
\textstyle{\AW^\Delta_{\infty,q}(\P,\P'):=\inf_{\pi\in\Pi^{bc}(\P,\P')}\esssup_\pi \sum_{t=0}^{T}|\Delta X_t-\Delta Y_t|_q}.
\]
\end{definition}

{To be precise, $\AW^\Delta_{\infty,q}(\P,\P')$ is only a distance on the space of Borel probability measures on $\overline{\Omega}$ which have a bounded support. The proof of this follows closely the classical case (e.g.\ \cite[Ch.7]{Vi03}) so we omit it. Nevertheless $\AW^\Delta_{\infty,q}(\P,\P')$ is always defined, taking values on the extended positive half-line, and it may even happen that $\AW^\Delta_{\infty,q}(\P,\P')<\infty$ without $\P$ or $\P'$ having a bounded support. The latter happens e.g.\ if one measure is a translation of the other, or if the two measures agree outside of a compact set. More important to our arguments is the existence of optimizers for such a transport problem, and this we establish now:}

\begin{lemma}\label{lem:existence_AW_opt}
The infimum in the definition of $\AW^\Delta_{\infty,q}(\P,\P')$ is attained.
\end{lemma}

\begin{proof}
Define $S(\pi):=\esssup_\pi \sum_{t=0}^{T}|\Delta X_t-\Delta Y_t|_q = \inf\{\lambda : \sum_{t=0}^{T}|\Delta X_t-\Delta Y_t|_q\leq \lambda \;\ \pi\text{-a.s.}\} $. Since $\Pi^{bc}(\P,\P')$ is a {weakly compact set by \cite[Lemma A.1]{AcBaJi21}, we only need to prove that $S$ is l.s.c.\ with respect to weak convergence\footnote{We recall that weak convergence of measures refers to convergence of all integrals where the integrands are continuous bounded functions.}.} For this, consider $\pi_n\to\pi$ and assume w.l.g.\ that 
$\liminf_n S(\pi_n)=\lim_n S(\pi_n)=:S<\infty$. If $|\{n:S(\pi_n)\leq S\}|=\infty$, there is nothing to prove since the set $\left\{\sum_{t=0}^{T}|\Delta X_t-\Delta Y_t|_q\leq S\right\}$ is closed and of $\pi_n$-measure 1 for infinitely many $n$'s, hence $S(\pi)\leq S$. So we may assume w.l.o.g.\ $S(\pi_n)>S$ and $S(\pi_n)\searrow S$. For $n\geq \ell$ we have $\pi_n\left(\left\{\sum_{t=0}^{T}|\Delta X_t-\Delta Y_t|_q\leq S(\pi_\ell)\right\}\right)=1$ and so taking limit in $n$ we have $\pi\left(\left\{\sum_{t=0}^{T}|\Delta X_t-\Delta Y_t|_q\leq S(\pi_\ell)\right\}\right)=1$, from where $\pi\left(\left\{\sum_{t=0}^{T}|\Delta X_t-\Delta Y_t|_q\leq S\right\}\right)=1$ follows.
\end{proof}

The next observation is crucial to all proofs in this section:

\begin{remark}
\label{rem:AW_inf}
Take $\Q\in\mathcal M_\epsilon(\P) $ and let $\pi\in\Pi^{bc}(\P,\P')$ be optimal for $\AW^\Delta_{\infty,q}(\P,\P')$, whose existence is guaranteed by Lemma~\ref{lem:existence_AW_opt}. Defining
\[
\hat\pi(dx,dy):= \pi^x(dy) \, \Q(dx),
\]
we have: $\hat\pi$ is causal from $X$ to $Y$ and $\hat\pi\sim\pi$. Indeed, causality is trivial, while on the other hand, as $\Q\sim \P$ and as $\hat\pi$ and $\pi$ share the same conditional kernels, then $\hat\pi\sim\pi$.
\hfill$\Diamond$
\end{remark}

{
 We now can present the connection between adapted $L^\infty$-distances and arbitrage. Towards this end we assume that the stock price process $S$ is the canonical process on $\overline{\Omega}$, i.e.\ either $X$ or $Y$ in the notation of this section, and that the filtration is the canonical one. This canonical setting, usual in every application of optimal transport, is not much of a restriction: if the original filtration is the one given by the stock prices, then one can always transfer the original problem into one in the canonical setting. See Remark \ref{rem:filtered} below for a short discussion on how to bypass the canonical framework and permit very general filtrations.
}

The following proposition establishes that, as long as reference measures are close w.r.t.\ the $\AW^\Delta_\infty$-distance, then the ``amount of arbitrage" that one can get in the two markets is similar. This is expressed in terms of $\eps$-martingale measures, thanks to Theorem~\ref{thm:ftap}.
\begin{proposition}\label{prop:cont}
We have
\begin{equation*}
    \mathcal M_\epsilon(\P)\neq\emptyset \implies \mathcal M_{\epsilon+\AW^\Delta_{\infty,q}(\P,\P')}(\P')\neq\emptyset.
\end{equation*}
\end{proposition}

\begin{proof}
Take $\Q\in\mathcal M_\epsilon(\P) $ and $\hat\pi,\pi$ as in Remark~\ref{rem:AW_inf}. Denote
 by $\Q'$ the second marginal of $\hat\pi$. As $\hat\pi\sim\pi$, then in particular $\Q'\sim \P'$. Since $\hat\pi$ is causal from $X$ to $Y$, we have 
 for any $H\in  \mathcal{H}$, that
$\hat H_t(X):=\E_{\hat\pi} [H_t(Y)|X]=\E_{\hat\pi} [H_t(Y)|X_{0:t-1}]$. Thus 
 $\hat H\in\mathcal{H}$. Hence
\begin{align}\notag
\E_{\Q'}\left[(H\bullet Y)_T-\eps\| H(Y)\|_p\right]&=\textstyle{\E_{\hat\pi}\left[\sum_{t=1}^T H_t(Y)\Delta Y_t-\eps\| H(Y)\|\right]}\\ \notag
&= \textstyle{\E_{\hat\pi}\left[\sum_{t=1}^T H_t(Y)\Delta X_t\right]
-\eps\E_{\hat\pi}[\| H(Y)\|_p]
+\E_{\hat\pi}\left[\sum_{t=1}^T H_t(Y)(\Delta Y_t-\Delta X_t)\right]}\\ \notag
&\leq \textstyle{\E_{\Q}\left[\sum_{t=1}^T \hat H_t(X)\Delta X_t\right]
-\eps\E_{\hat\pi}[\| H(Y)\|_p]
+\E_{\hat\pi}\left[\sum_{t=1}^T |H_t(Y)|_p |\Delta Y_t-\Delta X_t|_q\right]}\\ \notag
&\leq \textstyle{\eps\E_{\hat\pi}[\| \hat H(X)\|_p-\| H(Y)\|_p]+\E_{\hat\pi}[\| H(Y)\|_p]\esssup_{\hat\pi}\sum_{t=1}^T |\Delta Y_t-\Delta X_t|_q}\\ \notag
&\leq \textstyle{\E_{\Q'}[\| H\|_p]\esssup_{\pi}\sum_{t=1}^T |\Delta Y_t-\Delta X_t|_q}\\
&\leq \textstyle{\E_{\Q'}[\| H\|_p]\AW^\Delta_{\infty,q}(\P,\P')}. \label{eq:non-sharp-step}
\end{align}
Hence $\E_{\Q'}[(H\bullet Y)_T-(\eps+\AW^\Delta_{\infty,q}(\P,\P') )\| H(Y)\|_p]\leq 0$.
\end{proof}
    
In the next example we show how the continuity result in
Proposition \ref{prop:cont} fails when dropping the bicausality constraint in Definition \ref{def:AWinftyq}. 
\begin{example}Take $T=1=d$, $\epsilon\geq 0$,  and consider $\P_\epsilon$ such that $(X_0,X_1)=(\epsilon,1)$ or $(X_0,X_1)=(-\epsilon,-1)$ with equal probability.
Notice that  $\AW^\Delta_{\infty,q}(\P_0,\P_\epsilon)=2$. On the other hand, if in Definition \ref{def:AWinftyq} we would drop the bicausality constraint on the minimization problem between $\P:=\P_0$ and $\P':=\P_\epsilon$, then the optimal value would be $2\epsilon$. Remark now that $\P_0$ is a martingale law, so $\mathcal M_0(\P_0)\neq\emptyset$. On the other hand, $\P_\epsilon$ admits strict $\delta$-arbitrages for any $\delta<1-\epsilon$, hence $\mathcal M_\delta(\P_\epsilon)=\emptyset$ for such values. In particular, if $\epsilon<1/3$ we have $\mathcal M_{2\epsilon}(\P_\epsilon)=\emptyset$, showing that Proposition \ref{prop:cont} fails if we drop bicausality from Definition \ref{def:AWinftyq}. \hfill$\Diamond$
\end{example}
    
Another way to see the continuity of the ``amount of arbitrage" w.r.t.\ the $\AW^\Delta_{\infty,q}$-distance is via the concept of maximal arbitrage reachable in a given market. See Definition~\ref{def:critical_value} and Corollary~\ref{corol:crit} for a reformulation in terms of approximate martingales.

\begin{proposition}\label{lemma.conteps}
We have
\[
|\epsilon(\P)-\epsilon(\P')|\leq \AW^\Delta_{\infty,q}(\P,\P').
\]
\end{proposition}

\begin{proof}
Take $\pi\in\Pi^{bc}(\P,\P')$ optimal for $\AW^\Delta_{\infty,q}(\P,\P')$.
For any $m\in\N$, fix $\Q_m\in \mathcal M_{\epsilon(\P)+1/m}(\P)$, and define the coupling $\hat\pi_{m}(dx,dy):=\pi^x(dy)\Q_m(dx)$, which is causal from $X$ to $Y$, and denote its second marginal  by $\Q'_{m}$.
As in the proof of Proposition~\ref{prop:cont}, from $\Q_m\sim \P$ we also have $\hat\pi_{m} \sim \pi$.  
Now, for any $H\in\mathcal{H}$, let $H^m_t(X):=\E_{\hat\pi_{m}} [H_t(Y)|X]=\E_{\hat\pi_{m}} [H_t(Y)|X_{0:t-1}]$. Then, similarly to the proof of Proposition~\ref{prop:cont}, we find that
\[
\E_{\Q'_{m}}\left[(H\bullet S)_T\right]
\leq 
\textstyle{\left(\epsilon(\P)+1/m+\AW^\Delta_{\infty,q}(\P,\P')\right)\E_{\Q'_{m}}[\|H\|_p]
}.
\]
Since the above holds for arbitrary $m\in\N$, and $\Q'_m\sim\P'$, we deduce that $\epsilon(\P')\leq \epsilon(\P)+\AW^\Delta_{\infty,q}(\P,\P')$.

\end{proof}

\begin{proposition}\label{prop:stab_M}
Let $\Psi$ be $L$-Lipschitz, in the sense that $|\Psi(y)-\Psi(x)|\leq L\|\Delta y-\Delta x\|_q$,  and $\Psi\in L^1(\Omega,\F,\P)$. Assume further that $p>1$. We have: If $\psi$ is an $\eps$-fair price for $\Psi$ under $\P$, then $\psi$ is also an $(\eps+L\,\AW^\Delta_{\infty,q}(\P,\P'))$-fair price for $\Psi$ under $\P'$.
\end{proposition}

\begin{proof}
Per Theorem~\ref{thm:price_range} there is $\Q\in\M_\eps(\P)$ s.t.\ $|\psi-\E_\Q[\Psi]|\leq \eps$. By Remark~\ref{rem:AW_inf} and Proposition~\ref{prop:cont}, we can find $\Q'\in\M_{\eps+\AW^\Delta_{\infty,q}(\P,\P')}(\P')$, and further $\Q'$ is the second marginal of a coupling $\hat \pi$ whose first marginal is $\Q$ and s.t.\ $\hat\pi $ is equivalent to the optimizer $\pi$ of $\AW^\Delta_{\infty,q}(\P,\P')$. Thus
\begin{align*}|\E_{\Q'}[\Psi]-\E_{\Q}[\Psi]|=\E_{\hat \pi}[\Psi(Y)-\Psi(X)]\leq L\, \E_{\hat \pi}[\|\Delta Y-\Delta X\|_q]]&\leq L\,\esssup_{\hat \pi}\|\Delta Y-\Delta X\|_q \\ &= L\,\esssup_{\pi}\|\Delta Y-\Delta X\|_q \\&= L\,\AW^\Delta_{\infty,q}(\P,\P'). \end{align*}
 Hence
\[|\E_{\Q'}[\Psi]-\psi|\leq |\E_{\Q'}[\Psi]-\E_{\Q}[\Psi]|+|\E_{\Q}[\Psi]-\psi|\leq \eps+L\,\AW^\Delta_{\infty,q}(\P,\P'). \qedhere \]
\end{proof}

\begin{remark}
Both proofs of Propositions \ref{prop:cont} and \ref{lemma.conteps} can be made sharper if we drop from the definition of $\AW^\Delta_\infty(\P,\P')$ the term  with $t=0$ in the sum defining the cost function. See for instance the ``generous'' bound in \eqref{eq:non-sharp-step}. As a result, the statements of Propositions \ref{prop:cont} and \ref{lemma.conteps} can be made sharper by dropping $t=0$ from $\AW^\Delta_\infty(\P,\P')$. We do not go this way, since  e.g.\  in Proposition~\ref{prop:stab_M} it would make little sense to require the function $\Psi$ not to depend on the initial number of shares. \hfill$\Diamond$
\end{remark}

The following example shows how adapted Wasserstein distances (and hence Wasserstein distances) are not appropriate for arbitrage quantification for costs other than the $L^\infty$ one.
\begin{example}
Consider the simple setting of one asset ($d=1$) and one time step ($T=1$). For any $\delta>0$  arbitrarily small and $M>0$ arbitrarily big, we define $\P$ and $\P'$ as follows: Let $\P\in\mathcal{P}(\overline\Omega)$ be such that the canonical process $X=(X_0,X_1)$ satisfies $\P$-a.s.\ $X_0=0$ and $X_1=M$. Then clearly the market with asset $S:=X$ satisfies $\eps(\P)=M$. On the other hand, let $\P'\in\mathcal{P}(\overline\Omega)$ such that $X_0=0$ $\P'$-a.s., while $\P'(X_1=M)=1-p_\delta$ and $\P'(X_1=-\delta)=p_\delta$, where $p_\delta=\delta/(M+\delta)^r$. Then for $S:=X$ we have $\eps(\P')=0$. If we now consider $(\AW^\Delta_{r,q}(\P,\P'))^r$ defined as in Definition \ref{def:AWinftyq} but with the essential supremum replaced by $\E_\pi[(\dots)^r]$, then $\AW^\Delta_{r,q}(\P,\P')=\delta^{1/r}$. Thus $\P$ and $\P'$ are arbitrarily close for $\AW^\Delta_{r,q}$ and yet $\epsilon(\P)-\epsilon(\P')$ is arbitrarily large. Notice that here $\AW^\Delta_{\infty,q}(\P,\P')=M+\delta$.
~\hfill$\Diamond$
\end{example}

\begin{remark}\label{rem:filtered}
A stochastic process $S$ on $(\Omega,\F,\P)$ induces a measure on $\R^{dT}$ by taking its law under $\P$. The definitions in this section would then be applied to measures $\P,\P'$ being actually such image measures. Doing this, however, ignores the original filtration to which $S$ was supposed to be adapted. It is possible to directly define the concepts in this part (i.e.\ couplings, causality, and adapted distances) directly at the level of stochastic processes on a filtered probability space, without the use of image measures. See \cite{BaBePa21}  for a thorough discussion, albeit not for $L^\infty$ costs. This remark also pertains Section \ref{sec:approx_Laplace} below. \hfill$\Diamond$
\end{remark}

\section{Regularization by the Laplace method}\label{sect:reg}
We continue here with the canonical setting of Section \ref{sect:cont}, with the same notation. 
The $\infty$-Wasserstein distance was introduced in \cite{ChdPJu08}. In mathematical finance this distance has been applied in \cite{GeGu19,GeGu20,ObWi18}. If $\P,\P'\in\mathcal P(\mathbb R^{d(T+1)})$, then their $\infty$-Wasserstein distance is given by
\begin{definition} The $\infty$-Wasserstein distance between $\P$ and $\P'$ is given by

\[
\textstyle{\mathcal W_{\infty,q}(\P,\P'):=\inf_{\pi\in\Pi(\P,\P')}\esssup_\pi \|X- Y\|_q},
\]
and its counterpart in terms of bicausal couplings is
\[
\textstyle{\AW_{\infty,q}(\P,\P'):=\inf_{\pi\in\Pi^{bc}(\P,\P')}\esssup_\pi \sum_{t=0}^{T}| X_t- Y_t|_q}=\inf_{\pi\in\Pi^{bc}(\P,\P')}\esssup_\pi \| X- Y\|_q.
\]
\end{definition}

In this part we will provide a brief analysis of $\AW_{\infty,q}$. We decided to do so for $\AW_{\infty,q}$, rather than for $\AW^\Delta_\infty$ introduced in the previous part, as the former is more closely related to existing objects of interest such as $\mathcal W_{\infty,q}$. Nevertheless, 
we have the bounds\footnote{{This follows from the  inequalities $|\Delta z_t|_q + |\Delta z_{t+1}|_q \leq |z_{t-1}|_q +2|z_{t}|_q + |z_{t+1}|_q $ and $|z_t|_q =|\sum_{s=0}^t \Delta z_s|_q\leq \sum_{s=0}^t |\Delta z_s|_q $.}}
$$\frac{1}{T}\AW_{\infty,q}(\P,\P')\leq \AW^\Delta_{\infty,q}(\P,\P') \leq 2\AW_{\infty,q}(\P,\P'),$$
meaning that all metric properties in this part and in Section \ref{sect:cont} are shared between these two objects. In fact more is true, since everything we do in this part carries over when replacing the processes $X,Y$ by the increment processes $\Delta X,\Delta Y$.

\subsection{Approximation by adapted log-exponential divergences}\label{sec:approx_Laplace}

One way to appreciate the definitions of $\AW_{\infty,q}$ and $\mathcal W_{\infty,q}$, which we will find insightful, is by introducing an exponential version of these object. Observe that $\mathcal W_{\infty,q} = \AW_{\infty,q}$ if $T=1$, so in this sense it suffices to focus our discussion on $\AW_{\infty,q}$, as we will indeed do.

\begin{definition}
We introduce the adapted log-exponential divergence between $\P$ and $\P'$ as follows:
\[
E^q_\lambda(\P,\P'):= \frac{1}{\lambda}\log \inf_{\pi\in\Pi^{bc}(\P,\P')} \E_\pi[\exp(\lambda \|X-Y \|_q)].
\]
\end{definition}
{Remark that $E^q_\lambda(\P,\P')$ is always well-defined, taking values on the extended positive half-line.} We recall the so-called Laplace method, and the Gibbs variational principle, 
both of which will play a role in what follows. We use the notation $H(\tilde\pi |\pi)$ for the relative entropy of $\tilde\pi$ w.r.t. $\pi$, which is given by
\[
\E_{\tilde \pi}[\log(d\tilde\pi / d\pi)]\; \text{if $\tilde \pi\ll\pi$,\; and\, $+\infty $\, otherwise}.
\]

\begin{lemma}\label{lem:LapGib}
\begin{enumerate}
\item 
We have  $$\esssup_\pi \|X- Y\|_q=\lim_{\lambda \to \infty}\frac{1}{\lambda}\log  \E_\pi[\exp(\lambda \|X-Y \|_q)] ;$$
\item we have
$$\frac{1}{\lambda}\log  \E_\pi[\exp(\lambda \|X-Y \|_q) = \sup_{\tilde \pi \sim \pi}\left\{\E_{\tilde\pi}[\|X-Y \|_q] - \frac{1}{\lambda}H(\tilde \pi| \pi)\right\}.$$
\end{enumerate}
\end{lemma}
{ Note that item $(2)$ can be seen from the dual formulation of the entropic risk measure, see \cite[Example 4.34]{FoSc16}.}
\begin{proof}
For the second statement, we first apply the Gibbs variational principle to $\|X-Y\|_q\wedge m$ instead of $\|X-Y\|_q$, the former being continuous and bounded, so obtaining 
$$\frac{1}{\lambda}\log  \E_\pi[\exp(\lambda \|X-Y \|_q\wedge m) = \sup_{\tilde \pi \sim \pi}\left\{\E_{\tilde\pi}[\|X-Y \|_q\wedge m] - \frac{1}{\lambda}H(\tilde \pi| \pi)\right\}. $$
Then we send $m\to\infty$ applying upwards monotone convergence on both sides.

For the first claim,  we initially assume that $S(\pi):=\esssup_\pi \|X- Y\|_q$ is finite. Then the convergence 
$$\lim_{\lambda\to \infty}  \frac{1}{\lambda}\log  \E_\pi[\exp(\lambda \|X-Y \|_q)] = S(\pi), $$
is equivalent to
$$0=\lim_{\lambda\to \infty}  \frac{1}{\lambda}\log  \E_\pi[\exp(\lambda\{ \|X-Y \|_q- S(\pi)\})].$$ 
That $0\geq \limsup (\dots)$ is direct, as the exponent is negative, and so we need only derive $0\leq \liminf (\dots)$. If the latter was not the case, then there would be $\varepsilon>0$ so that
$$\E_\pi[\exp(\lambda\{\|X-Y\|_q-(S(\pi)-\varepsilon)\})]\leq 1,$$
for all $\lambda$ large enough. But by definition of $S(\pi)$ a $\delta>0$ must exist so that $p_{\delta}:=\pi(\{\|X-Y\|_q > S(\pi)-\varepsilon +\delta \})>0$. Hence
$$\E_\pi[\exp(\lambda\{\|X-Y\|_q-(S(\pi)-\varepsilon)\})]\geq e^{\lambda \delta}p_{\delta},$$
and this is not bounded by $1$, leading to a contradiction. If now $S(\pi)$ is not finite, we apply the same reasoning to $\|X-Y\|_q\wedge m$, which does have a finite essential supremum, and then take supremum in $m$. Indeed, by  Point (2), we have \[\frac{1}{\lambda}\log  \E_\pi[\exp(\lambda \|X-Y \|_q\wedge m) ]= \sup_{\tilde \pi \sim \pi}\left\{\E_{\tilde\pi}[\|X-Y \|_q\wedge m] - \frac{1}{\lambda}H(\tilde \pi| \pi)\right\},\]  the r.h.s.\ of which is increasing in both $m$ and $\lambda$, 
so we conclude
\begin{align*}\esssup_\pi \|X- Y\|_q &= \sup_m\sup_\lambda\sup_{\tilde \pi \sim \pi}\left\{\E_{\tilde\pi}[\|X-Y \|_q\wedge m] - \frac{1}{\lambda}H(\tilde \pi| \pi)\right\} \\  &=\sup_\lambda\sup_m\sup_{\tilde \pi \sim \pi}\left\{\E_{\tilde\pi}[\|X-Y \|_q\wedge m] - \frac{1}{\lambda}H(\tilde \pi| \pi)\right\} \\
&= \lim_\lambda \frac{1}{\lambda}\log  \E_\pi[\exp(\lambda \|X-Y \|_q)].
\end{align*}
\end{proof}

\begin{lemma} \label{lem:ElambdatoAinfty}
We have 
$$\AW_{\infty,q}(\P,\P')= \lim_{\lambda\to \infty} E^q_\lambda(\P,\P')=\sup_{\lambda\geq 0 } E^q_\lambda(\P,\P').$$
\end{lemma}

\begin{proof}
By Lemma~\ref{lem:LapGib}(2), we have 
\[
E^q_\lambda(\P,\P')=  \inf_{\pi\in\Pi^{bc}(\P,\P')}\sup_{\tilde\pi\sim \pi}\left\{\E_{\tilde\pi}[\|X-Y \|_q] - \frac{1}{\lambda}H(\tilde \pi| \pi)\right\}\leq \inf_{\pi\in\Pi^{bc}(\P,\P')}\sup_{\tilde\pi\sim \pi}\E_{\tilde\pi}[\|X-Y \|_q]=\AW_{\infty,q}(\P,\P') . 
\]
Hence we only have to show that $\AW_{\infty,q}(\P,\P')\leq \liminf_{\lambda\to \infty} E^q_\lambda(\P,\P')$. Hence we may suppose that $\liminf_{\lambda\to \infty} E^q_\lambda(\P,\P')= \lim_{\lambda\to \infty} E^q_\lambda(\P,\P')<\infty$. Furthermore, by compactness, then up to selecting a subsequence we may assume $\pi_\lambda\to \hat\pi  \in\Pi^{bc}(\P,\P')$ where $\pi_\lambda$ is up to an error of $1/\lambda$ an optimizer of $E^q_\lambda(\P,\P')$. Thus by Lemma~\ref{lem:LapGib}(2),
\begin{align*}
E^q_\lambda(\P,\P')+1/\lambda & \geq \sup_{\tilde \pi \sim \pi_\lambda}\left\{\E_{\tilde\pi}[\|X-Y \|_q] - \frac{1}{\lambda}H(\tilde \pi| \pi_\lambda)\right\} \\ & \geq \sup_{\tilde \pi \sim \pi_\lambda}\left\{\E_{\tilde\pi}[\|X-Y \|_q] - \frac{1}{\lambda_0}H(\tilde \pi| \pi_\lambda)\right\}\\ &= \frac{1}{\lambda_0}\log  \E_{\pi_\lambda}[\exp(\lambda_0 \|X-Y \|_q) , 
\end{align*}
for all $\lambda\geq \lambda_0$. Taking limit in $\lambda$ and by the Portmanteau theorem we find $\lim_{\lambda\to \infty} E^q_\lambda(\P,\P')\geq \frac{1}{\lambda_0}\log  \E_{\hat \pi}[\exp(\lambda_0 \|X-Y \|_q) $. Taking limit in $\lambda^0$ and applying Lemma~\ref{lem:LapGib}(1) we obtain $ \lim_{\lambda\to \infty} E^q_\lambda(\P,\P')\geq\esssup_{\hat\pi} \|X- Y\|_q$ and we conclude.
\end{proof}
    
\begin{remark}
At this point it is relevant to stress that $E^q_\lambda(\cdot,\cdot)$ is not an extended metric. The theory of Modular spaces \cite{Mu06} serves as an inspiration in order to build extended metrics from functionals like $E^q_\lambda$; for instance
\begin{align*}
\AW_{exp,q}(\P,\P'):= \inf_{a\in \R_+}\left [ a+ a E^q_{1/a}(\P,\P')  \right]. 
\end{align*}
\hfill$\Diamond$
\end{remark}

\begin{remark}
The attentive reader may wonder if there is a connection between the here described Laplace method and the concept of $\epsilon$-arbitrage. This is indeed the case: a trading strategy $H$ is a strict $\epsilon$-arbitrage iff $\lim_{\lambda\to\infty}\frac{-1}{\lambda}\log\E[\exp\{-\lambda[(H \bullet S)_T - \epsilon \| H \|_p]\}] \geq 0$ and for some $\tilde \P\ll\P$ also $\lim_{\lambda\to\infty}\frac{-1}{\lambda}\log \E_{\tilde \P}[\exp\{-\lambda[(H \bullet S)_T - \epsilon \| H \|_p]\}] >0$.
\hfill$\Diamond$
\end{remark}

\subsection{Empirical estimation}

One obvious way to approximate  $\AW_{\infty,q}(\P,\P')$ would be to use its empirical counterpart $\AW_{\infty,q}(\P_n,\P'_n)$. Here $\P_n=\frac{1}{n}\sum_{k\leq n}\delta_{X^k}$ and $\P'_n=\frac{1}{n}\sum_{k\leq n}\delta_{{X'}^k}$, with $\{X^k\}_k$ being iid samples from $\P$ and $\{{X'}^k\}_k$  iid samples from $\P'$, and these samples being independent of each other. This would be too naive, though, as the essential supremum is too strong as an optimization criterion and we in general cannot expect for instance that  $\AW_{\infty,q}(\P_n,\P)\to 0$. Luckily Lemma~\ref{lem:ElambdatoAinfty} provides a way around this issue, as it suggest to approximate instead $E^q_\lambda(\P,\P')$ by  $E^q_\lambda(\P_n,\P'_n)$ for $\lambda$ large enough, and so hopefully $\AW_{\infty,q}(\P,\P')$ is well approximated by $E^q_\lambda(\P_n,\P'_n)$. 
However, this brings us to the next difficulty, as in fact we cannot expect for instance that { $E^q_\lambda(\P_n,\P')\to E^q_\lambda(\P,\P')$}. Similar phenomena, which are originated by the presence of bicausal couplings in the definition of the optimization problems at hand, have been described in great detail in \cite{BaBaBeWi20,PfPi16}. 

In this part we assume that $\P$ and $\P'$ have bounded support, and so w.l.g.\ that they are concentrated on $[0,1]^{d(T+1)}$.  Borrowing from \cite{BaBaBeWi20} we introduce the adapted empirical measure:

\begin{definition}[Adapted empirical measure]
\label{def:adapted.empirical.measure}
Set $r=(T+2)^{-1}$ for $d=1$ and $r=(d(T+1))^{-1}$ for $d\geq 2$.
For all $n\geq 1$, partition the cube $[0,1]^d$ into the disjoint union of $n^{rd}$ cubes with edges of length $n^{-r}$ and let $\varphi^n\colon[0,1]^d\to[0,1]^d$ map each such small cube to its center\footnote{That is, the function $\varphi^n$ satisfies $\sup_{u\in[0,1]^d} |u-\varphi^n(u)|\leq Cn^{-r}$ and its range $\varphi^n([0,1]^d)$ consist of $n^{rd}$ points.}.
Then, for all $n\geq 1$, define
\[ {\bf P_n} :=\frac{1}{n}\sum_{k=1}^n\delta_{\varphi^n(X_0^k),\dots,\varphi^n(X_T^k)},\]
where $\{X^k\}_k$ are iid samples from $\P$. 
We call ${\bf P_n} $ the {\emph {adapted empirical measure}} of $\P$.
\end{definition}

Similarly let ${\bf P'_n}$ be the adapted empirical measure of $\P'$, built from the independent iid samples $\{{X'}^k\}_k$. We can now formalize the idea of approximating $\AW_{\infty,q}(\P,\P')$ by means of adapted empirical measures:

\begin{lemma}\label{lem:empiricalconvergesinfinity}
For each $\lambda$, the following almost sure limit holds:
$$\lim_{n\to \infty} E^q_\lambda({\bf P_n},{\bf P'_n})=E^q_\lambda(\P,\P').$$
Hence almost surely we have
$$\AW_{\infty,q}(\P,\P')=\sup_{\lambda\geq 0} \lim_{n\to \infty} E^q_\lambda({\bf P_n},{\bf P'_n}).$$
\end{lemma}

\begin{proof}
Define $\AW_{1,q}$ exactly as $\AW_{\infty,q}$ but with $\esssup_\pi \|X- Y\|_q$ replaced by $\mathbb E_\pi[\|X- Y\|_q]$. By \cite[Theorem 1.3]{BaBaBeWi20} we have the almost sure limits
\begin{align*}
\lim_{n\to \infty} \AW_{1,q}({\bf P_n},\P) &= 0; \\
\lim_{n\to \infty} \AW_{1,q}({\bf P'_n},\P') &=0.
\end{align*}
On the other hand, the function $([0,1]^{d(T+1)})^2 \ni (X,Y)\mapsto \exp(\lambda\|X-Y\|_q)$ is continuous and bounded by an affine function of $\|X\|_q+\|Y\|_q$. We may apply the stability result \cite[Theorem 3.6]{EcPa22} obtaining $\lim_{n\to \infty} E^q_\lambda({\bf P_n},{\bf P'_n})=E^q_\lambda(\P,\P')$. For the last statement we notice that the $E^q_\lambda$ are increasing in $\lambda$, so by Lemma~\ref{lem:ElambdatoAinfty} we may write $\AW_{\infty,q}(\P,\P')=\sup_{\lambda\in\mathbb N}E^q_\lambda(\P,\P')$, and conclude by the first part after intersecting countably many full sets. 
\end{proof}

We do not pursue this, or any other computational/approximation method any further in this paper. We merely remark that, upon strengthening the assumptions on $\P,\P'$ as done in \cite{BaBaBeWi20}, one may be able to obtain deviation and concentration bound in the setting of Lemma~\ref{lem:empiricalconvergesinfinity}.

\subsection{Remarks on \texorpdfstring{$\AW_{\infty,q}$}{}}

There is much to be understood about $\AW_{\infty,q}(\P,\P')$. For instance, the identification or characterization of optimal couplings seems challenging. The following example shows that the so-called Knothe-Rosenblatt coupling\footnote{Let $U_0,\dots,U_T$ be independent uniform random variables on $[0,1]$. The Knothe-Rosenblatt rearrangement between $\P$ and $\P'$ is the law of $(W,Z)$ where $W=(W_0,\dots,W_T)$ and $Z=(Z_0,\dots,Z_T)$ are given by
\[W_0:=F^{-1}_{P_0}(U_0)\,\,\ , \,\,\, Z_0:=F^{
-1}_{P'_0}(U_0) ; \]
and inductively
\[W_{t+1}:=F^{-1}_{P^{W_0,\dots,W_t}}(U_{t+1})\,\,\ , \,\,\, Z_{t+1}:=F^{-1}_{P'^{Z_0,\dots,Z_t}}(U_{t+1}) ,\]
whereby $F^{-1}_{\rho^{x_0,\dots,x_t}}$ denotes the quantile function of the law under measure $\rho$ of the $(t+1)$-th coordinate given that the first $t+1$ coordinates are equal to  $x_0,\dots,x_t$, and $\rho_0$ denotes the time zero marginal (zeroth coordinate, so to speak) of $\rho$.
} need not be optimal for  $\AW_{\infty,q}(\P,\P')$ even when $\P,\P'$ are Markov laws with increasing kernels\footnote{	We say that $ \mu\in \mathcal P(\R^{T+1})$ is \emph{Markov with increasing kernels}, if 
$\mu$ is the law of a discrete-time, real-valued, Markov process in $T+1$ steps, and for each $t\in{0,\dots,T-1}$: if $x_t\leq \bar x_t$ for $x_t,\bar x_t\in\text{supp}(\mu_t)$, then $F_{\mu^{x_t}}\geq F_{\mu^{\bar x_t}} $ pointwise.  	} This is drastically different from the usual setting of adapted Wasserstein distances (cf.\ \cite{BaBeLiZa17}) where one expects optimality of Knothe-Rosenblatt in the presence of Markov marginals with increasing kernels.

\begin{example}
For $T=1$, we consider the marginal process laws
$$\P:=\frac{1}{2}[\delta_{(0,3)}+\delta_{(1,5)}] \,\,\,\, \text{and} \,\,\,\,  \P':=\frac{1}{2}[\delta_{(1,1)}+\delta_{(3,2)} ].$$
It is direct to check that these are Markov with increasing kernels. The Knothe-Rosenblatt coupling is given by
$$\pi^{KR}=\frac{1}{2}[\delta_{((0,3),(1,1))}+ \delta_{((1,5),(3,2))}].$$
We establish now that $\pi^{KR}$ is not optimal for $\AW_{\infty,q}(\P,\P')$. Indeed, consider
$$\pi:= \frac{1}{2}[\delta_{((0,3),(3,2))}+ \delta_{((1,5),(1,1))}],$$
which is a feasible coupling. Then
$$\esssup_{\pi^{KR}}\{|x_1-y_1|+|x_2-y_2|\}=5>4=\esssup_{\pi}\{|x_1-y_1|+|x_2-y_2|\}.$$ \hfill$\Diamond$
\end{example}

Another aspect where the theory for $\AW_{\infty,q}(\P,\P')$ could be expanded concerns duality. In the case of $\mathcal W_{\infty,q}(\P,\P') $ this was done in \cite{BaBoJe17} via quasiconvex duality techniques.

\section{Postponed proofs}\label{sect:proofs}

\subsection{The canonical decomposition}\label{sect:proofs.candec}

We provide the pending proofs of Lemma~\ref{lem:H_bar} and Lemma~\ref{lem:decomp}.

\begin{proof}[Proof of Lemma~\ref{lem:H_bar}]
Fix any $t\in\{1,\ldots,T\}$.

\emph{Step 1}:
We construct $\bar H_t \in E_{\epsilon,t}$ with $|\bar H_t|_p \in \{0,1\}$ that has maximal support, i.e.,
\begin{align}
\label{eq:lem_H_bar.maximality_p>1}
\P[ \bar H_t \neq 0 ] =& \sup_{H_t \in E_{\epsilon,t}} \P[H_t \neq 0] &&\mbox{if }p > 1, \\
\label{eq:lem_H_bar.maximality_p=1}
\P[ \bar H_{t,i} \neq 0 ] =& \sup_{H_t \in E_{\epsilon,t}
} \P[H_{t,i} \neq 0] &&\mbox{if }p =1, \mbox{ for }i = 1,\ldots,d.
\end{align}
To see existence of such a strategy $\bar H_t$, fix $i \in \{1,\ldots,d\}$ when $p = 1$,
pick a maximizing sequence $(H^k_t)_{k \in \N}\subseteq E_{\epsilon,t}$ of the right-hand side of \eqref{eq:lem_H_bar.maximality_p>1} resp.\ \eqref{eq:lem_H_bar.maximality_p=1}. Define the $\F_{t-1}$-measurable set $A^k := \{ H^k_t \neq 0 \mbox{ and } H^j_t = 0 \mbox{ for }j=1,\ldots,k-1 \}$ for $k \in \N$, in case $p>1$, and $A^k_i := \{ H^k_{t,i} \neq 0 \mbox{ and } H^j_{t,i} = 0 \mbox{ for }j=1,\ldots,k-1 \}$ in case $p=1$.
Then, $\mathbbm 1_{A^k} H^k_t$ resp.\ $\mathbbm 1_{A^k_i} H^k_t$ is contained in $E_{\epsilon,t}$
and by Lemma~\ref{lem:E_conv_cone} also the sum
\[
\bar H_t := \sum_{k \in \N} \mathbbm 1_{A^k} \frac{H^k_t}{|H^k_t|_p} \in E_{\epsilon,t},
\]
 in case $p>1$, while for $p=1$
 \[
\bar H_t^i := \sum_{k \in \N} \mathbbm 1_{A^k_i} \frac{H^k_t}{|H^k_t|_p} \in E_{\epsilon,t}.
\]
Since $\{\bar H_t \neq 0\} \supset \{H^k_t \neq 0\}$ resp.\ 
$\{\bar H_{t,i}^i \neq 0 \} \supset \{H^k_{t,i} \neq 0 \}$ for all $k \in \N$, 
we find $\P[\bar H_t \neq 0] \ge \lim_{k \to \infty} \P[H^k_t \neq 0]$ and resp.\  $\P[\bar H_t^i \neq 0] \ge \lim_{k \to \infty} \P[H^k_{t,i} \neq 0]$. Since $(H^k_t)_{k \in \N}$ 
is a maximizing sequence of \eqref{eq:lem_H_bar.maximality_p>1} resp.\ \eqref{eq:lem_H_bar.maximality_p=1}, then
$\bar H_t$ resp.\ $\bar H_t^i$ attains the supremum. Since the sets $(A^k)_k$ resp.\ $(A_i^k)_k$ are disjoint, we have $|\bar H_t|_p, |\bar H_t^i|_p\in\{0,1\}$ as desired.

If $p = 1$, we have shown that for every $i \in \{1,\ldots,d\}$ there is an $\bar H_t^i$ that maximizes \eqref{eq:lem_H_bar.maximality_p=1} for that particular index.
Note that, when $H^1_t,H^2_t \in E_{\epsilon,t}$, by \eqref{eq:conv_cone} and the absence of strict $\epsilon$-arbitrage, we have (modulo null sets) for all $j \in \{1,\ldots,d\}$
\[
H^1_{t,j} \neq 0, H^2_{t,j} \neq 0 \implies \sgn(H^1_{t,j}) = \sgn(H^2_{t,j}),
\]
since the absolute value is only linear on the positive or the negative half-line.
Hence, 
$\bar H_t :=\frac{1}{d} \sum_{j = 1}^d \bar H_t^j$ maximizes \eqref{eq:lem_H_bar.maximality_p=1} for all $i \in \{1,\ldots,d\}$. Furthermore, with this choice, we have for all $H_t \in E_{\epsilon,t}$:
\[
H_t \cdot \bar H_t^\ast = \sum_{i = 1}^d H_{t,i} \sgn(\bar H_{t,i}) = H_t \cdot H_t^\ast = |H_t|_1|H_t^\ast|_\infty = |H_t|_1|\bar H_t^\ast|_\infty,
\]
which means that $\bar H_t^\ast \in F_{\epsilon,t}$ when $p=1$.

\emph{Step 2}:
For the remainder of the proof, assume that $p > 1$.
Let $H_t \in E_{\epsilon,t}$.
We claim that there exists $a_t \in L(\Omega,\F_{t-1},\P;\R_+)$ with $a_t \bar H_t = H_t$, where $\bar H_t$ was obtained in the previous step.
By \eqref{eq:lem_H_bar.maximality_p>1}, $\{\bar H_t = 0, H_t \neq 0\}$ must be a null set.
By Lemma \ref{lem:E_conv_cone} we have
\[
\epsilon (|H_t|_p + |\bar H_t|_p ) = (H_t + \bar H_t) \cdot \Delta S_t
= \epsilon |H_t + \bar H_t|_p.
\]
By strict convexity of the $p$-norm $| \cdot |_p$, $|H_t + \bar H_t|_p = |H_t|_p + |\bar H_t|_p$ can only hold true when almost surely $H_t$ is a non-negative multiple of $\bar H_t$, thus,
\[
E_{\epsilon,t} = \left\{a_t \bar H_t \colon a_t \in L^0(\Omega,\F_{t-1},\P;\R_+) \right\},
\]
which readily implies that $\bar H_t^\ast \in F_{\epsilon,t}$.
{ Consequently, as $\bar H_t \in \{ 0,1 \}$, we find {for $X \in \mathcal H_t$} that
\[
    X \in F_{\eps,t} \iff X \cdot \bar H_t = |X|_q|\bar H_t|_p \iff X = \mathbbm 1_{ \{ \bar H_t = 0 \} } X + |X|_q \mathbbm 1_{ \{ \bar H_t \neq 0 \} } \bar H_t^\ast.
\]
Therefore, we derive {for $H_t \in \mathcal H_t$ that}
\[
H_t \in F_{\epsilon,t}^\perp \iff H_t \cdot \bar H_t^\ast = 0 \mbox{ and } \mathbbm 1_{\{ \bar H_t = 0 \}}|H_t|_p = 0,
\]
which provides the desired representation of $F_{\epsilon,t}^\perp$.}
\end{proof}

\begin{proof}[Proof of Lemma \ref{lem:decomp}]
We want to decompose $H_t \in \bar{\mathcal H}_t$ in a univocal way in a triple $(a_t, G_t, \tilde G_t)$ where $a_t \in L^0(\Omega,\F_{t-1},\P;\R)$, $G_t \in F_{\epsilon,t}^\perp \cap E_{0,t}^\perp$ and $\tilde G_t \in  F_{\epsilon,t}^\perp \cap E_{0,t}$.

To find $a_t$, we solve the equation
\[
(H_t - a_t \bar H_t) \cdot \bar H_t^\ast = 0,
\] 
which admits a solution (unique on $\{ \bar H_t \neq 0 \}$) given by $a_t := H_t \cdot \bar H_t^\ast$. 
We deduce by Lemma~\ref{lem:H_bar} that $H_t - a_t \bar H_t \in F_{\epsilon,t}^\perp$.
Next, recall that by \cite[Lemma 6.2.1]{DeSc06} there exists an $\F_{t-1}$-measurable mapping $P_{0,t}$ taking values in the orthogonal projections in $\R^d$, so that {for $H_t \in \mathcal H_t$
\[
    H_t \in F_{\epsilon,t}^\perp \cap E_{0,t} \iff P_{0,t} H_t = H_t.
\]}
Now we define
\[
G_t := (\id - P_{0,t}) (H_t - a_t \bar H_t) \quad \mbox{and} \quad \tilde G_t := P_{0,t}(H_t - a_t \bar H_t),
\]
and note that $\tilde G_t \in F_{\epsilon,t}^\perp \cap E_{0,t}$, thus $G_t \in F_{\epsilon,t}^\perp \cap E_{0,t}^\perp$.

To see uniqueness (in an almost sure sense) let $(a_t',G_t', \tilde G_t')$ be an admissible decomposition of $H_t \in \bar{\mathcal H}_t$.
We have that $H_t - a_t' \bar H_t \in F_{\epsilon,t}^\perp$ which can only be the case when $a_t' = a_t$.
Then $G_t = G_t'$ and $\tilde G_t = \tilde G_t'$ follow from properties of the orthogonal projection-valued mapping $P_{0,t}$.

\emph{Case $p = 2$}: In this case, the map $P_{\epsilon,t} H_t := H_t - a_t \bar H_t$ coincides with the $\F_{t-1}$-measurable mapping that is provided by \cite[Lemma 6.2.1]{DeSc06} and takes values in the orthogonal projections such that $H_t \in E_{\epsilon,t} = F_{\epsilon,t} \cap \bar{\mathcal H}_t$ if and only if $P_{\epsilon,t} H_t = H_t$.
As elements in $F_{\epsilon,t} \cap \bar{\mathcal H}_t$ are almost surely orthogonal to elements in $F_{\epsilon,t}^\perp$, we find $P_{\epsilon,t} \circ P_{0,t} = P_{0,t} \circ P_{\epsilon,t}$ and, in particular, $(a_t\bar H_t, G_t, \tilde G_t)$ is an orthogonal decomposition of $H_t$.
\end{proof}

\subsection{\texorpdfstring{$\NA_\epsilon$}{}: the one-step case}\label{sec:Nae_one_step}

We begin with the postponed argument for Remark~\ref{rem:localization}:
{
\begin{proof}[Argument for Remark \ref{rem:localization}] Assume that $H \in \mathcal H$ is a strict $\eps$-arbitrage and, for $t \in \{1,\ldots, T \}$, denote by $V_t := (H \bullet S)_{t} - \epsilon \sum_{s = 1}^{t} |H_s|_p$ the value function with the convention that $V_0 = 0$.
Similar to the classical framework, we can use the linearity in time of the integral $H \bullet S$ and of the norm $\| \cdot \|_p$ and consider the time
\[
{t} := \min \left\{ s \in\{1,\ldots,T\} \colon V_s \ge 0\, \mbox{ and }\, \P[V_s > 0] > 0 \right\},
\]
which is well-defined as we assumed that $H$ is a strict $\eps$-arbitrage opportunity.
We distinguish two cases: either $\P[ V_{{t} - 1} < 0] > 0$ or $\P[V_{{t} - 1} = 0 ]=1$.
In the first case, we have $A := \{ V_{{t} - 1} < 0 \} \in \F_{t - 1}$ and $\tilde H_{t} := \mathbbm 1_A H_{t} / |H_{t}|_p \in \mathcal H_{{t}}$ is also a strict $\eps$-arbitrage, since
\[
    \tilde H_{t} \cdot \Delta S_{t} {- \eps |\tilde H_t
    |_p }= \frac{\mathbbm 1_A}{|H_{t}|_p} \left( V_{t} - V_{{t} - 1} \right) \ge 0,
\]
where the inequality is strict on $A$ and $\P[A] > 0$ by assumption.
On the other hand, in the second case we have that $V_{{t} - 1} = 0$ which means that $\tilde H_t := H_{t} / |H_{t}|_p \in \mathcal H_{t}$ already is a strict $\eps$-arbitrage since then
\[
    \tilde H_t \cdot \Delta S_t { - \eps |\tilde H_t
    |_p } = \frac{1}{|H_{t}|_p} V_t.
\]
Analogous reasoning establishes the localization of $\NA_\epsilon'$.   
\end{proof}
}
Remark \ref{rem:localization} motivates the fact that we should first study the situation of a single period market. Hence, we now fix $t=1,\ldots,T$ and argue for investments in a single period, $[t-1,t)$.\\

We will often employ the concept of measurable subsequences:

\begin{definition}
An $\N$-valued, $\F$-measurable function is called a random time.
A strictly increasing sequence $(\tau_k)_{k \in \N}$ of random times is called an $\F$-measurably parametrized subsequence or simply an $\F$-measurable subsequence.
\end{definition}

We recall our convention: we say that a sequence $(f_n)_{n \in \N}\subset L^0(\Omega,\F, \P; \R)$ is bounded (resp.\ bounded from below) if there is $c\in L^0(\Omega,\F, \P; \R_+)$ such that $|f_n|\leq c$ (resp.\ $f_n\geq -c $) for all $n$ a.s. 
{ Moreover, we say that a sequence $(f_n)_{n \in \N}\subset L^0(\Omega,\F, \P; \R^d)$ is bounded if $(|f_n|)_{n \in \N}\subset L^0(\Omega,\F, \P; \R)$ is bounded.
}    
We recall that, by \cite[Proposition 6.3.3]{DeSc06}, for any sequence $(f_n)_{n \in \N}$ of functions in $L^0(\Omega,\F, \P; \R^d)$ which is bounded, we may find an $\F$-measurable subsequence $(\tau_n)_{n \in \N}$ such that $(f_{\tau_n})_{n \in \N}$ converges for all $\omega \in \Omega$.

\begin{lemma}
\label{lem:1_step_H_bar_complement}
Assume absence of strict $\epsilon$-arbitrage and let $(H_t^k)_{k \in \N}$ be such that $(H_t^k \cdot \Delta S_t - \epsilon |H_t^k|_p)_{k \in \N}$ is almost surely bounded from below. Then
\[
\P\Bigg[ 
\Big\{ \limsup_{k \to \infty} |H_t^k|_p = \infty \Big\} 
\setminus
\{ \bar H_t \neq 0 \}
\Bigg] = 0.
\]
In particular, $(H^k_{t})_{k \in \N}$ has to be bounded on $\{\bar H_t = 0 \}$.
If $\epsilon > \epsilon(\P)$, then $(H^k_{t})_{k \in \N}$ has to be bounded.
\end{lemma}

\begin{proof} 
Write $A := \{ \limsup_{k \to \infty} |H_t^k|_p = \infty \}$ and define the normalized sequence
\begin{equation}\label{eq:def_normalized_sequence}
\tilde H_t^k := \mathbbm 1_A \frac{H_t^k}{|H_t^k|_p \vee 1}.
\end{equation}
Since $(\tilde H_t^k)_{k \in \N}$ is bounded, we can assume, by first passing to an $\F_{t-1}$-measurable subsequence, that $A = \{ \lim_{k \to \infty} |H_t^k|_p = \infty \}$ and that $(\tilde H_t^k)_{k \in \N}$ converges to some $\tilde H_t$.
Observe then that $\{ |\tilde H_t|_p = 1 \} = A$ and $\{ |\tilde H_t|_p = 0 \} = A^\textrm{c}$.
By the lower boundedness assumption, we have
\[
\tilde H_t \cdot \Delta S_t - \epsilon |\tilde H_t|_p = \lim_{k \to \infty} \tilde H_t^k \cdot \Delta S_t - \epsilon |\tilde H_t^k|_p \ge 0,
\]
which yields $\tilde H_t \in E_{\epsilon,t}$ by the absence of strict $\eps$-arbitrage.
Then the assertion follows by the structure of  $E_{\eps,t}$  exposed in Lemma~\ref{lem:H_bar}.

If $\eps > \eps(\P)$ then $\bar H_t = 0$ by Lemma \ref{lem:critical_value}, from which we obtain as a consequence of the first part that $(H^k_t)_{k \in \N}$ has to be bounded.
\end{proof}

\begin{lemma}
\label{lem:1_step_eps_arbitrage}
Assume $\NA_\epsilon'$ and let $(G^k_t)_{k \in \N}$ be a sequence in $F_{\epsilon,t}^\perp \cap E_{0,t}^\perp$. 
If $(G^k_t \cdot \Delta S_t)_{k \in \N}$ is a.s.\ bounded from below, then $(G^k_t)_{k \in \N}$ is a.s.\ bounded.
In particular, the following statements are true:
\begin{enumerate}[label = (\alph*)]
\item $\{ H_t \cdot \Delta S_t \colon H_t \in F_{\epsilon,t}^\perp \}$ is closed w.r.t.\ convergence in probability;
\item for any sequence $(Y^k)_{k \in \N}$ in $\{ H_t \cdot \Delta S_t \colon H_t \in F_{\epsilon,t}^\perp \}$ that is lower bounded, there is $Y \in \{ H_t \cdot \Delta S_t \colon H_t \in F_{\epsilon,t}^\perp \}$ with
\[
\liminf_{k \to \infty} Y^k \le Y.
\]
\end{enumerate}
\end{lemma}

\begin{proof}
We reason by relating the current setting to the classical no-arbitrage framework.
To this end, we consider a shifted process $\bar S$ whose increments are given by
\begin{align*}
\Delta \bar S_t :=& \mathbbm 1_{\{\bar H_t \neq 0\}} \left(\Delta S_t - \epsilon \bar H_t^\ast\right), &&\mbox{if }p > 1, \\
\Delta \bar S_{t,i} :=& \mathbbm 1_{\{ \bar H_{t,i} \neq 0\}} \left(\Delta S_{t,i} - \epsilon \bar H_{t,i}^\ast \right), &&\mbox{if }p=1,\, \text{for $i = 1,\ldots, d$}.
\end{align*}
Any $H_t \in \bar{\mathcal H}_t$ with canonical decomposition $(a_t,G_t,\tilde G_t)$, see Lemma~\ref{lem:decomp}, satisfies
\begin{equation}
\label{eq:1_step_eps_arbitrage_bar_Delta}
G_t \cdot \Delta S_t = G_t \cdot \Delta \bar S_t = H_t \cdot \Delta \bar S_t.
\end{equation}
Therefore, under $\NA_\epsilon'$ we find as a consequence of \eqref{eq:1_step_eps_arbitrage_bar_Delta} that $(\bar S_{t-1},\bar S_t)$ admits no arbitrage in the classical sense, i.e.\ for $H_t \in \mathcal H_t$ we have that
\begin{equation}
\label{eq:1_step_eps_arbitrage_NA}
H_t \cdot \Delta \bar S_t \ge 0 \implies H_t \cdot \Delta \bar S_t = 0.
\end{equation}
Indeed, when $H_t \in \mathcal H_t$ then the strategy $\tilde H_t$ given by
\begin{align*}
    \tilde H_t :=& \mathbbm 1_{\{ \bar H_t \neq 0 \}} H_t &&p > 1, \\
    \tilde H_{t,i}:= & \mathbbm 1_{\{ \bar H_{t,i} \neq 0 \}} H_{t,i} &&p = 1, \mbox{ for }i = 1,\ldots,d,
\end{align*}
is an element of $\bar{\mathcal H}_t$ with $\Delta \bar S_t \cdot H_t = \tilde H_t \cdot \Delta S_t$. Thus by \eqref{eq:1_step_eps_arbitrage_bar_Delta}, if $H_t \cdot \Delta \bar S_t$ is non-negative then according to $\NA_\epsilon'$ it has to vanish.
Hence, $\Delta \bar S_t$  admits no arbitrage in the sense of \eqref{eq:1_step_eps_arbitrage_NA}.

We conclude the proof by invoking \cite[Proposition 6.4.1]{DeSc06} and the fact that bounded $\mathcal F_{t-1}$-measurable sequences always allow to pass to convergent, $\mathcal F_{t-1}$-measurable subsequences.
\end{proof}

\begin{lemma}
\label{lem:p=1.H_bounded}
Assume the absence of strict $\epsilon$-arbitrage and $p = 1$.
Let $(H^k_t)_{k \in \N}$ be a sequence in $\mathcal H_t$ such that $(H^k_t \cdot \Delta S_t - \epsilon |H^k_t|_1)_{k \in \N}$ is a.s.\ bounded from below.
Then, for $i = 1,\ldots, d$, $(H^k_{t,i})_{k \in \N}$ is a.s.\ bounded on $\{ \bar H_{t,i} = 0 \}$.
\end{lemma}

\begin{proof}
By Lemma \ref{lem:1_step_H_bar_complement} we have that $(\mathbbm 1_{\{ \bar H_t = 0 \}} H^k_t)_{k \in \N}$ is bounded, therefore we may assume w.l.o.g.\ that $H^k_t$, $k \in \N$ vanishes
outside of $\{\bar H_t \neq 0 \}$ for every $k \in \N$.
On $\{ \bar H_t \neq 0 \}$ we define, for $k \in \N$ and $i\in \{1,\ldots,d\}$,
\[
\hat H^k_{t,i} := \mathbbm 1_{\{\bar H_{t,i} \neq 0 \}} H^k_{t,i}
\quad \mbox{and} \quad
\tilde H^k_{t,i} := \mathbbm 1_{\{ \bar H_{t,i} = 0 \} } H^k_{t,i},
\]
and observe that $|H^k_t|_1 = |\hat H^k_t|_1 + |\tilde H^k_t|_1$.
It suffices 
to show that $(\tilde H^k_t)_{k \in \N}$ is a.s.\ bounded. 

Set $A := \{ \limsup_{k \to \infty} |\tilde H^k_t|_1 = \infty \}$ and consider the normalized sequence $(\tilde H^k_t / (|\tilde H^k_t|_1 \vee 1))_{k \in \N}$.
By the same reasoning as in Lemma \ref{lem:1_step_H_bar_complement} (that involves passing to an $\F_{t-1}$-measurable subsequence), we can assume that the normalized sequence converges to some $\tilde H_t \in \mathcal H_t$, $\mathbbm 1_A = |\tilde H_t|_1$, and 
\[
\tilde H_t \cdot \Delta S_t - \epsilon |\tilde H_t|_1 + \liminf_{k \to \infty} \frac{1}{|\tilde H_t^k|_1 \vee 1} \left( \hat H^k_t \cdot \Delta S_t - \epsilon |\hat H^k_t| \right) \ge 0.
\]
Since by \eqref{eq:p-norm_lower_bound} $\hat H^k_t \cdot \Delta S_t - \epsilon |\hat H^k_t| \le G^k_t \cdot \Delta S_t$ (where $(a^k_t, G^k_t, \tilde G^k_t)$ denotes the canonical decomposition of $\hat H^k_t$, for $k \in \N$), 
{ we have that $(G^k_t \cdot \Delta S_t / (|\tilde H_t^k|_1 \vee 1))_{k \in \N}$ is bounded from below.
Using Lemma \ref{lem:1_step_eps_arbitrage} we can pass to a measurable subsequence such that $(G^k_t / (|\tilde H_t^k|_1 \vee 1))_{k \in \N}$ converges in probability.
Passing to the limit by Lemma \ref{lem:p=1.strict_G} we find $H_t \in \bar{\mathcal H}_t$ with}
\[
\tilde H_t \cdot \Delta S_t - \epsilon |\tilde H_t|_1 + H_t \cdot \Delta S_t - \epsilon |H_t|_1 = 0,
\]
where the equality is due to the absence of strict $\epsilon$-arbitrage.
Then $\tilde H_t + H_t \in E_{\epsilon,t}$.
Since $\bar H_t$ is maximal, in the sense that (modulo null sets)
\[
\{\bar H_{t,i} \neq 0\} \supseteq \{ \hat G_{t,i} \neq 0 \}\quad \forall \hat G_{t,i} \in E_{\epsilon,t},
\]
and since (modulo null sets) $\{ \tilde H_{t,i} + H_{t,i} \neq 0\} = \{\tilde H_{t,i} \neq 0\} \cup \{ H_{t,i} \neq 0\}$, thus $\P[A] = 0$.
\end{proof}

\begin{lemma}
\label{lem:1_step_propA_NA}
Assume $\NA_\epsilon'$.
Then any pair $(H_t,G_t) \in \mathcal H_t \times F_{\epsilon,t}^\perp$ satisfies
\[
(H_t + G_t) \cdot \Delta S_t - \epsilon |H_t|_p \ge 0 \implies (H_t + G_t) \cdot \Delta S_t - \epsilon |H_t|_p = 0.
\]
\end{lemma}

\begin{proof}
Note that when $p = 1$ the assertion follows directly from the absence of strict $\epsilon$-arbitrage and Lemma \ref{lem:p=1.strict_G}.

Now let $p > 1$.
Since on $\{ \bar H_t = 0\}$ we have that $G_{t} = 0$, 
 the assertion is a consequence of the absence of strict $\epsilon$-arbitrage. So,
assume w.l.o.g.\ that $H_t \in \bar{\mathcal H}_t$ with canonical decomposition $(a_t, \hat G_t, \tilde G_t)$ and $G_t \in F_{\epsilon,t}^\perp$ be such that
\[
(H_t + G_t) \cdot \Delta S_t - \epsilon |H_t|_p \ge 0.
\]
Observe that by \eqref{eq:p-norm_lower_bound} we have
\[
(H_t + G_t) \cdot \Delta S_t - \epsilon |H_t|_p \le (\hat G_t + G_t) \cdot \Delta S_t.
\]
Therefore, we find that $(\hat G_t + G_t) \cdot \Delta S_t \ge 0$.
We conclude by $\NA_\epsilon'$ that also
\[
0 = (\hat G_t + G_t) \cdot \Delta S_t = (H_t + G_t) \cdot \Delta S_t - \epsilon |H_t|_p,
\]
which yields the claim.
\end{proof}

\begin{lemma}
\label{lem:1_step_closure_simple}
Assume the absence of strict $\epsilon$-arbitrage.
Let $(H^k_t)_{k \in \N}$ be a sequence in $\bar{\mathcal H}_t$ with canonical decompositions $(a^k_t, G^k_t, \tilde G^k_t)$.
If $(|G^k_t|)_{k \in \N}$ and $(a^k_t - |H^k_t|_p)_{k \in \N}$ are a.s.\ bounded, then there exists an $\mathcal F_{t-1}$-measurable subsequence of $(H^k_t \cdot \Delta S_t - \epsilon |H_t^k|_p)_{k \in \N}$ with limit in \[\{ H_t \cdot \Delta S_t - \epsilon |H_t|_p + G_t \cdot \Delta S_t \colon H_t \in \mathcal H_t, G_t \in F_{\epsilon,t}^\perp, H_t^\ast \cdot G_t = 0 \}.\]
Moreover, the latter set coincides with
$\{ H_t \cdot \Delta S_t - \epsilon |H_t|_p + \mathbbm 1_{\{H_t = 0\}} G_t \cdot \Delta S_t \colon H_t \in \mathcal H_t, G_t \in F_{\epsilon,t}^\perp \}$.
\end{lemma}

\begin{proof}
For the last assertion, observe that it suffices to show the statement for $H_t \in \bar{\mathcal H}_t$ 
and $G_t \in F_{\epsilon,t}^\perp$ since on $\{\bar H_t = 0 \}$ (when $p > 1$) resp.\ $\{ \bar H_{t,i} = 0 \}$ (when $p = 1$ and $i \in \{1,\ldots,d\}$) $G_t$ resp.\ $G_{t,i}$ have to vanish.
To this end, let $H_t \in \bar{\mathcal H}_t$ with decomposition
$(a_t,G_t,\tilde G_t)$
and $\hat G_t \in F_{\epsilon,t}^\perp$ with $H_t^\ast \cdot \hat G_t = 0$.
Since
\[
|H_t + \hat G_t|_p \ge H_t^\ast \cdot (H_t + \hat G_t) = |H_t|_p,
\]
and the map
\[
\R \ni x \mapsto a_t + x - |x \bar H_t + H_t + \mathbbm 1_{\{H_t \neq 0 \} } \hat G_t|_p
\]
is a.s.\ continuous with limit equal to $0$ for $x\to\infty$, cf.\ Remark \ref{lem:square_root_asympt} and \eqref{eq:p-norm_lower_bound}, there is $\hat a_t \in L^0(\Omega,\F_{t-1},\P;\R)$ with $|\hat a_t| \ge | a_t|$, $\{ H_t = 0 \} \subset \{ \hat a_t = 0 \}$, and 

\[
a_t - |H_t|_p = \hat a_t - |(\hat a_t - a_t) \bar H_t + H_t + \mathbbm 1_{\{ H_t \neq 0 \}} \hat G_t|_p.
\]
Define $\hat H_t$ as the element of $\bar {\mathcal H}_t$ with canonical decomposition $(\hat a_t, G_t +  \mathbbm 1_{\{ H_t \neq 0 \}} \hat G_t, \tilde G_t)$.
Then {$\{\hat H_t = 0 \} = \{H_t = 0 \}$ and } $\hat H_t$ and $\hat G_t$  admit the desired representation, since
\begin{align*}
(\hat H_t + \mathbbm 1_{ \{ \hat H_t = 0 \} } \hat G_t) \cdot \Delta S_t - \epsilon |\hat H_t|_p &=     (\hat H_t + \mathbbm 1_{ \{  H_t = 0 \} } \hat G_t) \cdot \Delta S_t - \epsilon |\hat H_t|_p \\
&= (G_t + \hat G_t) \cdot \Delta S_t + \epsilon (\hat a_t - |\hat H_t|_p)
\\
&= (G_t + \hat G_t) \cdot \Delta S_t + \epsilon (a_t - |H_t|_p)
\\
&= (H_t + \hat G_t) \cdot \Delta S_t - \epsilon |H_t|_p.
\end{align*}
This shows {that the last assertion, i.e., the two sets coincide.}
    
We now turn our attention towards the first assertion.
{By potentially passing to an $\mathcal F_{t-1}$-measurable subsequence, we can assume w.l.o.g.\ that the sequences, $(|G^k_t|_p)_{k \in \N}$, $(a^k_t - |H^k_t|_p)_{k \in \N}$ and $(\min(a^k_t, 0))_{k \in \N}$ are convergent with limit $G_t \in F^\perp_{\epsilon,t} \cap E^\perp_{0,t}$, $c_t \in L^0(\Omega,\mathcal F_{t-1}, \mathbb P; \mathbb R^-)$ and $\tilde a_t \in L^0(\Omega,\mathcal F_{t-1}, \mathbb P; \mathbb R^-)$ respectively. }
Therefore,
\begin{equation}\label{eq:limhk}
\lim_{k \to \infty} H_t^k \cdot \Delta S_t - \epsilon |H_t^k|_p = G_t \cdot \Delta S_t + \epsilon c_t.
\end{equation}
Observe that $c_t$ is $\mathcal F_{t-1}$-measurable and consider the $\mathcal F_{t-1}$-measurable sets
\[
A := \{ c_t = 0 \}, \quad 
B := \{ c_t < 0, \tilde G^k_t = 0\ \forall k \in \N\}, \quad
C := \{ c_t < 0, \exists k \in \N \text{ s.t. } \tilde G^k_t \neq 0\}.
\]
Then there exists an $\F_{t-1}$-measurable function $\kappa \colon \Omega \to \N$ with $\tilde G^\kappa_t \neq 0$ on $C$, in particular, $\tilde G^\kappa_t \in E_{0,t}$.
Since, by the triangle inequality $|x \bar H_t + G_t|_p \le |x| + |G_t|_p$, we have on { $B \cap \{ \tilde a_t = 0 \}$} that $c_t \ge - | { G_t} |_p$ { whereas on $B \cap \{ \tilde a_t < 0 \}$ we have that $c_t = \tilde a_t - |\tilde a_t \bar H_t +{ G_t} |_p$.} 
Thus,
\begin{equation}
\label{eq:lem_1_step_closure_B}
x - |x \bar H_t + G_t|_p = c_t
\end{equation}
has a solution.
We write $a_t \in L^0(\Omega,\mathcal F_{t-1},\P; \mathbb R)$ for a function solving \eqref{eq:lem_1_step_closure_B} on $B$ and vanishing on the complement.
Since {on $C$,} $x \mapsto |G_t + x \tilde G_t^\kappa|_p$ diverges for $x \to \infty$, there is $b_t \in L^0(\Omega, \mathcal F_{t-1},\P; \mathbb R_+)$ that vanishes outside of $C$ and satisfies
\[
-|G_t + b_t \tilde G^{\kappa}_t|_p = c_t.
\]
Then
we can use the pair $(H_t, \mathbbm 1_A G_t)$, where $H_t$ has the canonical decomposition $(a_t, \mathbbm 1_{\{B \cup C \}} G_t, b_t \tilde G^\kappa_t)$, in order to represent the limit in \eqref{eq:limhk} as
\[
G_t \cdot \Delta S_t + \epsilon c_t = G_t \cdot \Delta S_t + \epsilon (a_t - |a_t \bar H_t +{ \mathbbm 1_{\{B \cup C\}} }G_t + b_t \tilde G_t^\kappa|_p) = H_t \cdot \Delta S_t - \epsilon |H_t|_p + \mathbbm 1_A G_t \cdot \Delta S_t.
\]
\end{proof}

\begin{corollary}
\label{cor:1_step_boundedness_canonical_decomp}
Assume $\NA_\epsilon'$.
Let $H^k_t \in \bar{\mathcal H}_t$, $k \in \N$ with canonical decomposition $(a^k_t, G^k_t, \tilde G^k_t)$ be such that $(H^k \cdot \Delta S_t - \epsilon |H^k_t|_p)_{k \in \N}$ is almost surely bounded from below.
Then the sequences $(|G^k_t|_p)_{k \in \N}$ and $( |H^k_t|_p - a^k_t)_{k \in \N}$ are almost surely bounded.
\end{corollary}

\begin{proof}
Since, by \eqref{eq:p-norm_lower_bound}, we have, for $k \in \N$, that
\[
G^k_t \cdot \Delta S_t \ge H^k_t \cdot \Delta S_t - \epsilon |H^k_t|_p,
\]
we can apply Lemma~\ref{lem:1_step_eps_arbitrage} and conclude that $(|G^k_t|)_{k \in \N}$ has to be almost surely bounded.
Moreover, observe that $|H_t^k|_p - a^k_t$ is by \eqref{eq:p-norm_lower_bound} non-negative and that the left-hand side of
\[
\sup_{k \in \N} |G_t^k|_p |\Delta S_t|_q + \epsilon \liminf_{k \to \infty} a^k_t - |H_t^k|_p \ge \liminf_{k \to \infty} H_t^k \cdot \Delta S_t - \epsilon |H_t^k|_p
\]
does not diverge to $-\infty$ if and only if $(|H_t^k|_p - a^k_t)_{k \in \N}$ is a.s.\ bounded.
\end{proof}

\begin{proposition}
\label{prop:1_step_closure}
    Assume $\NA_\epsilon'$.
    Then the closure of $K_t = \{ H_t \cdot \Delta S_t - \epsilon |H_t|_p \colon H \in \mathcal H_t \}$ w.r.t.\ convergence in probability is given by
    \begin{equation}
        \label{eq:lem_1_step_closure}
        \overline{K}_t = \{ H_t \cdot \Delta S_t - \epsilon |H_t|_p + G_t \cdot \Delta S_t \colon H_t \in \mathcal H_t, G_t \in F_{\epsilon,t}^\perp, H_t^\ast \cdot G_t = 0 \}.
    \end{equation}
    Furthermore, for any sequence $(Y^k)_{k \in \N}\subset \overline{K}_t$ that is bounded from below, there is $Y \in \overline{K}_t$ with
    \begin{equation*}
        \label{eq:lem_1_step_closure_liminf}
        Y \ge \liminf_{k \to \infty} Y^k.
    \end{equation*}
\end{proposition}

\begin{proof}
    We write $\hat K_t$ for the right-hand side of \eqref{eq:lem_1_step_closure}.

    \emph{Step 1}: We show the inclusion $\hat K_t \subseteq \overline{K}_t$.
    Pick any $H_t \in \mathcal H_t$ and $G_t \in  F^\perp_{\epsilon,t}$ and recall that by Lemma~\ref{lem:1_step_closure_simple} it suffices to show that
    \begin{equation}
        \label{eq:lem_1_step_closure_toshow1}
        H_t \cdot \Delta S_t - \epsilon |H_t|_p + \mathbbm 1_{\{ H_t = 0 \}} G_t \cdot \Delta S_t \in \overline{K}_t.
    \end{equation}

    Since $H_t \cdot \Delta S_t - \epsilon |H_t|_p \in K_t$, { proving}
    \eqref{eq:lem_1_step_closure_toshow1} boils down  to showing that $G_t \cdot \Delta S_t \in \overline{K}_t$.
    To this end, define
    $H^k_t :=  (k \bar H_t + G_t)$ for $k\in \N$ and compute on $\{\bar H_t\neq 0\}$ {(as in Remark \ref{lem:square_root_asympt})}:
    \[
        \lim_{k \to \infty} H^k_t \cdot \Delta S_t - \epsilon |H^k_t|_p = G_t \cdot \Delta S_t + \epsilon \lim_{k \to \infty} k - |k\bar H + G|_p = G_t \cdot \Delta S_t-\bar H_t^\ast\cdot G_t=G_t \cdot \Delta S_t.
    \]
    On the other hand, on $\{\bar H_t= 0\}$ we have $H^k_t \cdot \Delta S_t - \epsilon |H^k_t|_p=G_t \cdot \Delta S_t=0$. Hence overall we obtained $G_t \cdot \Delta S_t \in \overline{K}_t$.
    
    \emph{Step 2}: We prove the reverse inclusion $\hat K_t \supset \overline{K}_t$.
    We argue first that we can restrict w.l.o.g.\ to sequences $(H^k_t)_{k \in \N}$ in $\bar{\mathcal H}_t$.
    So, let $(H^k_t)_{k \in \N}$ be a sequence in $\mathcal H_t$ where $(H^k_t \cdot \Delta S_t - \epsilon |H^k_t|_p)_{k \in \N}$ is bounded from below.

    \emph{Case $p > 1$}: By Lemma \ref{lem:1_step_H_bar_complement} this sequence admits an $\F_{t-1}$-measurable subsequence such that
    \[
        (H^{\tau_k}_t \cdot \Delta S_t - \epsilon |H^{\tau_k}_t|_p)_{k \in \N} \mbox{ converges on }\{\bar H_t = 0\}.
    \]

    \emph{Case $p = 1$}: By Lemma \ref{lem:p=1.H_bounded} this sequence admits an $\F_{t-1}$-measurable subsequence such that
    \[
        (H^{\tau_k}_{t,i} \cdot \Delta S^i_t - \epsilon |H^{\tau_k}_{t,i}|)_{k \in \N} \mbox{ converges on }\{\bar H_{t,i} = 0 \},    
    \]
    for all $i \in \{1,\ldots, d\}$.

    In both cases, we may now assume w.l.o.g.\ that $(H_t^k)_{k \in \N}$ is a sequence in $\bar{\mathcal H}_t$, from where we conclude  the proof of Step 2 by combining  Corollary \ref{cor:1_step_boundedness_canonical_decomp} with Lemma \ref{lem:1_step_closure_simple}.
    
    \emph{Step 3}: For the last assertion, a diagonalization argument provides us with a sequence $(H_t^k)_{k \in \N}$  such that  $H^k_t \cdot \Delta S_t - \epsilon |H^k_t|_p\geq Y_k-a_k$, where $a_k\searrow 0$. In particular, $(H^k_t \cdot \Delta S_t - \epsilon |H^k_t|_p)_{k \in \N}$ is bounded from below, so {Corollary} \ref{cor:1_step_boundedness_canonical_decomp} with Lemma \ref{lem:1_step_closure_simple} yield, for an $\F_{t-1}$-measurable subsequence, a limit for $(H^{\tau_k}_t \cdot \Delta S_t - \epsilon |H^{\tau_k}_t|_p)_{k \in \N}$. This limit must necessarily live in $\hat K_t$ and be larger or equal than $\liminf_k Y_k$. This concludes the proof, as $\hat K_t=\overline{K}_t$.
\end{proof}

\subsection{\texorpdfstring{$\NA_\epsilon$}{}: the multi-step case}\label{sec:NA_multistep}

{We begin with the pending proof  of Proposition \ref{thm:epscritval}:

\begin{proof}[Proof of Proposition \ref{thm:epscritval}]
Let $\eps>  \eps(\P)$.
By localization, cf.\ Remark \ref{rem:localization}, it suffices to check the absence of $\eps$-arbitrage for single periods $[t-1,t)$ with $t \in \{1,\ldots, T\}$.
Recall that by Lemma \ref{lem:critical_value} we have that $E_{\epsilon,t} = \{ 0 \}$.
Consequently, if $(H^k_t)_{k \in \N}$ is a sequence in $\mathcal H_t$ such that $(H^k_t \cdot \Delta S_t - \eps |H_t^k|_p)_{k \in \N}$ is bounded from below, then we deduce from Lemma \ref{lem:1_step_H_bar_complement} that $(H^k_t)_{k \in \N}$ has to be bounded.
We conclude that there is an $\F_{t-1}$-measurable subsequence $(\tau_k)_{k \in \N}$ such that $\lim_{k \to \infty} H^{\tau_k}_t = H_t \in \mathcal H_t$ and
\[
    H_t \cdot \Delta S_t - \eps |H_t|_p = \lim_{k \to \infty} H^{\tau_k}_t \cdot \Delta S_t - \eps |H^{\tau_k}_t|_p \ge \liminf_{k \to \infty} H^k_t \cdot \Delta S_t - \eps |H^k_t|_p.
\]
In particular, there can be no $\eps$-arbitrage opportunities, as the above argument would produce a strict $\eps$-arbitrage.
\end{proof}

Next we provide the pending proof  of Theorem~\ref{thm:closure}:}
\begin{proof}[Proof of Theorem~\ref{thm:closure}, $p > 1$]
We show the claim by induction.
To this end, we write { $\tilde K_{T} := \overline{K}_{T}$} and recursively set, for $t = 1,\ldots, T-1$,
\begin{align*}
\tilde K_t :=& \left\{ X + Y \colon X \in \overline{K}_t, Y \in \tilde K_{t + 1}  \right\},
\\
\tilde K_t^0 :=& \left\{ (H \bullet S)_T - \epsilon \|H\|_p \colon H = (H_s)_{s = 1}^{T} \in \mathcal H, H_s = 0 \text{ for } s \le t-1 \right\},
\end{align*}
where {$\overline{K}_t$} is given in Proposition~\ref{prop:1_step_closure}. This same proposition, for $t=T$, constitutes the base of the induction. That is, the closure of $\tilde K_{T}^0$ w.r.t.\ convergence in probability is $\tilde K_{T}$ and for any sequence $(Z^k)_{k \in \N}$ in $\tilde K_{T}$ bounded from below there is $Z \in \tilde K_{T}$ with $Z \ge \liminf_{k \to \infty} Z^k$. Moreover, by $\NA_\epsilon'$, if $Z \in \tilde K_{T}$ is s.t. $Z \ge  0$, then actually $Z = 0$. Next, for the induction step, assume:
\begin{enumerate}[label = (\alph*)]
\item the closure of $\tilde K_{t + 1}^0$ w.r.t.\ convergence in probability is $\tilde K_{t + 1}$;
\item for any sequence $(Z^k)_{k \in \N}$ in $\tilde K_{t + 1}$ that is a.s.\ bounded from below there is $Z \in \tilde K_{t + 1}$ with $Z \ge \liminf_{k \to \infty} Z^k$;
\item for $Z \in \tilde K_{t + 1}$ we have
\begin{equation}
\label{eq:thm_closure.NAeps_inductive}
Z \ge 0 \implies Z = 0.
\end{equation}
\end{enumerate}

By the inductive assumption, \eqref{eq:thm_closure.NAeps_inductive} {and (b) have} as a consequence that, for any sequence $(Z^k)_{k \in \N}$ in $\tilde K_{t + 1}$, the following implication holds:
\begin{equation}
\label{eq:thm_closure.NAeps_liminf}
\liminf_{k\to\infty} Z^k \ge 0 \implies
\liminf_{k\to\infty} Z^k = 0.
\end{equation}
We now show the same claims for $\tilde K_t$.
Let $H_t^k \in \mathcal H_t$, $k \in \N$ with canonical decomposition $(a^k_t, G^k_t, \tilde G^k_t)$ on $\{\bar H_t \neq 0\}$, and let $(Y^k)_{k \in \N}$ be a sequence in 
$\tilde K_{t + 1}$ such that
\[
\inf_{k \in \N} H_t^k \cdot \Delta S_t - \epsilon |H_t^k|_p + Y^k > -\infty.
\]
Similar to Corollary~\ref{cor:1_step_boundedness_canonical_decomp}, we aim to first show that
\begin{equation}
\label{eq:thm_closure_step0}
\sup_{k \in \N} |G^k_t|_p < \infty,\quad
\inf_{k \in \N} a^k_t - |H^k_t|_p > -\infty,\quad\mbox{and }
\sup_{k \in \N} |H^k_t|_p < \infty \mbox{ on }\{\bar H_t\neq 0\},
\end{equation}
which will then allow us to pass to convergent subsequences by Lemma~\ref{lem:1_step_closure_simple}.

\emph{Step 1}. 
Here we consider the set $\{ \bar H_t = 0 \}$ and claim that $\sup_{k \in \N} |H^k_t|_p$ is finite on $\{\bar H_t = 0 \}$.
Write $A := \{ \limsup_{k \to \infty} |H_t^k|_p = \infty, \bar H_t = 0\}$ and define the normalized sequence $(\tilde H_t^k)_{k \in \N}$ as in \eqref{eq:def_normalized_sequence}.
By passing to an $\F_{t-1}$-measurable subsequence, we can assume that this sequence converges to $\tilde H_t \in F_{\epsilon,t}^\perp$ with $|\tilde H_t|_p = \mathbbm 1_A$, and
\begin{equation}
\label{eq:thm_closure_step1}
\tilde H_t \cdot \Delta S_t - \epsilon |\tilde H_t| + \liminf_{k \to \infty} \tilde Y^k \ge 0,
\end{equation}
where $\tilde Y^k := \mathbbm 1_{A} \frac{Y^k}{|H^k_t|_p \vee 1} \in \tilde K_{t + 1}$.
By \eqref{eq:thm_closure_step1} the sequence $(\tilde Y^k)_{k \in \N}$ is bounded from below.
Thus, the inductive assumption provides $\tilde Y \in \tilde K_{t + 1}$ with $\tilde Y \ge \liminf_{k \to \infty} \tilde Y^k$.
Since $\mathbbm 1_{ \{ \tilde H_t \cdot \Delta S_t - \epsilon |\tilde H_t|_p \le 0\} } \tilde Y \in \tilde K_{t + 1}$ we deduce from \eqref{eq:thm_closure.NAeps_inductive} that
\[
\mathbbm 1_{ \{ \tilde H_t \cdot \Delta S_t - \epsilon |\tilde H_t|_p \le 0\} } \tilde Y = 0,
\]
and in particular, by \eqref{eq:thm_closure_step1},
\[
\tilde H_t \cdot \Delta S_t - \epsilon |\tilde H_t|_p\ge 0,
\]
By the absence of strict $\epsilon$-arbitrage, this yields
\[
\tilde H_t \cdot \Delta S_t - \epsilon |\tilde H_t|_p = 0, \quad
\mbox{i.e.,}\quad
\tilde H_t \in E_{\epsilon,t} \cap F_{\epsilon,t}^\perp,
\]
which necessarily means $\tilde H_t \equiv 0$, from where we deduce that $\P[A] = 0$, as wanted.
    
\emph{Step 2}. 
We claim that $\sup_{k \in \N} |G^k_t|_p$ is a.s.\ finite and proceed by contradiction:
Let the event $A := \{ \sup_{k \in \N} |G^k_t|_p = \infty \}$ have positive probability.
Then passing to an $\F_{t-1}$-measurable, convergent subsequence leads to
\[
0 \le \mathbbm 1_A \liminf_{k \to \infty} \frac{G_t^k}{{|G_t^k|_p \vee 1}}  \cdot \Delta S_t + Y^k {\le} \bar G_t \cdot \Delta S_t + \mathbbm 1_A \liminf_{k \to \infty} Y^k,
\]
where $0 \neq \bar G_t \in F_{\epsilon,t}^\perp \cap E_{0,t}^\perp$.
On $B := \{ \bar G_t \cdot \Delta S_t < 0 \}$ we find
\[
\liminf_{k \to \infty} \mathbbm 1_{A \cap B} Y^k > 0,
\]
which yields a contradiction to \eqref{eq:thm_closure.NAeps_liminf}, thus implying $\P[B] = 0$.
Therefore, $\bar G_t \cdot \Delta S_t \ge 0$ almost surely, where we find by Definition \ref{def:NA'} \ref{it:def_NA'_NA} that $\bar G_t \cdot \Delta S_t = 0$ (that is, $\bar G_t \in E_{0,t}$).
Clearly this is not possible as $0 \neq \bar G_t \in E_{0,t}^\perp$, from where we deduce that $\mathbb P[A] = 0$.
    
\emph{Step 3}. 
By Steps 1 and 2 we have that $\sup_{k \in \N} |G^k_t|_p$ is a.s.\ finite which allows us to assume w.l.o.g.\ (by passing to an $\mathcal F_{t-1}$-measurable subsequence) that $(G_t^k)_{k \in \N}$ converges to some $G_t \in F_{\epsilon,t}^\perp \cap E_{0,t}^\perp$.
Then, we get
\[
-\infty < \liminf_{k \to \infty} G_t^k \cdot \Delta S_t + \epsilon (a_t^k - |H_t^k|_p) + Y^k = G_t \cdot \Delta S_t + \liminf_{k \to \infty} \epsilon (a_k - |H_t^k|_p) + Y^k,
\]
and claim that $A := \{ \liminf_{k \to \infty} a_k - |H_t^k|_p = -\infty \}$ is a null set.
Indeed, if $\mathbb P[A] > 0$, we could pass to an $\F_{t-1}$-measurable subsequence such that $(\mathbbm 1_A(a_t^k - |H_t^k|_p))_{k \in \N}$ diverges on $A$, and obtain
\[
\liminf_{k \to \infty} \mathbbm 1_A Y^k = +\infty \mbox{ where } \mathbbm 1_A Y^k \in \tilde K_{t + 1},
\]
which contradicts the inductive assumption \eqref{eq:thm_closure.NAeps_liminf}.

\emph{Step 4}.
By combining the Steps 1-3 we have shown \eqref{eq:thm_closure_step0} and can, thanks to Lemma~\ref{lem:1_step_closure_simple}, extract an $\mathcal F_{t-1}$-measurable subsequence such that
\[
\lim_{k \to \infty} H^k_t \cdot \Delta S_t - \epsilon |H^k_t|_p = H_t \cdot \Delta S_t - \epsilon |H_t|_p + G_t \cdot \Delta S_t \in \overline{K}_t,
\]
where $H_t \in \mathcal H_t$, $G_t \in F_{\epsilon,t}^\perp$ and $H_t^\ast \cdot G_t = 0$.
The next computation reveals that $(Y^k)_{k \in \N}$ is a.s.\ bounded:
\[
-\infty < \liminf_{k \to \infty} H^k_t \cdot \Delta S_t - \epsilon |H^k_t|_p + Y^k = H_t \cdot \Delta S_t - \epsilon |H_t|_p + G_t \cdot \Delta S_t + \liminf_{k \to \infty} Y^k.
\]
Therefore, we can make use of the inductive hypothesis in order to find $\tilde Y \in \tilde K_{t + 1}$ with $\tilde Y \ge \liminf_{k \to \infty} Y^k$ where equality holds in case that $(Y^k)_{k \in \N}$ is convergent.
In particular, we have
\begin{align}
\label{eq:thm_closure.dominated_liminf1}
H_t \cdot \Delta S_t - \epsilon |H_t|_p + G_t \cdot \Delta S_t + \tilde Y
&\ge
\liminf_{k\to \infty}
H^k_t \cdot \Delta S_t - \epsilon |H_t^k|_p + Y^k.
\end{align}
Denote the left-hand side of \eqref{eq:thm_closure.dominated_liminf1} by $Y$.
Therefore we have shown the existence of an element $Y \in \tilde K_t$ that dominates the right-hand side of \eqref{eq:thm_closure.dominated_liminf1}.
Moreover, in case that $(H^k_t \cdot \Delta S_t - \epsilon |H^k_t| + Y^k)_{k \in \N}$ was already convergent we even have equality in \eqref{eq:thm_closure.dominated_liminf1}
Hence, $\tilde K_t$ is closed under convergence in probability which concludes the inductive step.
    
Finally, we obtain by the same reasoning as for \eqref{eq:thm_closure.NAeps_inductive} the last assertion, that is, \eqref{eq:thm_closure.NAeps_liminf}.
\end{proof}

\begin{proof}[Proof of Theorem \ref{thm:closure}, $p = 1$]
For any $H \in \mathcal H$ we write { $\hat H$} for the element that is coordinatewise defined by
\[
    \hat H_{t,i} := \mathbbm 1_{ \{ \bar H_{t,i} \neq 0\} } H_{t,i},
\]
and denote its canonical decomposition by $(a_t,G_t,\tilde G_t)$.

We claim, for $t = 1,\ldots, T$, any sequence $(H^k)_{k \in \N}$ in {$\tilde K_t^0$}, cf.\ the proof of Theorem \ref{thm:closure} for $p > 1$, where $((H^k \bullet S)_T - \epsilon \|H^k\|_1)_{k \in \N}$ is bounded from below, satisfies
\begin{enumerate}[label = (\roman*)]
    \item \label{it:closure_p=1.1}$(H^k - \hat H^k)_{k \in \N}$ is bounded,
    \item \label{it:closure_p=1.2}both sequences, $(G^k_t)_{k \in \N}$ and $(\sum_{s = t}^T a_s^k - \|\hat H^k_t\|_1)_{k \in \N}$, are bounded.
\end{enumerate}
This claim is shown by induction.
The base case $t = T$ was already dealt with in the proof of Proposition \ref{prop:1_step_closure}.
So, next let $t = 1,\ldots,T-1$ and assume that have shown the claim for $t + 1$.

To see \ref{it:closure_p=1.1}, consider $(H^k_t - \hat H^k_t)_{k \in \N}$ and $A := \{ \limsup_{k \to \infty} |H^k_t - \hat H^k_t|_1 = \infty \}$.
By passing to an $\F_{t-1}$-measurable subsequence we can assume that the normalized sequence
\[
    \mathbbm 1_A \frac{H^k_t}{|H^k_t| \vee 1},  \quad k \in \N,
\]
converges to $\tilde H_t$ with $|\tilde H_t|_1 \in \{0,1\}$, $\{ |\tilde H_t|_1 = 1\} = A$, and
$\{ \tilde H_{t,i} \neq 0 \} \subseteq \{ \bar H_{t,i} = 0 \}$ for $i = 1,\ldots,d$.
Hence, 
\[
    \tilde H_t \cdot \Delta S_t - \epsilon |\tilde H_t| + \liminf_{k \to \infty} \frac{\mathbbm 1_A}{|H^k_t| \vee 1} \sum_{s = t + 1}^T H^k_s \cdot \Delta S_t - \epsilon |H^k_s|_1 \ge 0,
\]
and by the same arguments as in \emph{Step 1} in the proof of Theorem \ref{thm:closure} for $p > 1$, we find that $\tilde H_t \in E_{\epsilon,t}$ which yields $\P[A] = 0$.

To see \ref{it:closure_p=1.2}, we can argue as in \emph{Step 2} in the proof of Theorem \ref{thm:closure} for $p > 1$ and find that $(G^k_t)_{k \in \N}$ has to be bounded.
Then, analogously to \emph{Step 3} in the proof of Theorem \ref{thm:closure} for $p > 1$, we get also that $(a_k^t - |\hat H^k_t|_1)_{k \in \N}$ has to be bounded.

This concludes the inductive step thanks to the inductive assumption.

Now, let $(H^k)_{k \in \N}$ be a sequence $\mathcal H = \tilde K_1^0$ such that $((H^k \bullet S)_T - \epsilon \| H^k \|_1 )_{k \in \N}$ is bounded from below.
Then, thanks to \ref{it:closure_p=1.1} and \ref{it:closure_p=1.2}, we can pass to measurable subsequence such that there are, {by Proposition \ref{prop:1_step_closure} and Lemma \ref{lem:p=1.strict_G}, ${H_t} \in K_t$} with
\[
    \lim_{k \to \infty} (H^k \bullet S)_T - \epsilon \|H^k\|_1 = (H \bullet S)_T - \epsilon \|H\|_1.
\]
Therefore, $\overline{K} = K$ and in particular, for any sequence $(Y^k)_{k \in \N}$ in $K$ that is bounded from below there is $Y \in K$ with
\[
    Y \ge \liminf_{k \to \infty} Y^k.    
\]
This concludes the proof.
\end{proof}

With this at hand we can provide the pending proofs of Proposition \ref{cor:NAeps_equivalence} and Theorem~\ref{cor:C_closed}:

\begin{proof}[Proof of Proposition \ref{cor:NAeps_equivalence}]
{The forward implication was shown right after the statement of the proposition.}
To see the backward implication, assume $\NA_\epsilon'$.
By Theorem~\ref{thm:closure} we then have for any sequence $(H^k)_{k \in \N}$ in $\mathcal H$ that
\[
\liminf_{k \to \infty} (H^k \bullet S)_T - \epsilon \| H^k \|_p \ge 0 \implies
\liminf_{k \to \infty} (H^k \bullet S)_T - \epsilon \| H^k \|_p = 0.
\]
Hence, there is no $\epsilon$-arbitrage.
\end{proof}

\begin{proof}[Proof of Theorem~\ref{cor:C_closed}]
That $C$ is a convex cone is immediate. Let us then show the closure property. To this end, let $(X^k)_{k \in \N}$ be a sequence in $C$ that is bounded from below, 
and $(Y^k)_{k \in \N}$ in $\overline{K}$ with $X^k \le Y^k$.
Clearly, we have
\[
\liminf_{k \to \infty} Y^k \ge \liminf_{k \to \infty} X^k,
\]
and, by Theorem~\ref{thm:closure}, there is $Y \in \overline{K}$ with $Y \ge \liminf_{k \to \infty} Y^k$.
In case that $(X^k)_{k \in \N}$ converges to some $X$, then it is bounded from below and hence $Y \ge X$ and in particular $X \in C$.
Moreover, if $\liminf_{k \to \infty} X^k \ge 0$ and therefore $Y \ge 0$, then by $\NA_\epsilon$ we have that $Y = 0$, which yields that $C \cap L^0_+(\Omega,\mathcal F_T, \mathbb P; \mathbb R) = \{ 0 \}$.
\end{proof}

\subsection{The FTAP, duality, and the price range}\label{sect:proof:ftap}

\begin{proof}[Proof of Theorem~\ref{thm:ftap}]
The fact that the existence of an $\epsilon$-martingale measure implies absence of $\epsilon$-arbitrage is a consequence of Remark \ref{rem:eps_mart_arb}.

In order to show the other implication, let us assume that the market satisfies $\NA_\epsilon$.
If {$S \notin L^1(\Omega,\F_T,\P)$ we pass} to an equivalent probability measure $\P' \sim \P$, for example $\frac{d\P'}{d\P} \propto \exp(- \| S \|_p)$, such that $S \in L^1(\Omega,\F_T,\P')$. {In this case}, we redefine $\P:=\P'$.
Recall that by Theorem~\ref{cor:C_closed} the set 
\[
C = \{ X \in L^0(\Omega, \F_T,\P;\R) \colon \exists Y \in \overline{K} \text{ with } X \le Y \}
\]
is closed w.r.t.\ convergence in probability and
satisfies $C \cap L^0_+(\Omega,\F_T,\P) = \{ 0 \}$. Consider the convex cone given by
\[
C^1 := 
\left\{
X \in L^1(\Omega, \F_T, \P) \colon
\exists Y \in C \text{ with }  X \le Y
\right\} \subseteq C.
\]
Since $C$ is closed w.r.t.\ convergence in probability, we have that $C \cap L^1(\Omega,\F_T,\P)$ is closed in $L^1$ and coincides with $C^1$ by definition.
In particular, we have
\[
\{ 0 \} \subseteq C^1 \cap L^1_+(\Omega, \F_T, \P) \subseteq C \cap L^0_+(\Omega, \F_T, \P) = \{ 0 \}.
\]
Hence we can apply the Kreps-Yan theorem (see \cite[Theorem 1.62]{FoSc16})  to the closed convex cone $C^1$ in order to obtain $\Q \sim \P$ with $\frac{d\Q}{d\P} \in L^\infty_+(\Omega,\F_T,\P')$ such that $\E_\Q[X] \le 0$ for all $X \in C^1$.
Using the latter property we get
\[
\E_\Q[ \mathbbm 1_A \left( H_t \cdot \Delta S_t - \epsilon \right) ] \le 0 \quad \forall A \in \F_{t-1} \mbox{ and } H_t \in \mathcal H_t \mbox{ with } |H_t|_p = 1.
\]
By the defining property of the conditional expectation, we find that
\[
H_t \cdot \E_\Q[ \Delta S_t | \F_{t-1} ] \le \epsilon |H_t|_p  \quad \forall H_t  \in \mathcal H_t.
\]
Finally, as the $q$-norm is the dual of the $p$-norm, we have
$\left| \E_\Q[ \Delta S_t | \F_{t-1} ] \right|_q \le \epsilon$, which implies that $\Q$ is an $\eps$-martingale and concludes the proof.
\end{proof}

\begin{proof}[Proof of Theorem~\ref{thm:duality}]
As usual, weak duality is trivially satisfied since, for any $\Q \in \mathcal M_\epsilon(\P)$ and admissible super-replicating strategy $x \in \R$ and $(H,G) \in \mathcal H \times F_\epsilon^\perp$ (in the sense of those used in the r.h.s. in \eqref{eq:dual}), we have
\[
\E_\Q[\Psi] \le x+\E_\Q[(H\bullet S)_T + (G \bullet S)_T - \epsilon \| H \|_p ]\le x,
\]
as follows from Lemma~\ref{lem:eps_mart_means}. To prove the converse, recall that $C^1$ is closed in $L^1(\Omega,\F_T,\P)$. The bipolar theorem, see for example \cite[Theorem 5.103]{Ch06},  tells us that
\begin{equation}
\label{eq:duality_bipolar_theorem}
C^1 = \left\{ Y \in L^1(\Omega, \F_T, \P) \colon \sup_{\Q \in \mathcal M_\epsilon(\P)} \E_\Q[Y] \le 0  \right\}.
\end{equation}
Denoting the value of the left-hand side of \eqref{eq:dual} by $l \in \R$, we have
\[
\sup_{\Q \in \mathcal M_\epsilon} \E_\Q[\Psi - l] = 0,
\]
which yields that $\Psi - l \in C^1$.
Therefore, by \eqref{eq:duality_bipolar_theorem} and the definition of $C^1$, there exists an admissible pair $(G,H)$ with $l + (H \bullet S)_T + (G \bullet S)_T \ge \Psi + \epsilon \|H\|_p$.
\end{proof}

\begin{proof}[Proof of Theorem~\ref{thm:price_range}]
    \emph{Part 1}:
To prove the first assertion, let $\Q \in \mathcal M_\epsilon(\P)$ and $\psi \in [\E_\Q \Psi - \epsilon, \E_\Q + \epsilon]$.
Note that $\E_\Q[\Psi - \psi] \in [-\epsilon,\epsilon]$.
Consider $(H^k)_{k \in \N}$ in $\mathcal H$ and $(a^k)_{k \in \N}$ in $\R$ with
\[
\liminf_{k \to \infty} (H^k \bullet S)_T + a^k(\Psi - \psi) - \epsilon \left( \|H^k\|_p + |a^k| \right)\ge 0.
\]
We need to show that the above is actually an equality.
    
\emph{Case 1}: Assume that $(a^k)_{k \in \N}$ is bounded. W.l.o.g., by passing to a subsequence, we may assume that the sequence converges to $a \in \R$.
Then $((H^k \bullet S)_T - \epsilon \|H^k\|_p)_{k \in \N}$ is bounded from below and by Theorem~\ref{thm:closure} there exists $Y \in \overline{K}$ such that
\begin{equation}\label{eq:yalim}
Y + a (\Psi - \psi) - \epsilon |a| \ge \liminf_{k \to \infty} (H^k \bullet S)_T + a^k(\Psi - \psi) - \epsilon \left( \|H^k\|_p + |a^k| \right) \ge 0.
\end{equation}
Since $\Psi \in L^1(\P)$, we deduce for all $n \in \N$ that $Y \wedge n \in L^1(\P)$ with $Y \wedge n \le Y \in \overline{K}$.
Moreover, because $\Q \in \mathcal M_\epsilon(\P)$, we obtain $\E_\Q[Y \wedge n] \le 0$.
Finally, note that $Y \wedge n \nearrow Y$ for $n \to \infty$, which yields by monotone convergence
\[
\E_\Q[Y] = \lim_{n \to \infty} \E_\Q[Y \wedge n] \le 0.
\]
Hence, $\E_\Q[ Y + a(\Psi - \psi) - \epsilon |a|] \le 0$, and from \eqref{eq:yalim} we get $Y + a (\Psi - \psi) - \epsilon |a| = 0$, as wanted.

\emph{Case 2}: Assume that $(a^k)_{k \in \N}$ is unbounded.
W.l.o.g.\ let the sequence diverge to either $+\infty$ or $-\infty$ and call the limit $a$. 
By normalizing, we find
\[
\liminf_{k \to \infty} \frac{1}{|a^k|} \left( (H^k \bullet S)_T - \epsilon \|H^k\|_p  \right) \ge -\sgn(a) \left(\Psi - \psi \right) + \epsilon.
\]
Therefore, by Theorem~\ref{thm:closure} there exists $Y \in \overline{K}$ with
\[
Y + \sgn(a) \left(\Psi - \psi \right) - \epsilon \ge 0.
\]
As in Case 1, we find that $\E_\Q[Y] = 0$, $Y + \sgn(a) \left(\Psi - \psi \right) - \epsilon = 0$, and in particular $Y \in L^1(\P)$.
Since $Y \in \overline{K}$, there exist $H \in \mathcal H$, $G \in F_{\epsilon}^\perp$ with $H_t^\ast \cdot G_t = 0$ when $p > 1$ and $G_t = 0$ when $p = 1$ { (since then $K = \overline{K}$ by Theorem \ref{thm:K=barK})}, such that $Y = ((G+H)\bullet S)_T - \epsilon \|H\|_p$.
    
\emph{Claim}: For $t = 1,\ldots, T$, we have
\begin{equation}
\label{eq:fair_price_range.claim}
\E_\Q[(G_t + H_t) \cdot \Delta S_t - \epsilon |H_t|_p | \F_{t-1}] = 0.
\end{equation}
Since $\Q$ is an $\epsilon$-martingale measure, { $\E_\Q[G_t \cdot \Delta S_t | \F_{t-1}] = 0$ by Lemma \ref{lem:eps_mart_means} while $\E_\Q[H_t \cdot \Delta S_t - \epsilon |H_t|_p| \F_{t-1}] \le 0 $}, thus the left-hand side of \eqref{eq:fair_price_range.claim} has to be non-positive.
For the reverse inequality, consider for $N \in \N$ the localizing sequence of $(\F_{t - 1})_{t
=1}^T$-stopping times
\[
\tau_N := \max \left\{ t \in \{0,\ldots,T\} \colon \sum_{s = 1}^t |G_t|_p + |H_t|_p  \leq  N \right\},
\]
and for any random time $\tau$, set
\[
Y_\tau := \sum_{t = 1}^{\tau} (G_t + H_t) \cdot \Delta S_t - \epsilon  |H_t|_p.
\]
Clearly, $\tau_N \nearrow T$ for $N \to \infty$ and $Y_{\tau_N \wedge t} \in \overline{K} \cap L^1(\Omega,\F_T,\P)$ for $t \in \{0,\ldots,T\}$ and $N \in \N$.
We find that $Y - Y_{\tau_N \wedge t}$ is also contained in $\overline{K} \cap L^1(\Omega,\F_T,\P)$.
Therefore, since $\E_\Q[Y] = 0$, we deduce that { $\E_\Q[Y - Y_{\tau_N \wedge t}] = 0$ and $\E_\Q[Y_{\tau_N \wedge t}] = 0$.}
By monotone convergence, we get
\begin{align*}
\E_\Q \left[ \E_\Q[ (G_t + H_t) \cdot \Delta S_t - \epsilon |H_t|_p | \F_{t-1} ] \right] 
&= \lim_{N \to \infty}
\E_\Q \left[ \mathbbm 1_{ \{ \tau_N > t - 1\} }
\E_\Q [ (G_t + H_t) \cdot \Delta S_t - \epsilon |H_t|_p | \F_{t-1} ] \right]
\\
&= \lim_{N \to \infty}
\E_\Q \left[ Y_{\tau_N \wedge t} - Y_{\tau_N \wedge (t - 1)} \right] = 0,
\end{align*}
which yields \eqref{eq:fair_price_range.claim}, since $\E_\Q[(G_t + H_t) \cdot \Delta S_t - \epsilon |H_t|_p | \F_{t-1}] \le 0$.
    
As a consequence of \eqref{eq:fair_price_range.claim} and Lemma \ref{lem:eps_mart_means}, we obtain
\[
H_t \cdot \E_\Q[\Delta S_t | \F_{t-1}] = \epsilon  |H_t|_p,
\]
and again by Lemma \ref{lem:eps_mart_means}, $\E_\Q[\Delta S_t | \F_{t-1}] \in F_{\epsilon,t}$.
The process $\tilde S$ whose increments are given by
\[
\Delta \tilde S_t := \Delta S_t - \E_\Q[\Delta S_t | \F_{t-1}] 
\]
is a martingale under $\Q$.
We compute
\begin{align*}
((G + H) \bullet \tilde S)_T &= ((G + H) \bullet S)_T - 
\sum_{t = 1}^T (G_t + H_t) \cdot \E_{\Q}[\Delta S_t | \F_{t-1}] \\
&= ((G + H) \bullet S)_T - \sum_{t = 1}^T H_t \cdot \E_{\Q}[\Delta S_t | \F_{t-1}] \\
&= ((G + H) \bullet S)_T - \epsilon \|H \|_p = -\left( \sgn(a) \left( \Psi - \psi \right) - \epsilon \right),
\end{align*}
and, due to H\"older's inequality,
\begin{align*}
(H^k \bullet \tilde S)_T &= (H^k \bullet S)_T - \sum_{t = 1}^{T} H^k_t \cdot \E_\Q[\Delta S_t | \F_t] \ge (H^k \bullet S)_T - \epsilon \|H^k\|_p.
\end{align*}
Hence,
\begin{align*}
\liminf_{k \to \infty} ((H^k- |a^k| (G + H))&\bullet \tilde S)_T \\
&\ge \liminf_{k \to \infty} (H^k \bullet S)_T - \epsilon \|H^k\|_p + a^k ( \Psi - \psi)- \epsilon |a^k|
\ge 0,
\end{align*}
which can only be the case when the left-hand side vanishes, since $\tilde S$ is a martingale under $\Q$.
{This is readily derived from results of the classical theory, see e.g.\ \cite[Proposition 6.9.1]{DeSc06}.}

\emph{Part 2}: {It is enough to
show the second assertion when $p>1$, as in the case $d=1$ the $p$-norm on $\mathbb R$ is just the absolute value.}
Assume that introducing $\Psi$ at price $\psi$ in the market does not introduce $\epsilon$-arbitrage.
We claim that there exists $\Q \in \mathcal M_\epsilon(\P)$ with $\psi \in [\E_\Q[\Psi] - \epsilon, \E_\Q[\Psi] + \epsilon]$.
Consider sequences $(H^k)_{k \in \N}$ of trading strategies and $(a^k)_{k \in \N}$ in $\R$ with
\begin{equation}
\label{eq:fair_price_range.part2_sequences}
\liminf_{k \to \infty} (H^k \bullet S)_T - \epsilon \|H^k\|_p + a^k \left( \Psi - \psi \right) - \epsilon |a^k| \ge Y \in L^0(\Omega,\F_T,\P).
\end{equation}
As before, we distinguish two cases:
    
\emph{Case 1}:
Suppose that, for all sequences satisfying  \eqref{eq:fair_price_range.part2_sequences},
$(a^k)_{k \in \N}$ is bounded.
We claim that
\[
C_{\Psi,\psi} := \left\{Y \in L^0(\Omega,\F_T,\P) \colon \exists X \in \overline{K} \mbox{ and } a \in \R \mbox{ such that } Y \le  X + a(\Psi - \psi) - \epsilon |a| \right\}
\]
is closed w.r.t.\ convergence in probability, from which we derive, by the same arguments as in the proof of Theorem~\ref{thm:ftap}, that there exists $\Q \in \mathcal M_\epsilon(\P)$ with $-\epsilon \le \E_\Q[\Psi - \psi] \le \epsilon$.
As in Case~1 of Part~1, we find $\tilde Y \in \overline{K}$ {and $a \in \R$} such that
\[
\tilde Y + a(\Psi - \psi) - \epsilon |a| \ge Y,
\]
thus, $Y \in C_{\Psi,\psi}$ {and we conclude that $ C_{\Psi,\psi}$ is closed w.r.t.\ convergence in probability}.

\emph{Case 2}: 
Suppose there is a sequence satisfying  \eqref{eq:fair_price_range.part2_sequences} with $(a^k)_{k \in \N}$ bounded. 
We proceed just as in the beginning of Case~2 of Part~1, and get by normalization
\[
\tilde Y \ge \liminf_{k \to \infty} \frac1{|a^k|} \left( (H^k \bullet S)_T - \epsilon \|H^k\|_p \right) \ge -\sgn(a)\left( \Psi - \psi \right) + \epsilon,
\]
for some $\tilde Y \in \overline{K}$.
Since $\psi$ is an $\epsilon$-fair price for $\Psi$, the inequalities in the above display have to be equalities.
Thus, by representing $\tilde Y$ in terms of a pair $(H,G) \in \mathcal H \times F_{\epsilon}^\perp$ with $H^\ast_t \cdot G_t = 0$, we have
\[
\sgn(a) \left( \Psi - \psi \right) - \epsilon = - \tilde Y = - ((G + H) \bullet S)_T + \epsilon \|H \|_p.
\]
Consider the auxiliary process $\tilde S$ whose increments are given by 
\[
\Delta \tilde S_t := \mathbbm 1_{ \{ H_t = 0 \} } \Delta S_t.
\]
By assumption, $\tilde S$ admits no $\epsilon$-arbitrage, whence, by Remark \ref{rem:Doob_decomposition}, there is a predictable process $\zeta$ with $|\Delta \zeta_t|_q \le \epsilon$ such that $\tilde S - \zeta$ has no arbitrage in the classical sense.
Next, we define another auxiliary process $\hat S$ whose increments are given by
\begin{align*}
\Delta \hat S_t := \Delta S_t - \Delta \zeta_t - \epsilon H_t^\ast.
\end{align*}    
We claim that $\hat S$ has no arbitrage either.
Indeed, by localization (cf.\ Remark \ref{rem:localization}) it suffices to consider $\tilde H_t \in \mathcal H_t$, and compute
\begin{align*}
\tilde H_t \cdot \Delta \hat S_t &= \mathbbm 1_{\{ H_t = 0 \} } \tilde H_t \cdot (\Delta S_t - \Delta \zeta_t) + \mathbbm 1_{\{H_t \neq 0 \}} \tilde H_t \cdot\left( \Delta S_t - \epsilon H_t^\ast \right).
\end{align*}
In order to deduce the absence of arbitrage, it remains to show the implication
\[
\mathbbm 1_{\{ H_t \neq 0 \} } \tilde H_t \cdot (\Delta S_t - \epsilon H_t^\ast)
\ge 0 \implies
\mathbbm 1_{\{ H_t \neq 0 \} }\tilde H_t \cdot(\Delta S_t - \epsilon H_t^\ast)= 0.
\]
To this end, we assume that $\{ \tilde H_t \neq 0\} \subseteq \{H_t \neq 0\}$.
We decompose $\tilde H_t = \hat H_t + \alpha_t H_t$ such that $\hat H_t \in \mathcal H_t$, $\hat H_t \cdot H_t^\ast = 0$ and $\alpha_t \in L^0(\Omega,\F_{t-1},\P)$, which is possible by the same arguments as in Lemma \ref{lem:H_bar}.
For $k \in \N$, we define $\alpha^k_t := \max \{ \alpha_t + k, 0 \}$.
{ By the same reasoning as in Remark \ref{lem:square_root_asympt} note that}

\[
    \lim_{k \to \infty} |\hat H_t + \alpha_t^k H_t|_p - k|H_t|_p = \alpha_t|H_t|_p,
\]
whence we find
\begin{align*}
\tilde H_t \cdot \Delta \hat S_t 
&= \tilde H_t \cdot \Delta S_t - \epsilon \tilde H_t \cdot H_t^\ast
= (\hat H_t + \alpha_t H_t) \cdot \Delta S_t - \epsilon \alpha_t |H_t|_p
\\
&= \lim_{k \to \infty} (\hat H_t + \alpha^k_t H_t ) \cdot \Delta S_t - k H_t \cdot \Delta S_t
- \epsilon \left( |\hat H_t + \alpha^k_t H_t|_p - k|H_t|_p\right)
\\
&= \lim_{k \to \infty} (\hat H_t + \alpha_t^k H_t) \cdot \Delta S_t - \epsilon |\hat H_t + \alpha_t^kH_t|_p - k (H_t \cdot \Delta S_t - \epsilon |H_t|_p ).
\end{align*}

If $\tilde H_t \cdot \Delta \hat S_t \ge 0$, then the right-hand side of the display above has to vanish, since $\psi$ is an $\epsilon$-fair price for $\Psi$.
To see this, define $H_s^k = k H_s$ for $s \neq t$, $G^k := k G$, and $H_t^k := \hat H_t + \alpha_t^k H_t$ and note that
\begin{align*}
\tilde H_t \cdot \Delta \hat S_t &= \lim_{k \to \infty} ((G^k + H^k) \bullet S)_T - \epsilon \|H^k\|_p - k ((G + H) \bullet S)_T + \epsilon k \|H\|_p.
\end{align*}
We conclude that $\hat S$ admits no arbitrage.
    
Finally, we can apply the classical fundamental theorem of asset pricing \cite[Theorem 6.1.1]{DeSc06} to find a measure $\Q \sim \P$, $\frac{d\Q}{d\P} \in L^\infty(\P)$ such that $\hat S$ is a martingale under $\Q$.
Therefore, $\E_{\Q}[\tilde Y] = \E_\Q[((G+H)\bullet \hat S)_T] = 0$.
Denote by $\hat \zeta$ the predictable process with increments $\Delta \hat \zeta_t := \Delta \zeta_t + \epsilon H_t^\ast$.
Note that $|\hat \zeta_t|_q \leq \epsilon$, therefore $S = \hat S - \hat \zeta$ is an $\epsilon$-martingale under $\Q$, by Remark \ref{rem:Doob_decomposition}.
Finally, observe that
\[
\sgn(a) \E_\Q[\Psi - \psi] + \epsilon = \E_\Q[ \tilde Y] = 0,
\]
thus, $\psi \in [\E_\Q [\Psi] - \epsilon, \E_\Q [\Psi] + \epsilon]$.
\end{proof}

\bibliography{joint_biblio(2)}{}

\begin{thebibliography}{10}

\bibitem{AcBePeSc13}
Beatrice Acciaio, Mathias Beiglb{\"o}ck, Friedrich Penkner, and Walter
  Schachermayer.
\newblock A model-free version of the fundamental theorem of asset pricing and
  the super-replication theorem.
\newblock {\em Math. Finance}, 26(2):233--251, 2016.

\bibitem{AcBaJi21}
Beatrice Acciaio, Julio~Backhoff Veraguas, and Junchao Jia.
\newblock Cournot--{N}ash equilibrium and optimal transport in a dynamic
  setting.
\newblock {\em SIAM Journal on Control and Optimization}, 59(3):2273--2300,
  2021.

\bibitem{Ch06}
Charalambos~D Aliprantis and Kim~C Border.
\newblock {\em Infinite dimensional analysis: A hitchhiker's guide}.
\newblock Springer, 3rd edition, 2006.

\bibitem{BaBaBeEd19a}
Julio Backhoff, Daniel Bartl, Mathias Beiglb\"{o}ck, and Manu Eder.
\newblock Adapted {W}asserstein distances and stability in mathematical
  finance.
\newblock {\em Finance Stoch.}, 24(3):601--632, 2020.

\bibitem{BaBaBeWi20}
Julio Backhoff, Daniel Bartl, Mathias Beiglb{\"o}ck, and Johannes Wiesel.
\newblock Estimating processes in adapted {W}asserstein distance.
\newblock {\em The Annals of Applied Probability}, 32(1):529--550, 2022.

\bibitem{BaBeLiZa17}
Julio Backhoff, Mathias Beiglb\"{o}ck, Yiqing Lin, and Anastasiia Zalashko.
\newblock Causal transport in discrete time and applications.
\newblock {\em SIAM J. Optim.}, 27(4):2528--2562, 2017.

\bibitem{BaBoJe17}
E~Barron, M~Bocea, and R~Jensen.
\newblock Duality for the $l^{\{\infty\}}$ optimal transport problem.
\newblock {\em Transactions of the American Mathematical Society},
  369(5):3289--3323, 2017.

\bibitem{BaBePa21}
Daniel Bartl, Mathias Beiglb\"ock, and Gudmund Pammer.
\newblock The {W}asserstein space of stochastic processes.
\newblock {\em ArXiv e-prints}, 2021.

\bibitem{bayraktar2016fundamental}
Erhan Bayraktar and Yuchong Zhang.
\newblock Fundamental theorem of asset pricing under transaction costs and
  model uncertainty.
\newblock {\em Mathematics of Operations Research}, 41(3):1039--1054, 2016.

\bibitem{GuMu13}
Fred~Espen Benth, Dan Crisan, Paolo Guasoni, Konstantinos Manolarakis, Johannes
  Muhle-Karbe, Colm Nee, Philip Protter, Paolo Guasoni, and Johannes
  Muhle-Karbe.
\newblock Portfolio choice with transaction costs: a user’s guide.
\newblock {\em Paris-Princeton Lectures on Mathematical Finance 2013: Editors:
  Vicky Henderson, Ronnie Sircar}, pages 169--201, 2013.

\bibitem{biagini2017robust}
Sara Biagini, Bruno Bouchard, Constantinos Kardaras, and Marcel Nutz.
\newblock Robust fundamental theorem for continuous processes.
\newblock {\em Mathematical Finance}, 27(4):963--987, 2017.

\bibitem{ChdPJu08}
Thierry Champion, Luigi De~Pascale, and Petri Juutinen.
\newblock The {$\infty$}-{W}asserstein distance: Local solutions and existence
  of optimal transport maps.
\newblock {\em SIAM Journal on Mathematical Analysis}, 40(1):1--20, 2008.

\bibitem{ChKuTa17}
Patrick Cheridito, Michael Kupper, and Ludovic Tangpi.
\newblock Duality formulas for robust pricing and hedging in discrete time.
\newblock {\em SIAM Journal on Financial Mathematics}, 8(1):738--765, 2017.

\bibitem{DaMoVi90}
Robert~C Dalang, Andrew Morton, and Walter Willinger.
\newblock Equivalent martingale measures and no-arbitrage in stochastic
  securities market models.
\newblock {\em Stochastics: An International Journal of Probability and
  Stochastic Processes}, 29(2):185--201, 1990.

\bibitem{DeSc06}
Freddy Delbaen and Walter Schachermayer.
\newblock {\em The mathematics of arbitrage}.
\newblock Springer Finance. Springer-Verlag, Berlin, 2006.

\bibitem{DoSo13}
Y.~Dolinsky and M.~H. Soner.
\newblock Robust hedging with proportional transaction costs.
\newblock {\em Finance Stoch.}, 18(2):327--347, 2014.

\bibitem{EcPa22}
Stephand Eckstein and Gudmund Pammer.
\newblock Computational methods for adapted optimal transport.
\newblock {\em arXiv:2203.05005}, 2022.

\bibitem{FoSc16}
Hans F{\"o}llmer and Alexander Schied.
\newblock Stochastic finance.
\newblock In {\em Stochastic Finance}. de Gruyter, 2016.

\bibitem{GeGu19}
Stefan Gerhold and I~Cetin G{\"u}l{\"u}m.
\newblock Peacocks nearby: Approximating sequences of measures.
\newblock {\em Stochastic Processes and their Applications}, 129(7):2406--2436,
  2019.

\bibitem{GeGu20}
Stefan Gerhold and Ismail~Cetin G{\"u}l{\"u}m.
\newblock Consistency of option prices under bid--ask spreads.
\newblock {\em Mathematical finance}, 30(2):377--402, 2020.

\bibitem{GuRaSc10}
Paolo Guasoni, Mikl{\'o}s R{\'a}sonyi, and Walter Schachermayer.
\newblock The fundamental theorem of asset pricing for continuous processes
  under small transaction costs.
\newblock {\em Annals of Finance}, 6(2):157--191, 2010.

\bibitem{GuOb19}
Gaoyue Guo and Jan Ob\l\'{o}j.
\newblock Computational methods for martingale optimal transport problems.
\newblock {\em Ann. Appl. Probab.}, 29(6):3311--3347, 2019.

\bibitem{HaKr79}
J.~M. Harrison and D.~M. Kreps.
\newblock Martingales and arbitrage in multiperiod securities markets.
\newblock {\em Journal of Economic Theory}, 20(3):381--408, 1979.

\bibitem{HaPl81}
J.~M. Harrison and S.~R. Pliska.
\newblock Martingales and stochastic integrals in the theory of continuous
  trading.
\newblock {\em Stochastic Processes and their Applications}, 11(3):215--260,
  1981.

\bibitem{herdegen2017no}
Martin Herdegen.
\newblock No-arbitrage in a num{\'e}raire-independent modeling framework.
\newblock {\em Mathematical Finance}, 27(2):568--603, 2017.

\bibitem{Ka09}
Yuri Kabanov.
\newblock {\em Markets with Transaction Costs Mathematical Theory}.
\newblock Springer, 2009.

\bibitem{KaMh09}
Yuri Kabanov and Mher Safarian.
\newblock {\em Markets with transaction costs}.
\newblock Springer Finance. Springer-Verlag, Berlin, 2009.
\newblock Mathematical theory.

\bibitem{KaFe09}
Ioannis Karatzas and Robert Fernholz.
\newblock Stochastic portfolio theory: an overview.
\newblock {\em Handbook of numerical analysis}, 15:89--167, 2009.

\bibitem{Kr81}
David~M Kreps.
\newblock Arbitrage and equilibrium in economies with infinitely many
  commodities.
\newblock {\em Journal of Mathematical Economics}, 8(1):15--35, 1981.

\bibitem{Mu06}
Julian Musielak.
\newblock {\em Orlicz spaces and modular spaces}, volume 1034.
\newblock Springer, 2006.

\bibitem{ObWi18}
Jan Ob{\l}{\'o}j and Johannes Wiesel.
\newblock Robust estimation of superhedging prices.
\newblock {\em The Annals of Statistics}, 49(1):508--530, 2021.

\bibitem{PePe18}
Teemu Pennanen and Ari-Pekka Perkki{\"o}.
\newblock Convex duality in optimal investment and contingent claim valuation
  in illiquid markets.
\newblock {\em Finance and Stochastics}, 22:733--771, 2018.

\bibitem{PfPi16}
Georg~Ch. Pflug and Alois Pichler.
\newblock From empirical observations to tree models for stochastic
  optimization: convergence properties.
\newblock {\em SIAM J. Optim.}, 26(3):1715--1740, 2016.

\bibitem{Ru13}
Johannes Ruf.
\newblock Hedging under arbitrage.
\newblock {\em Mathematical Finance: An International Journal of Mathematics,
  Statistics and Financial Economics}, 23(2):297--317, 2013.

\bibitem{Sch04}
Walter Schachermayer.
\newblock The fundamental theorem of asset pricing under proportional
  transaction costs in finite discrete time.
\newblock {\em Mathematical Finance: An International Journal of Mathematics,
  Statistics and Financial Economics}, 14(1):19--48, 2004.

\bibitem{Sch17}
Walter Schachermayer.
\newblock {\em Asymptotic theory of transaction costs}.
\newblock European Mathematical Society, 2017.

\bibitem{Vi03}
C.~Villani.
\newblock {\em Topics in optimal transportation}, volume~58 of {\em Graduate
  Studies in Mathematics}.
\newblock American Mathematical Society, Providence, RI, 2003.

\bibitem{xu2020cot}
Tianlin Xu, Li~Kevin Wenliang, Michael Munn, and Beatrice Acciaio.
\newblock {COT-GAN: Generating sequential data via causal optimal transport}.
\newblock {\em Advances in Neural Information Processing Systems},
  33:8798--8809, 2020.

\end{thebibliography}
\bibliographystyle{plain}

\end{document}